\documentclass[11pt]{article}

\usepackage{sn-preamble} 
\usepackage{tikz}
\usepackage{verbatim}
\usepackage{cancel}
\usepackage{microtype}

\usetikzlibrary{calc,patterns,angles,quotes}

\usepackage{titlesec}
\makeatletter
\titleformat{\part}[block]
  {\Large\bfseries}
  {\partname~\thepart:}
  {10pt}
  {}
\makeatother

\newcommand{\blue}[1]{{{\color{blue}#1}}}

\newcommand{\ignore}[1]{}

\newcommand{\GSA}{\mathbf{GSA}}
\newcommand{\dG}{\mathrm{dist}_{\mathrm{G}}}
\newcommand{\dH}{\mathrm{dist}_{\mathrm{H}}}

\newcommand{\FC}{\mathrm{FC}}
\newcommand{\Facet}{\mathrm{Facet}}
\newcommand{\spann}{\mathrm{span}}

\newcommand{\newinf}{\mathop{\mathrm{inf}\vphantom{\mathrm{sup}}}} 
\renewcommand{\S}{\mathbb{S}}
\newcommand{\B}{\mathbb{B}}

\newcommand{\zoom}[3]{#1_{#2 \mid #3}}
\newcommand{\bup}{{\boldsymbol{\upsilon}}}

\newcommand{\Naz}{\mathrm{Naz}}
\newcommand{\g}[1]{\bg^{(#1)}}
\newcommand{\din}{d_{\mathrm{in}}}
\newcommand{\dout}{d_{\mathrm{out}}}

\newcommand{\HV}{\mathbf{HyperVar}}   \newcommand{\hypvar}{\HV}
\newcommand{\paramset}{\frac{1}{\sqrt{\log n}} \leq \eps}

\newcommand{\bucket}{\mathrm{buck}}
\newcommand{\Cramer}{Cram\'er}
\renewcommand{\i}[1]{\bi^{(#1)}}

\newcommand{\newN}{m}

\newcommand{\Dp}{\calD}
\newcommand{\VC}{\text{VC}}

\usepackage{etoolbox}
\makeatletter
\AtBeginDocument{%
  \patchcmd{\use@@tikzlibrary}{\global}{}{}{}%
}

\title{Gaussian Approximation of Convex Sets\\
 by Intersections of Halfspaces
\vspace{0.5em}}

\author{
Anindya De\thanks{University of Pennsylvania. Email: \href{mailto:anindyad@seas.upenn.edu}{\texttt{anindyad@seas.upenn.edu}}.}
\and 
Shivam Nadimpalli\thanks{Columbia University.  Email: \href{mailto:sn2855@columbia.edu}{\texttt{sn2855@columbia.edu}}.}
\and 
Rocco A. Servedio\thanks{Columbia University.  Email: \href{mailto:ras2105@columbia.edu}{\texttt{ras2105@columbia.edu}}.}
\vspace{0.5em}
}

\date{\small\today}

\begin{document}

\pagenumbering{gobble}
\maketitle

\begin{abstract}
We study the approximability of general convex sets in $\R^n$ by intersections of halfspaces, where the approximation quality is measured with respect to the \emph{standard Gaussian distribution} $N(0,I_n)$ and the complexity of an approximation is the number of halfspaces used.
While a large body of research has considered the approximation of convex sets by intersections of halfspaces under distance metrics such as the Lebesgue measure and Hausdorff distance, prior to our work there has not been a systematic study of convex approximation under the Gaussian distribution.

We establish a range of upper and lower bounds, both for general convex sets and for specific natural convex sets that are of particular interest.
Our results demonstrate that the landscape of approximation is intriguingly different under the Gaussian distribution versus previously studied distance measures.  For example, we show that $2^{\Theta(\sqrt{n})}$ halfspaces are both necessary and sufficient to approximate the origin-centered $\ell_2$ ball of Gaussian volume 1/2 to any constant accuracy, and that for $1 \leq p < 2$, the origin-centered $\ell_p$ ball of Gaussian volume 1/2 can be approximated to any constant accuracy as an intersection of $2^{\tilde{O}(n^{3/4})}$ many halfspaces.  These bounds are quite different from known approximation results under more commonly studied distance measures.

Our results are proved using techniques from many different areas. These include classical results on convex polyhedral approximation, Cram\'{e}r-type bounds on large deviations from probability theory, and---perhaps surprisingly---a range of topics from computational complexity, including computational learning theory, unconditional pseudorandomness, and the study of influences and noise sensitivity in the analysis of Boolean functions.

\end{abstract}

%
%
%
%

\newpage

\setcounter{tocdepth}{2}
{\tableofcontents}

\newpage

\pagenumbering{arabic}


\section{Introduction} \label{sec:intro}

A long line of mathematical research has investigated \emph{convex polyhedral approximation of convex bodies}, i.e. the broad question of how to best approximate general  convex bodies using intersections of halfspaces (or equivalently, convex hulls of finite point sets).
Research on questions of this sort dates back at least to the first half of the twentieth century, see for example the early works of Sas \cite{Sas39}, Fejes T\'{o}th \cite{FejesToth48} and Macbeath \cite{Macbeath51}, among others.  
Contemporary motivation for the study of polyhedral approximation of convex bodies arises from many areas including discrete and computational geometry, geometric convexity, the study of finite-dimensional normed spaces, and optimization. 
Given this breadth of connections, it is not surprising that the existing body of work on the subject is vast, as witnessed by the hundreds of references that appear in multiple surveys (including the 1983 \cite{Gruber83} and 1993 \cite{Gruber93} surveys of Gruber and the 2008 survey of Bronstein \cite{Bronstein08})
on approximation of convex bodies by polyhedra.

In the polyhedral approximation literature to date, a number of different distance notions have been used to measure the accuracy of convex polyhedral approximations of convex sets. 
These include the Hausdorff distance, the Nikodym metric (volume of the symmetric difference), the Banach-Mazur distance, distances generated by various $L_p$ metrics, and more.  
(See \cite{Bronstein08} for descriptions of each of these metrics and for an extensive overview of approximability results under each of those distance measures.)

This paper adopts a new perspective, by considering approximation of convex bodies in $\R^n$ using the \emph{(standard) Gaussian volume of the symmetric difference} as the distance measure. 
We refer to this distance measure simply as the \emph{Gaussian distance}. 
This is a natural and well-studied distance measure in theoretical computer science, as witnessed by the many works that have considered learning, testing, derandomization, and other problems using it (see e.g. \cite{KOS:08,Vempalafocs10,MORS:10,Kane:Gotsman11,Kane12,Kane14,KK14,KNOW14,Kane15,CDS19,DMN19,DKPV21,OSTK21,DMN21,HSSV22,DNS23-convex} 
and references therein). Given this body of work, it is perhaps surprising that Gaussian distance does not appear to have been previously studied in a systematic way for the convex polyhedral approximation problem.

Let us make our notion of approximation precise:

\begin{definition}  \label{def:Gaussian-volume}
Given (measurable) sets $K,L \subseteq \R^n$, we define the \emph{Gaussian distance} between $K$ and $L$ to be
\[
	\dG(K,L) 
	:= \Prx_{\bx \sim N(0,I_n)}\sbra{\bx \in K \, \triangle \, L}
\]
where $N(0,I_n)$ denotes the $n$-dimensional standard Gaussian distribution.
\end{definition}

We consider approximators which are intersections of finitely many $n$-dimensional halfspaces, and we measure the complexity of such an approximator by the number of halfspaces (facets) that it contains; we refer to this as its \emph{facet complexity}.
Thus, we are concerned with the following broad question:  

\begin{quote}
Given a convex set $K$ in $\R^n$ and a value $0 < \eps < 1/2$, what is the minimum 
facet complexity of an intersection of halfspaces $L$ 
such that $\dG(K, L) \leq \eps$?  
\end{quote}

\subsection{Motivation}

The broad approximation question stated above arises naturally in the context of contemporary theoretical computer science.  As alluded to earlier, an increasing number of results across different areas of TCS involve convex sets and the Gaussian distance;  taking the next step on various natural problems in these areas would seem to require an understanding of how well intersections of halfspaces can approximate convex sets in $\R^n$.  We give two specific examples below:

\begin{itemize}

\item \textbf{Testing Convex Sets:}  A number of researchers have considered the property testing problem of efficiently determining whether an unknown set $K \subseteq \R^n$ is convex versus far from convex \cite{RV04,CFSS17,BBB20,BB20,BBH23}. The Gaussian distance provides a  clean mathematical framework for this problem, and one in which a natural analogy emerges between the well-studied problem of monotonicity testing over the Boolean hypercube and convexity testing \cite{CDNSW23inprep}. Understanding how convex sets can be approximated by ``simpler objects'' under the Gaussian distance seems likely to be helpful in developing algorithms for testing convexity.  (In support of this thesis, we recall that the study of approximating linear threshold functions by ``simpler objects'' \cite{Servedio:07cc} yielded insights and technical ingredients, namely the notion of the ``critical index,'' that played an essential role in the development of efficient algorithms for testing linear threshold functions \cite{MORS:10}.) We also recall that several papers have given testing algorithms for single halfspaces \cite{MORS:10} and intersections of $k$ halfspaces \cite{DMN19, DMN21} over $\R^n$ under the Gaussian distance; combining these results with results about approximating general convex sets by intersections of halfspaces offers a potential avenue to developing testers for general convex sets.

\item \textbf{Properly Learning Convex Sets:} It has long been known \cite{BshoutyTamon:96} that monotone Boolean functions over $\{-1,1\}^n$ can be learned to any constant accuracy under the uniform distribution in time $2^{\tilde{O}(\sqrt{n})}$, via the ``low-degree'' algorithm based on the Fourier decomposition. Similarly, it has been known for some time \cite{KOS:08} that convex sets can be learned to any constant accuracy under the Gaussian distance, via a ``Hermite polynomial'' variant of the low-degree algorithm.  However, neither the algorithms of \cite{BshoutyTamon:96} nor \cite{KOS:08} are \emph{proper}: in the monotone case the output hypothesis of the learning algorithm is not a monotone function, and in the convex setting the output hypothesis is not a convex set.  Exciting recent work \cite{LRV22} has given a $2^{\tilde{O}(\sqrt{n})}$-time uniform-distribution learning algorithm for monotone Boolean functions over $\{-1,1\}^n$ that \emph{is} proper (even in the demanding \emph{agnostic learning} model, see \cite{LV23}); given this, it is a natural goal to seek a $2^{\tilde{O}(\sqrt{n})}$-time proper learning algorithm for convex sets under the Gaussian distance.  What would the output hypothesis of such a hoped-for proper learning algorithm look like? It seems quite plausible that it would be an intersection of halfspaces; this begs the question of understanding the capabilities and limitations of intersections of halfspaces as approximators for general convex sets in $\R^n$.

\end{itemize}

We further note that intersections of halfspaces have been intensively studied in many branches of ``concrete complexity theory;'' indeed, many of the TCS results alluded to earlier that involve Gaussian distance are more specifically about intersections of halfspaces, e.g.~\cite{KOS:08,Vempalafocs10,CDS19,DMN19,DKPV21,DMN21,HSSV22}. 

Finally, we note that the suite of tools which turn out to be relevant in our study link our investigations closely to theoretical computer science.  Perhaps unexpectedly, we will see that there are close connections between the study of facet complexity of polyhedral convex approximators and a number of fundamental ingredients and results (such as noise sensitivity, random restrictions, influence of variables/directions, various extremal constructions, etc.) in the \emph{analysis of Boolean functions} over the discrete domain $\{-1,1\}^n$.  From this perspective, our study can be seen as a continuation of the broad theme of exploring the emerging analogy between convex sets in $\R^n$ under the Gaussian distribution and monotone Boolean functions over $\{-1,1\}^n$ under the uniform distribution (see e.g.~\cite{DS21colt,DNS21itcs,DNS22,HSSV22}).

 

Before turning to an overview of our results and techniques, let us fix some convenient notation and terminology:

\begin{notation} \label{notation:fc-intro}
If $L \subseteq \R^n$ is an intersection of $N$ halfspaces, then as stated earlier we say that the \emph{facet complexity} of $L$ is $N$.
Given a convex set $K \subseteq \R^n$ and a value $0 \leq \eps < 1/2$, we write $\FC(K,\eps)$ to denote the minimum value $N$ such that there is a convex set $L$ that is an intersection $N$ halfspaces 
satisfying $\dG(K,L) \leq \eps.$
\end{notation}



\section{Overview of Results and Techniques}

\begin{table}[]
\centering
\renewcommand{\arraystretch}{1.25}
{
\begin{tabular}{@{}llll@{}}	
\toprule
Convex {Set} $K \sse\R^n$ & $\eps$        & Upper Bound                                                          & Lower Bound                                                                 \\ \midrule
Arbitrary $K$ with $\Vol(K) = 1-\delta$ & $c$                             & \begin{tabular}[t]{@{}l@{}}$\frac{1}{\delta} \cdot \left(n\log\pbra{\frac{1}{\delta}}\right)^{O(n)}$ \\[-0.1em] (\Cref{thm:rel-err-dud})  \end{tabular} & \hfill---~~~~~~~~                                                                                        \\[2.25em] 
$K = \ell_2$-ball {with $\Vol(K)=1/2$}          & $c$                        & \begin{tabular}[t]{@{}l@{}}$2^{O(\sqrt{n})}$ \\[-0.25em] (\Cref{thm:de-nazarov}) \end{tabular} & \begin{tabular}[t]{@{}l@{}}$2^{\Omega(\sqrt{n})}$ \\[-0.25em] (\Cref{thm:high-influence-lb}) \end{tabular}                                        \\[2.25em]
$K= \ell_1$-ball {with $\Vol(K)=1/2$}         & $c$                        & \begin{tabular}[t]{@{}l@{}}$2^{O(n^{3/4})}$  \\[-0.25em] (\Cref{thm:ell-p-approx}) \end{tabular}  & \begin{tabular}[t]{@{}l@{}}$2^{\Omega(\sqrt{n})}$ \\[-0.25em] (\Cref{thm:high-influence-lb}) \end{tabular}                                       \\[2.25em] 
Non-explicit $K$~~ & 
 \begin{tabular}[t]{@{}l@{}}$\frac12 - \tau$, {for} \\[-0.25em] {$\tau = \omega (n^{-1/2} \log n)$} \end{tabular}
 & ~~~~~~~~~--- & \begin{tabular}[t]{@{}l@{}}$2^{\Omega(\tau\cdot \sqrt{n})}$ \\[-0.25em] (\Cref{lb:non-explicit-convex})\end{tabular} \\[2.25em]
$K$ s.t. $\GNS_{n^{-c}}(K) \geq \frac{1}{2} - n^{-c}$~~ &  ${\frac{1}{2} - n^{-c}}$                                                                 & ~~~~~~~~~---~~~      & \begin{tabular}[t]{@{}l@{}}$2^{n^{\Omega(1)}}$ \\[-0.25em] (\Cref{thm:high-GNS-hard}) \end{tabular} \\ \bottomrule
\end{tabular}
}
\caption{A summary of some of our bounds on $\FC(K,\eps)$ (see~\Cref{notation:fc-intro}) for convex sets $K\sse\R^n$, where $c$ is any sufficiently small absolute constant.}
\label{tab:results}
\end{table}

A partial summary of some of our main results is presented above in~\Cref{tab:results}. 

\subsection{Overview of Positive Results}

We start with our positive results on approximability of convex sets under the Gaussian distance by intersections of halfspaces. 

\paragraph{Universal Approximation via Hausdorff Distance Approximation.}  As a warmup, we begin in \Cref{subsec:dudley} by giving a fairly simple ``universal approximation'' result that upper bounds $\FC(K,\eps)$ for any convex set $K \subseteq \R^n$.  The key observation (implicit in the proof of \Cref{thm:gaussian-dudley}) is that for any convex set $K$, if $L$ is an intersection of halfspaces which is a high-accuracy \emph{outer}\footnote{Recall that $L$ is an \emph{outer} (respectively \emph{inner}) approximator to $K$ if $K\sse L$ (respectively $L\sse K$).} approximator of $K$ under the Hausdorff distance measure, then $L$ is also a (slightly lower accuracy) $\dG$-approximator for $K$.  (Recall that the Hausdorff distance between two sets $S,T \subseteq \R^n$ is 
\[\max\cbra{\sup_{s \in S}\|s-T\|, \sup_{t \in S}\|t - S\|} \quad\text{where}~ \|x-Y\| = \inf_{y \in Y}\|x-y\|\] and $\|x-y\|$ is the Euclidean distance between $x$ and $y$.)  This is easy to establish by simply integrating Keith Ball's universal $O(n^{1/4})$ upper bound \cite{Ball:93} on the \emph{Gaussian surface area} of any convex set (see \Cref{thm:ball}), using the fact that for any $t > 0$ and any convex set $K$, the ``$t$-enlargement'' $K_t$ of $K$ is also a convex set (see the proof of \Cref{thm:gaussian-dudley}).

Combining the above observation with standard Gaussian tail bounds and classical upper bounds on the facet complexity of an outer Hausdorff approximator for any convex body with bounded radius, we obtain the following:

\begin{theorem} [Informal version of \Cref{thm:gaussian-dudley}] \label{thm:gaussian-dudley-informal}
For any convex set $K$ and any $\eps > 0$, we have $\FC(K,\eps) \leq (n/\eps)^{O(n)}.$
\end{theorem}

\paragraph{A ``Relative-Error'' Universal Approximation Sharpening of \Cref{thm:gaussian-dudley-informal}.}  Suppose that $K \subseteq \R^n$ is a convex set whose Gaussian volume is very large, i.e.~$1-\delta$ where $\delta$ is very small.  In this situation the trivial approximator which is all of $\R^n$ already has $\dG(K,\R^n)=\delta$, so a natural goal for approximation is to achieve small error \emph{relative to} the $\delta$ amount of mass which lies outside of $K$, i.e. to construct an intersection of halfspaces $L$ for which 
\[\dG(K,L) \leq \eps \delta.\]  
While \Cref{thm:gaussian-dudley-informal} shows that $(n/(\eps \delta))^{O(n)}$ halfspaces suffice for this, in \Cref{subsec:good-dudley} we give a stronger result which has a significantly improved dependence on $\delta$:

\begin{theorem} [Informal version of \Cref{thm:rel-err-dud}]
\label{thm:rel-err-dud-informal}
Let $0<\eps,\delta<1$ and let $K \subsetneq \R^n$ be a  convex set with 
$\Vol(K) = 1-\delta.$  Then 
\[\FC(K,\eps \delta) \leq 
\frac{1}{\delta} \cdot \pbra{\frac{n}{\eps}\log\pbra{\frac{1}{\delta}}}^{O(n)}.
\]
\end{theorem}

\noindent 
(We remark that the strengthening which \Cref{thm:rel-err-dud} achieves over \Cref{thm:gaussian-dudley} will be crucial for the construction of our $\ell_p$ ball approximators, discussed below.)

We briefly explain the main idea underlying the construction of the approximator  of \Cref{thm:rel-err-dud}.  First, we recall a basic property of the $N(0,I_n)$ distribution, which is that along any ray $\{tv: t \geq 0\}$ from the origin in $\R^n$, the distribution of Gaussian mass along that ray follows the chi-distribution (see \Cref{subsec:prelims-gaussian}).  Since $\Vol(K)=1-\delta$, this means that the ``amount of chi-distribution mass'' which is ``lopped off'' by $K$ along direction $v$, averaged over all directions $v$, is $\delta$.  The high-level idea of our construction of the approximator $L$ is to place tangent hyperplanes on the boundary of $K$ in such a way as to ``lop off'' {at least a} $(1-\eps)$-fraction of the total $\delta$ amount of Gaussian mass that lies outside the set $K$. We do this via an analysis that proceeds direction by direction: We can ignore directions in which $K$ ``lops off'' only a very small amount (less than $\eps \delta/4$, say) of the chi-distribution's mass. For the other directions, we use a bucketing scheme, a suitable net construction for the points in each bucket, and a mixture of careful geometric arguments and tail bounds to argue that not too many halfspaces are required to adequately handle all of the other directions.

\paragraph{\bf Sub-Exponential Approximation for Specific Convex Sets.}
Going beyond the universal approximation results mentioned above, it is natural to wonder whether improved approximation bounds can be obtained for specific interesting convex sets.
Perhaps the two most natural convex sets to consider in this context are the $\ell_1$ and $\ell_2$ balls
of Gaussian volume $1/2$, which we briefly discuss below.\footnote{It is also natural to wonder about the $\ell_\infty$ ball, but the $\ell_\infty$ ball $\cbra{x \in \R^n \ : \ \|x\|_\infty \leq r}$ is an intersection of $2n$ halfspaces $-r \leq x_i \leq r, i \in [n]$.}
  
The $\ell_1$ ball of Gaussian volume $1/2$, also known as the ``spectrahedron'' or ``cross-polytope,'' is the convex set 
\[
 \cbra{x \in \R^n \ : \|x\|_1 \leq c n},
\]
where $c = \sqrt{2/\pi} \pm o_n(1)$ (this is an easy consequence of the Berry-Esseen theorem).  This set is an intersection of exactly $2^n$ halfspaces, namely
\[
 b_1 x_1 + \cdots + b_n x_n \leq c n
\text{~for all~}(b_1,\dots,b_n) \in \{-1,1\}^n,
\]
and it is natural to wonder whether this set can be approximated as an intersection of significantly fewer than $2^n$ halfspaces.  {(As we will explain later, it is slightly more convenient for us to work with a slight rescaling of this ball, namely the set 
\[B_1 := \{x \in \R^n \ : \ \|x\|_1 \leq \sqrt{{\frac 2 \pi}} n\},\] which can easily be shown to have $\Vol(B_1) = 1/2 \pm o_n(1).$)}

The $\ell_2$ ball of Gaussian volume $1/2$ is simply the Euclidean ball
\[
\cbra{ x \in \R^n \ : \|x\|^2 \leq \text{median}(\chi^2(n))}.
\]
(Of course, writing this set \emph{exactly} as an intersection of halfspaces requires infinitely many halfspaces.)
Using well-known bounds on the chi-squared distribution we have that the Euclidean ball of radius ${\sqrt{n}}$, which we denote $B(\sqrt{n})$, has Gaussian volume $1/2 + o_n(1)$; it will be slightly more convenient for us to use this as the $\ell_2$ ball that we seek to approximate. 

\medskip

\paragraph{A 
$2^{O(\sqrt{n})}$-Facet 
Approximator for the $\ell_2$ Ball.}
For several standard distance measures, including the Hausdorff distance and the symmetric-difference measure, the $\ell_2$ ball is the ``hardest case'' for approximation by  intersections of halfspaces, matching or essentially matching known upper bounds for universal approximation.  For example, it is known \cite{Bronstein08} that any intersection of halfspaces which is an $\eps$-approximator to $B(\sqrt{n})$ in Hausdorff distance must have facet complexity $\Omega((\sqrt{n}/\eps)^{(n-1)/2})$, matching the universal approximation upper bound for Hausdorff distance to within a constant factor \cite{AFM12}.
For the symmetric-difference measure, Ludwig, Sch\"{u}tt and Werner \cite{LSW06} showed that if $L$ is an intersection of $M$ halfspaces for which 
{\[{\frac {\mathrm{Leb}(L \ \triangle \ B(\sqrt{n}))}{\mathrm{Leb}(B(\sqrt{n}))}} \leq c\] (here $\mathrm{Leb}$ denotes volume under the standard Lebesgue measure)},
where $c>0$ is some absolute constant, then $M$ must be at least $2^{\Omega(n)}$.

In contrast, for the Gaussian distance we are able to exploit the rotational symmetry of the standard Gaussian distribution and tail bounds on the chi-distribution to get an approximation of the $\ell_2$ ball which uses much fewer than $2^n$ facets, and hence is much better than the universal approximation upper bounds given by \Cref{thm:gaussian-dudley-informal} or \Cref{thm:rel-err-dud-informal}.  For any constant $\eps$, we give a probabilistic construction of a $2^{O(\sqrt{n})}$-facet polytope which is an $\eps$-approximator of the $\ell_2$ ball.  This construction is analogous to a probabilistic construction of O'Donnell and Wimmer \cite{o2007approximation}, which shows that a variant of a random monotone CNF construction due to Talagrand \cite{Talagrand:96} gives an $\eps$-accurate approximator for the Boolean Majority function.  We similarly modify  a probabilistic construction of a convex body due to Nazarov~\cite{Nazarov:03}, and use it to prove the following in \Cref{sec:nazarov-ub}:

\begin{theorem} [Informal version of \Cref{thm:de-nazarov}]
\label{thm:de-nazarov-informal}
Let $B_2$ denote the $\ell_2$ ball of radius $\sqrt{n}$ in $\R^n$ (so $\Vol(B_2) = 1/2 + o_n(1)$). Then for any constant $\eps > 0$,
\[\FC(B_2,\eps) \leq 2^{O(\sqrt{n})}.\]
\end{theorem}

\paragraph{A $2^{O(n^{3/4})}$-Facet Approximator for the $\ell_1$-Ball.} 
Since the $\ell_1$ ball does not enjoy the same level of rotational symmetry as the $\ell_2$ ball, it is natural to wonder whether it can be similarly approximated with a sub-exponential number of facets.
As we now explain, another motivation for this question comes from considering a result of O'Donnell and Wimmer on the inability of small CNF formulas to approximate the ``tribes'' Boolean DNF formula over $\zo^n$, in the context of a  recently-explored analogy between monotone Boolean functions and symmetric convex sets~\cite{DS21colt,DNS21itcs,DNS22,HSSV22}. 

In \cite{o2007approximation} O'Donnell and Wimmer showed that any CNF formula that computes the $n$-variable DNF tribes function correctly on $(1-\eps) \cdot 2^n$ inputs, for $\eps = 0.1$, must have $2^{\Omega(n/\log n)}$ clauses.  As detailed in \cite{DS21colt,DNS22}, there is a natural correspondence between $s$-clause CNF formulas over $\{0,1\}^n$ and symmetric intersections of $2s$ halfspaces over $N(0,I_n)$ (this correspondence is at the heart of the main lower bound of \cite{DS21colt} and several of the results of \cite{HSSV22,DNS22}).  Thus, in seeking convex sets that may require $2^{\tilde{\Omega}(n)}$ halfspaces to approximate under $N(0,1)^n$, it is natural to look for a Gaussian space convex analogue of of the DNF tribes function over $\{0,1\}^n$.  Now, 
\begin{itemize}

\item The DNF tribes function is the Boolean dual of the CNF Tribes function; 
\item The Gaussian-space polytope corresponding to the CNF Tribes function is the $\ell_\infty$ ball \cite{DNS22,HSSV22}; and 
\item The polytope which is dual to the $\ell_\infty$ ball is the $\ell_1$ ball.  

\end{itemize}
Thus the  following question naturally suggests itself:  
\begin{quote}
	Does the $\ell_1$ ball $B_1$ in $\R^n$ 
	require $2^{\tilde{\Omega}(n)}$ halfspaces for constant-accuracy approximation, analogous to how the DNF Tribes function requires $2^{\tilde{\Omega}(n)}$-clause CNFs for constant-accuracy approximation?
\end{quote}
While the above conjecture may seem intuitively plausible (and indeed, the authors initially tried to prove it), it turns out to be false. In \Cref{sec:ell-p}, we show that the $\ell_1$ ball can be approximated to any constant accuracy as an intersection of sub-exponentially many halfspaces:

\begin{theorem} [Special case of~\Cref{thm:ell-p-approx}]
\label{thm:ell-1-approx-informal}
Let 
$B_1$ denote the origin-centered $\ell_1$ ball, i.e. 
\[B_1 := \left\{x \in \R^n \ : \ \|x\|_1 \leq \sqrt{{\frac 2 \pi}} n\right\}.\]
Then for any constant $\eps$, we have
\[\FC({B_1},\eps) \leq
2^{O(n^{3/4})}.\]
\end{theorem}

We prove \Cref{thm:ell-1-approx-informal} using a probabilistic construction, but one which is very different from the probabilistic construction described above in the sketch of \Cref{thm:de-nazarov-informal} for the $\ell_2$ ball.  The high-level intuition is as follows:  Determining whether or not a point $x \in \R^n$ lies within $B_1$ is the same as determining whether or not the sum $|x_1| + |x_2| + \dots + |x_n|$ exceeds the threshold value $\sqrt{2/\pi}n$.  Given this, the main idea is to take a probabilistic approach which exploits both \emph{anti-concentration} and \emph{sampling}: 

\begin{itemize}

\item  Anti-concentration tells us that for $\bx \sim N(0,I_n)$, only a small fraction of outcomes of $|\bx_1|+\cdots + |\bx_n|$ will lie very close to the boundary value of $\sqrt{2/\pi}n$.  Thus, at the cost of a small error, we can assume that a typical outcome $x$ of $\bx$ either has $|x_1| + \cdots + |x_n| \geq \sqrt{2/\pi}n + \tau$ (call this a ``heavy'' $x$), or has $|x_1| + \cdots + |x_n| \leq \sqrt{2/\pi}n - \tau$ (call this a ``light'' $x$) for a suitable margin parameter $\tau$. 

\item Given this, we can use a sampling-based approach:   if we uniformly sample $\newN$ coordinates $\bi_1,\dots,\bi_\newN$ from $[n]$ and evaluate $|x_{\bi_1}| + \cdots + |x_{\bi_\newN}|$, then for a suitable threshold $\theta_1$ there will be a noticeable gap between (i) the (extremely large) probability that $|x_{\bi_1}| + \cdots + |x_{\bi_\newN}| \leq \theta_1$ when $x$ is a typical light point, and  (ii) the (still very large, but slightly smaller) probability that $|x_{\bi_1}| + \cdots + |x_{\bi_\newN}| \leq \theta_1$ when $x$ is a typical heavy point.

\end{itemize} 

The above gap means that by ANDing together a carefully chosen number $M$ of random sets of the form
\begin{equation}
\cbra{x \in \R^n \ : \ |x_{\bi_1}| + \cdots + |x_{\bi_\newN}| \leq \theta_1},\label{eq:juntaell1}
\end{equation}
we get an approximator for $B_1$ which is accurate on almost all of the points $x \in \R^n$ which are either light or heavy (and, as sketched above, almost all points drawn from $N(0,I_n)$ are either light or heavy, by anti-concentration). The detailed analysis requires sophisticated Cram\'{e}r-type bounds which give very tight \emph{multiplicative} control on tails of sums of random variables drawn \emph{without} replacement from a finite population of values.

 Note that each set of the form (\ref{eq:juntaell1}) is an $\newN$-variable ``junta'' $\ell_1$ ball over its $\newN$ relevant coordinates (and is also an intersection of $2^\newN$ halfspaces), so the AND of $M$ such sets is an intersection of $M2^\newN$ many halfspaces.  With careful setting of parameters we get that this is at most $2^{O(n^{3/4})}$ halfspaces, for any constant-factor approximation.

\medskip

\noindent {\bf A $2^{\tilde{O}(n^{3/4})}$-facet approximator for the $\ell_p$ ball, $1 \leq p < 2.$} \Cref{thm:de-nazarov-informal} and \Cref{thm:ell-1-approx-informal} give sub-exponential approximators for the $\ell_1$ and $\ell_2$ balls using two different sets of techniques.  We generalize the approach of \Cref{thm:ell-1-approx-informal} to give approximators of similar facet complexity for $\ell_p$ balls of volume $1/2 \pm o_n(1)$, for all $1 \leq p < 2$:

\begin{theorem} [Informal version of \Cref{thm:ell-p-approx}]
\label{thm:ell-p-approx-informal}
For $1 < p < 2$ let $B_p$ be the origin-centered $\ell_p$ balls in $\R^n$ defined in \Cref{eq:Bp}. 
 Then
for any constant $\eps$, we have
 \[
 \FC(B_p,\eps) \leq 2^{\tilde{O}(n^{3/4})}.
 \]
\end{theorem}

The high-level approach, in its first stage, is similar to the $\ell_1$-ball approximation sketched above.   We now employ anticoncentration of the sum $|\bx_1|^p + \cdots + |\bx_n|^p$, and now in place of random sets of the form (\ref{eq:juntaell1}) we use random sets of the form
\begin{equation} \label{eq:juntaellp}
\cbra{x \in \R^n \ : \ |x_{\bi_1}|^p + \cdots + |x_{\bi_\newN}|^p \leq \theta_p}
\end{equation}
for a suitable threshold value $\theta_p$.  As before, ANDing together a carefully chosen number $M$ of sets of this form will be a high-accuracy approximator to the volume-$1/2$ $\ell_p$ ball.  But while in the $p=1$ case each set of the form (\ref{eq:juntaell1}) was already an intersection of only $2^\newN$ halfspaces, for $1<p<2$ a set of the form (\ref{eq:juntaellp}) is not an intersection of any finite number of halfspaces.  Instead, we must replace each set of the form (\ref{eq:juntaellp}) with an intersection of (not too many) halfspaces which is a very high accuracy approximator for it.  Crucially, each set of the form (\ref{eq:juntaellp}) has extremely high volume; this is exactly the setting in which our efficient relative-error approximator, \Cref{thm:rel-err-dud-informal}, is useful.  We apply \Cref{thm:rel-err-dud-informal} to approximate each of the $M$ sets of the form (\ref{eq:juntaellp}) as an intersection of $N$ halfspaces (where $N$ is not too large), and then ANDing together all of those approximators we achieve an $MN$-halfspace approximator to the original $\ell_p$ ball $B_p$ in $\R^n$.

\subsection{Overview of Lower Bounds}

We now turn to a technical overview of our lower bounds. 

\paragraph{Non-Explicit Average-Case Lower Bounds.}  It is well known that standard counting arguments easily yield strong (average-case) lower bounds on the complexity of approximating (non-explicit) \emph{Boolean functions} that map $\bn$ to $\bits$. 
For a range of different computational models, these lower bounds show that even achieving approximation error $1/2 - o_n(1)$ requires any approximator to be exponentially large.

In our context of approximating convex sets in $\R^n$ by intersections of halfspaces, counting arguments are not quite as straightforward because there are infinitely many distinct halfspaces over $\R^n$.  Nevertheless, counting arguments can be brought to bear to prove lower bounds. In more detail, \cite{KOS:08} used a counting argument to establish the existence of $N=2^{2^{\Omega(\sqrt{n})}}$ distinct convex sets $K_1,\dots,K_N$ in $\R^n$, each of which is an intersection of $2^{\Omega(\sqrt{n})}$ many halfspaces, such that $\dG(K_i,K_j) \geq 1/44000$ for each $1 \leq i \neq j \leq N.$  Using this result, it is possible to establish the existence of a convex set $K$ (one of the sets $K_1,\dots,K_N$) such that $\FC(K,c) \geq 2^{c \sqrt{n}}$, where $c>0$ is a small absolute constant.

In \Cref{sec:nonexplicit} we give a stronger (but still non-explicit) lower bound, which shows that many halfspaces are required even to achieve error $1/2 - o_n(1)$; in other words, we give an \emph{average-case} lower bound (also known as a \emph{correlation bound}). Our lower bound trades off the accuracy of the approximator against the number of halfspaces:

\begin{theorem} [Informal version of \Cref{lb:non-explicit-convex}] \label{thm:lb-non-explicit-convex-informal}
For any $\tau = \omega (n^{-1/2} \cdot \log n)$, there exists some convex set $K \subset \mathbb{R}^n$ that has $\FC(K,\frac12 - \tau) = 2^{\Omega(\tau\cdot \sqrt{n})}$.
\end{theorem}

\noindent For example, taking $\tau = n^{-1/4}$, \Cref{thm:lb-non-explicit-convex-informal} implies the existence of an $n$-dimensional convex set $K$ such that no intersection of $2^{cn^{1/4}}$ many halfspaces can approximate $K$ to accuracy even as large as $1/2 + n^{-1/4}$; taking $\tau$ to be a constant, we recover the non-explicit $2^{\Omega(\sqrt{n})}$ lower bound mentioned in the previous paragraph.

The proof of \Cref{thm:lb-non-explicit-convex-informal}  combines classical results from statistical learning theory with an information theoretic lower bound on \emph{weak} learning convex sets under $N(0,I_n)$ that was established in recent work \cite{DS21colt}. 
A high-level sketch is as follows:  If every convex set could be approximated to high accuracy ($1/2 - \tau$) by an intersection of ``few'' halfspaces, then classical results from statistical learning theory would imply the existence of (computationally inefficient but sample-efficient) algorithms to ``weakly'' learn any convex set $K$ to error $1/2 - \tau/2$.  On the other hand, a recent result of De and Servedio \cite{DS21colt} gives a strong lower bound on the error that any sample-efficient weak learning algorithm for convex sets must incur. The tension between these two bounds can be shown to imply a lower bound on the number of halfspaces that any $(1/2-\tau)$-approximator must use; see \Cref{sec:nonexplicit} for the detailed argument.
We remark that the proof of \Cref{thm:lb-non-explicit-convex-informal} is the first result we know of in which an information-theoretic sample complexity lower bound for learning is used to establish an \emph{inapproximability} result.

\paragraph{\bf Average-Case Lower Bounds via Gaussian Noise Sensitivity.} 
One drawback of \Cref{thm:lb-non-explicit-convex-informal} is that it is non-constructive: not only does it not exhibit any particular convex set which is hard to approximate, it does not establish any particular criterion that implies that a convex set is hard to approximate.
\Cref{sec:noise-sensitivity} gives such a criterion. We show that any set with very high \emph{Gaussian noise sensitivity} (at least $1/2 - 1/\poly(n)$) at very low noise rates (at most $1/\poly(n)$) requires many halfspaces to approximate to accuracy even $1/2 + 1/\poly(n)$:

\begin{theorem} [Informal version of \Cref{thm:high-GNS-hard}]
\label{thm:high-GNS-hard-informal}
If $K \subset \R^n$ is any convex set which satisfies $\GNS_{n^{-c}}(K) \geq {\frac 1 2} - n^{-c}$, then $\FC(K,{\frac 1 2} - n^{-c}) \geq 2^{n^{\Omega(1)}},$ where $c>0$ is a suitable absolute constant.
\end{theorem}

\Cref{thm:high-GNS-hard-informal} naturally raises the question of whether there in fact exist $n$-dimensional convex sets that have Gaussian noise sensitivity as high as $1/2 - 1/\poly(n)$ at noise rates that are as low as $1/\poly(n)$.  While natural, this question does not appear \cite{OD23pc} to have been asked or answered prior to the current work.  In \Cref{sec:existence} we combine the information-theoretic lower bounds of  \cite{DS21colt} on weak learning convex sets with the ``low-degree algorithm'' for learning bounded functions to show that there do exist $n$-dimensional convex sets $K$ with $\GNS_{1/\poly(n)}(K) \geq 1/2 - 1/\poly(n)$ (see \Cref{thm:highGNSlownoise}). We believe that this result may be of independent interest.

The proof of \Cref{thm:high-GNS-hard-informal} combines results from learning theory \cite{KOS:08} and the study of unconditional pseudorandomness \cite{OSTK21} with classical ideas from the study of correlation bounds for Boolean functions \cite{Hastad86,Cai86}.  
It employs a notion of a ``random zoom'' for a function $f: \R^n  \to \R$ which is closely analogous to the well-studied notion of a ``random restriction'' of a function from $\bn$ to $\bn$.

Analogous to the well-known phenomenon that low-depth Boolean circuits over $\bn$ simplify under random restrictions \cite{Hastad86,Cai86}, in the context of developing efficient pseudorandom generators for intersections of halfspaces it was shown in \cite{OSTK21} that low-degree polynomials over $N(0,I_n)$ simplify under random zooms; more precisely, they are very likely to have very low ``hypervariance'' after a random zoom (see \Cref{subsec:zoom-prelims} for detailed definitions).  Moreover, it was shown in \cite{KOS:08} (in the context of developing efficient agnostic learning algorithms for intersections of halfspaces) that any intersection of ``not too many'' halfspaces is close (in a suitable sense) to a low-degree polynomial.  Combining these results, we get that (1) applying a random zoom to any intersection of few halfspaces, we get a function that must have very low hypervariance, and hence must be very close to a constant.  But on the other hand, it can be shown that (2) if a set $K$ has very high Gaussian noise stability at very low noise rates, then after a random zoom it must be very \emph{far} from a constant function.  Similar to classical correlation bounds for shallow Boolean circuits (which use the fact that that the parity function is far from a constant function after random restriction, while shallow Boolean circuits become close-to-constant after random restrictions), the tension between (1) and (2) can be used to deduce that any intersection of ``not too many'' halfspaces must be very far from any set $K$ that has high Gaussian noise stability at low noise rates.


\paragraph{Lower Bounds via Convex Influences.}
Neither of the two lower bound techniques discussed thus far yield lower bounds for any explicit convex sets that we know of. 
For our final lower bound, we use the recently introduced notion of convex influences~\cite{DNS21itcs,DNS22} to prove lower bounds for the $\ell_1$ and $\ell_2$ balls. 
In more detail, in~\Cref{sec:conv-inf} we give a general criterion on a convex set $K \subset \R^n$ which suffices to ensure that $\FC(K,\eps) = 2^{\Omega(\sqrt{n})}$, where $\eps$ is a suitable constant. 
This criterion is that $K$ is symmetric\footnote{Recall that a set $K\sse\R^n$ is \emph{symmetric} if $x\in K$ implies $-x\in K$.} and has maximal \emph{convex influence}, up to a multiplicative constant factor.
Since both the $\ell_1$ ball of Gaussian volume $1/2$ and the $\ell_2$ ball of volume $1/2$ satisfy this condition, we get a $2^{\Omega(\sqrt{n})}$ lower bound for each of these explicit sets.

The notion of convex influence was recently introduced in the paper~\cite{DNS22} as an analogue of the classical notion of influence from Boolean function analysis over the discrete cube~\cite{KKL:88, odonnell-book}. In particular, for any unit vector $v$ and convex set $K$, $\TInf_v[K]$ is meant to capture how $K$ varies just in the direction $v$ (while averaging out every other direction) --- see~\Cref{def:convex-influence} for the precise formulation. An attractive feature of this notion of influence is that for any orthonormal basis of $\mathbb{R}^n$, say $\{v_1, \ldots, v_n\}$, the sum $\sum_{i=1}^n\TInf_{v_i}[K]$ is independent of the choice of the basis and depends only on the set $K$. This leads us to the notion of total influence of a set $K$, denoted by $\TInf[K] = \sum_{i=1}^n\TInf_{v_i}[K]$  (where $\{v_1, \ldots, v_n\}$ is any orthonormal basis of $\mathbb{R}^n$). 

\cite{DNS22} proved several structural properties of both influence and total influence. In particular, Proposition~19 of \cite{DNS22} established that 
the total influence of any convex set $K$ is $O(\sqrt{n})$; furthermore, this bound is tight, as exhibited by $K= B(\sqrt{n})$, which is shown in \cite{DNS22} to be the convex set of maximal total influence.

In this paper, we show that any symmetric set $K$ with maximal total influence (up to a constant factor) requires any approximator to have high facet complexity. In particular, we prove the following theorem:

\begin{theorem} [Informal version of \Cref{thm:high-influence-lb}] \label{thm:high-influence-lb-informal}
If $K$ is symmetric and $\TInf[K] = \Omega(\sqrt{n})$, then for a suitable constant $\eps>0$ we have that $\FC(K,\eps) = 2^{\Omega(\sqrt{n})}.$
\end{theorem}

The main idea in the proof of the above theorem is a new structural result establishing that if $L$ is a convex set which is an intersection of $s$ halfspaces, then $\TInf[L] = O(\log s)$ (Proposition~\ref{prop:convex-boppana}). Coupled with the fact that $\TInf[K]$ can be expressed in terms of the degree-$2$ Hermite coefficients of $K$,  a fairly simple argument using Parseval's identity and Cantelli's inequality leads to Theorem~\ref{thm:high-influence-lb-informal}.  

Once again, we observe that there is an analogy here with the Boolean setting. Namely,  Boppana~\cite{Boppana1997} showed that for any $s$-clause CNF formula, its total influence (as defined over the Boolean cube) is bounded by $O(\log s)$. We note that despite the syntactic analogy with~\Cref{prop:convex-boppana}, the underlying techniques are completely different; while~\Cref{prop:convex-boppana} is an inherently geometric argument, the proof of Boppana relies on the method of random restrictions~\cite{Hastad86}. 


\medskip


%
%
%
%
%

\subsection{Discussion}

The current paper takes the first steps in studying polytopal approximations of convex sets under Gaussian measure. Several tantalizing questions emerge from our work; we highlight a few of these now. 

Perhaps the most natural question is to close the gap between the worst case upper and lower bounds on $\FC(K,\epsilon)$ in the constant-$\epsilon$ regime. In particular, both \Cref{thm:high-influence-lb} and \Cref{lb:non-explicit-convex} guarantee the existence of convex sets $K$ such that 
$\FC(K,\epsilon)= 2^{\Omega(\sqrt{n})}$ for some constant $\epsilon>0$. On the other hand, \Cref{thm:gaussian-dudley,thm:rel-err-dud} show that for any constant $\epsilon>0$, $\FC(K,\epsilon) = n^{O(n)}$. Can we close this gap? 

In a related vein, it would also be interesting to obtain tighter upper and lower bounds on $\FC(K,\epsilon)$ when $K$ is the $\ell_p$ ball of Gaussian volume $1/2$. In particular, for $p>2$, we do not have an upper bound on the facet complexity (when the error $\epsilon $ is a constant) that is better than $n^{O(n)}$. Is it possible to obtain polytopal approximations for the $\ell_p$ ball with $2^{n^{1-c}}$ facets for constant $c>0$ in the constant error regime? 

Finally, looking ahead, it would be interesting to study Gaussian-distance polytopal approximation of convex sets via approximators with small \emph{vertex complexity}. Very little seems to be known here and indeed, the vertex complexity analogues of many questions considered in the current paper seem to be wide open. For  example,~\Cref{sec:nazarov-ub} gets a $2^{O(\sqrt{n})}$-facet approximator for the $\ell_2$ ball. Can we obtain a similar approximator for the $\ell_2$ ball which has subexponential vertex complexity?

\section{Preliminaries}
\label{sec:prelims}

\Cref{subsec:notation} recalls basic background and sets up notation,
\Cref{subsec:prelims-gaussian} gives standard facts about the Gaussian and chi-squared distributions,
\Cref{subsec:gsa} recalls known bounds on Gaussian surface area, and finally  \Cref{subsec:distance-metrics} defines the various distance metrics between convex sets that we will use.

\subsection{Basic Notation and Terminology}
\label{subsec:notation}

We use boldfaced letters such as $\bx, \boldf,\bA$, etc. to denote random variables (which may be real-valued, vector-valued, function-valued, or set-valued; the intended type will be clear from the context).
We write $\bx \sim \calD$ to indicate that the random variable $\bx$ is distributed according to probability distribution $\calD$. We will frequently identify a set $K\sse\R^n$ with its $0/1$-valued indicator function.  

\begin{notation} [Multiplicative approximation]
    We use the following notation to denote that two nonzero reals $a,b$ are multiplicatively close:  For $\nu \geq 0$,
    \[
        a \approx_\nu b \quad \iff \quad e^{-\nu} \leq a/b \leq e^\nu
    \]
    (note that this condition is symmetric in $a$ and $b$).  
\end{notation}

\paragraph{Geometry.} We write $e_i\in\R^n$ to denote the $i^{\text{th}}$ standard basis vector. For $r >0$, we write $\S^{n-1}(r)$ to denote the origin-centered sphere of radius $r$ in $\R^n$
and $B(r)$ to denote the origin-centered ball of radius $r$ in $\R^n$, i.e.,
$$
\S^{n-1}(r) = \big\{x \in \R^n: \|x\|=r\big\}\quad\text{and}\quad
B(r) = \big\{x \in \R^n : \|x\| \leq r\big\},
$$
where $\|x\|$   denotes the $\ell_2$ norm $\|\cdot \|_2$ of $x\in \R^n$.
We also write $\S^{n-1}$ for the unit sphere $\S^{n-1}(1)$.

\paragraph{Convex Sets and Convex Bodies.}
A set $K \subseteq \R^n$ is \emph{convex} if $x,y \in K$ implies $\alpha\hspace{0.03cm}x + (1-\alpha) y \in K$ for all $\alpha\in [0,1]$.  We recall (see e.g.~\cite{lang1986note}) that all convex sets are Lebesgue measurable. 

A convex \emph{body} in $\R^n$ is a compact convex set with non-empty interior.  Our results will hold for general convex sets (not only bodies), but since we are working with Gaussian distance, it is easy to see that it suffices to consider only convex bodies.  (If a convex set has empty interior then its Gaussian volume is 0; the Gaussian volume of the closure of a convex set is the same as the Gaussian volume of the set; and if a convex set $A$ is unbounded, we can ``truncate'' it by intersecting it with a sufficiently large ball to obtain a bounded convex set $A' \subset A$ with $\dG(A,A') \leq \eps/2$.)

Finally, for sets $K,L \subseteq \R^n$, we write $K+L$ to denote the Minkowski sum $\{x+y : x \in K\ \text{and}\ y \in L\}.$ For a set $K \subseteq \R^n$ and $r > 0$ we write $rK$ to denote the set $\{rx : x \in K\}$.

\paragraph{Intersections of Halfspaces and Approximation.}
For the sake of readability we will mostly state our results in a self-contained way, but the following notation will sometimes be useful:

\begin{itemize}

\item We write $\Facet(n,M)$ to denote the class of all convex sets which are intersections of $M$ halfspaces in $\R^n$, i.e.~the class of all $M$-facet convex sets.  

\item For a convex set $K \subseteq \R^n$ and a value $0 < \eps < 1$, we write $\FC(K,\eps)$ to denote the minimum value $M$ such that there is some intersection of $M$ halfspaces $L \in \Facet(n,M)$ such that $\dG(K,L) \leq \eps$, i.e.~$\FC(K,\eps)$ is the minimum ``facet complexity'' of any $\eps$-approximator of $K$.

\end{itemize}


\subsection{The Gaussian and Chi-Squared Distributions}
\label{subsec:prelims-gaussian}

We write $N(0,I_n)$ to denote the $n$-dimensional standard Gaussian distribution, and denote its density function by $\phi_n$, i.e. 
\[\varphi_n(x) = (2\pi)^{-n/2} e^{-\|x\|^2/2}.\]
When the dimension is clear from context, we may simply write $\phi$ instead of $\phi_n$. We write $\vol(K)$ to denote the Gaussian measure of a (Lebesgue measurable) set $K \subseteq \R^n$, that is 
\[\vol(K) := \Prx_{\bx \sim N(0,I_n)}[\bx \in K].\]  
We recall the following standard tail bound on Gaussian random variables:

\begin{proposition}[Theorem~1.2.6 of \cite{durrett_2019} or Equation~2.58 of \cite{TAILBOUND}] \label{prop:gaussian-tails}
	Suppose $\bg\sim N(0,1)$ is a one-dimensional Gaussian random variable. Then {for all $r>0$,}
\[
\varphi_1(r)
\left({\frac 1 r} - {\frac 1 {r^3}} \right) \leq \Prx_{\bg \sim N(0,1)}[\bg \geq r] \leq
\varphi_1(r)
\left({\frac 1 r} - {\frac 1 {r^3}} + {\frac 3 {r^5}}\right)
\]
where $\phi_1$ is the one-dimensional Gaussian density.
\end{proposition}

{
We will also make use of the Berry-Esseen central limit theorem~\cite{berry,esseen} (alternatively, see Section 11.5 of~\cite{odonnell-book}):

\begin{theorem} \label{thm:berry-esseen}
	Let $\bX_1, \ldots, \bX_n$ be independent random variables with $\Ex\sbra{\bX_i} = 0$ and $\Varx\sbra{\bX_i} = \sigma_i^2$, and assume $\sumi \sigma_i^2 = 1$. Let $\bS = \sumi \bX_i$ and let $\bZ \sim N(0,1)$ be a standard univariate Gaussian. Then for all $u\in\R$, 
	\[\abs{\Prx\sbra{\bS\leq u} - \Prx\sbra{\bZ \leq u}} \leq c\gamma\] 
	where 
	\[\gamma = \sumi \Ex\sbra{|\bX_i|^3}\]
	and $c \leq 0.56$ is a universal constant. 
\end{theorem}
}

Recall that the norm of an $n$-dimensional Gaussian random vector is distributed according to the chi-squared distribution with $n$ degrees of freedom, i.e. if $\bx\sim N(0,I_n)$ then $\|\bx\|^2 \sim \chi^2(n)$. {It is well known (see e.g.~\cite{Wiki-chisquare}) that the mean of the $\chi^2(n)$ distribution is $n$, the median is $n\pbra{1 - \Theta(1/n)}$, and for $n\geq 2$ the pdf is everywhere at most 1. We note that an easy consequence of these facts is that the origin-centered ball $B(\sqrt{n})$ of radius $\sqrt{n}$ in $\R^n$ has $\Vol(B(\sqrt{n})) = 1/2 + o_n(1).$} 

We will require the following tail bound on $\chi^2(n)$ random variables:

\begin{proposition}[Section 4.1 of \cite{laurent2000}] \label{prop:chi-squared-tail}
	Suppose $\by\sim\chi^2(n)$. Then for any $t > 0$, we have 
	\begin{gather*}
		\Prx_{\by\sim\chi^2(n)}\sbra{\by \geq n + 2\sqrt{nt} + 2t} \leq \exp\pbra{-t}, \\
		\Prx_{\by\sim\chi^2(n)}\sbra{\by \leq n - 2\sqrt{nt}} \leq \exp\pbra{-t} .
	\end{gather*}
\end{proposition}

\subsection{Bounds on Gaussian Surface Area}
\label{subsec:gsa}

Given a measurable set $K\sse\R^n$, recall that the \emph{Gaussian surface area of $K$}---which we will denote as $\GSA(K)$---is given by 
\[\GSA(K) := \lim_{\delta\to0}\frac{\vol\pbra{K + B(\delta)} - \vol(K)}{\delta}.\] 
For convex sets (more generally, for sets which are sufficiently regular, e.g. have smooth boundary except at a set of measure zero), we have
\begin{equation} \label{eq:gsa-def}
	\GSA(K) = \int_{\partial K} \phi(x) \,d\sigma(x)
\end{equation}
See~\cite{Nazarov:03} or Definition~2 of~\cite{KOS:08} for further discussion on this point.
The following upper bound on the Gaussian surface area of an intersection of halfspaces was obtained by Nazarov; see~\cite{KOS:08,Kane:sens14}.

\begin{theorem}[Nazarov's bound] \label{thm:nazarov}
	Suppose $K \sse \R^n$ is an intersection of $s$ halfspaces. Then we have
	\[\GSA(K) \leq \sqrt{2\ln s} + 2.\] 
\end{theorem} 

We will also require the following upper bound on the Gaussian surface area of an arbitrary convex set in $\R^n$, which was obtained by Ball~\cite{Ball:93}:

\begin{theorem}[Ball's theorem] \label{thm:ball}
	Suppose $K\sse\R^n$ is a convex set. Then we have 
	\[\GSA(K) \leq O(n^{1/4}).\]
\end{theorem}

Finally, we recall the Gaussian isoperimetric inequality~\cite{sudakov1978extremal,borell1975brunn}:

\begin{theorem}
	Suppose $K\sse\R^n$. Let $H \sse\R^n$ be a halfspace (i.e. $H = \{x \in \R^n : \langle x, v\rangle \leq \theta \}$ for $v\in\S^{n-1}$ and $\theta\in\R$) such that $\vol(H) = \vol(K)$.
	Then \[\GSA(K) \geq \GSA(H) = \phi_1(\theta)\]
	where $\phi_1$ denotes the univariate Gaussian p.d.f. 
\end{theorem}

\subsection{Distance Metrics Between Sets}
\label{subsec:distance-metrics}

The primary distance metric we will use throughout this paper is the following: Given two measurable sets $K, L \sse\R^n$, we define
\[
	\dG(K,L) := \Prx_{\bx\sim N(0,I_n)}\sbra{K(\bx) \neq L(\bx).}
\]
In other words, $\dG(K,L) = \vol(K\, \triangle\, L)$, i.e. the Gaussian measure of the symmetric difference of the sets $K$ and $L$. 

We also recall the \emph{Hausdorff distance} between two sets $K,L\sse\R^n$, which is defined as 
\[\dH(K,L) := \max\cbra{\sup_{x\in K} \newinf_{y\in L}d(x,y),\, \sup_{y\in L} \newinf_{x\in K} d(x,y)}\]
where $d(x,y) := \|x-y\|_2$ denotes the usual Euclidean $\ell_2$ distance between the points $x$ and $y$. We will rely on bounds on polytope approximation under the Hausdorff distance in~\Cref{subsec:dudley}.


%

\part{Upper Bounds}

We start by giving generic upper bounds on the number of halfspaces needed to approximate arbitrary $n$-dimensional convex sets under the Gaussian distance.
In~\Cref{subsec:dudley}, we show that 
\[\pbra{\wt{O}(n^{5/4}/\epsilon)}^{(n-1)/2}\] halfspaces suffice to $\eps$-approximate any convex set under the Gaussian distance metric, where the $\wt{O}$ hides a logarithmic factor in $1/\epsilon$. 
In~\Cref{subsec:good-dudley}, we sharpen the bound from~\Cref{subsec:dudley} for sets of large-enough volume. In particular, we show that if $K \subsetneq \R^n$
is a convex set with $\vol(K) = 1 -\delta$, then
\[\frac{1}{\delta} \cdot \pbra{\frac{n}{\eps}\log\pbra{\frac{1}{\delta}}}^{O(n)}\] halfspaces 
suffice to $\eps\delta$-approximate $K$ under the Gaussian distance metric. 

In~\Cref{sec:nazarov-ub,sec:ell-p}, we show that the above exponential-in-$n$ bounds on the facet complexity of convex sets can be significantly sharpened for specific convex sets. We show in~\Cref{sec:nazarov-ub} that $\exp(O(\sqrt{n}))$ halfspaces suffice to approximate the $\ell_2$ ball of Gaussian volume $1/2$ 
to constant accuracy. In~\Cref{sec:ell-p}, we give an entirely different construction from that in~\Cref{sec:nazarov-ub} to show that for $1 \leq p < 2$, $\exp(\wt{\Theta}(n^{3/4}))$ halfspaces suffice to approximate $\ell_p$ balls of Gaussian volume $1/2$ to constant accuracy.
Crucially, our construction in~\Cref{sec:ell-p} will make use of the improved generic upper bound from~\Cref{subsec:good-dudley}.


\section{Generic Upper Bounds}
\label{sec:generic-ubs}

We first give upper bounds on the facet complexity of arbitrary convex sets in $\R^n$ under the Gaussian distance metric. 

\subsection{Warmup: Universal Approximation via Hausdorff Distance Approximation}
\label{subsec:dudley}

The following upper bound on the facet complexity of polytope approximators under the \emph{Hausdorff distance} (cf.~\Cref{subsec:distance-metrics}) was independently obtained by Dudley~\cite{Dudley1974} and by 
Bronstein and Ivanov~\cite{bronshteyn1975approximation}; see also Section~4.1 of~\cite{Bronstein08}. 


\begin{theorem}[\cite{bronshteyn1975approximation}] \label{thm:dudley}
	Suppose $K \sse\R^n$ is a compact convex set with non-empty interior that is contained in the unit ball in $\R^n$, i.e. 
	\[K \sse B(1).\]
	For $0 < \eps < 10^{-3}$, there exists a convex body $L$ which is the intersection of 
	\[3\sqrt{n}\pbra{\frac{9}{\eps}}^{(n-1)/2}~\text{halfspaces}\] 
	such that 
	$K \sse L$ and $\dH(K,L) \leq \epsilon.$
\end{theorem}

We remark that the original results due to~\cite{Dudley1974,bronshteyn1975approximation} consider the \emph{vertex complexity} (i.e. number of vertices) of approximators instead of the {facet complexity} and furthermore only consider inner approximators to $K$ (i.e. $L\sse K$ in the above); their arguments, however, can be easily modified to obtain the above \cite{Mountslides18}.
Finally, we note that the bound in \Cref{thm:dudley} is known to be close to tight: any outer (or inner) approximation to the unit ball $B(1)$ in $\R^n$ requires $\Omega(1/\eps^{(n-1)/2})$ halfspaces \cite{AFM12,Bronstein08}.

\Cref{thm:dudley} and Ball's universal bound on the Gaussian surface area of any convex set (\Cref{thm:ball}) imply the following upper bound for approximating generic convex sets in $\R^n$ under the Gaussian distance:

\begin{theorem} [Universal approximation of convex sets.] \label{thm:gaussian-dudley}
	Let $0 < \eps < 10^{-3}$. For every convex set $K \sse\R^n$, there exists a convex set $L$ which is an intersection of 
	\[O\pbra{\frac{n^{5/4} + {2n^{3/4}}\sqrt{{\ln}(2/\epsilon)}}{\eps}}^{(n-1)/2}~\text{halfspaces} 
	\]
	such that 
	$\dG(K,L) \leq \epsilon$, i.e.~we have $\FC(K,\eps) \leq O\pbra{\frac{n^{5/4} + {2n^{3/4}}\sqrt{{\ln}(2/\epsilon)}}{\eps}}^{(n-1)/2}.$
\end{theorem}

\begin{proof}
	For the rest of the argument, set
	\[r^\ast := \sqrt{n} + 2\sqrt{{\ln}\pbra{\frac{2}{\eps}}} \qquad\text{and}\qquad \eps' := \Theta\pbra{\frac{\epsilon}{n^{1/4}}}.\]
	Recalling that the squared norm of a Gaussian vector is distributed according to a $\chi^2(n)$ distribution, we have by \Cref{prop:chi-squared-tail} that 
	\begin{equation} \label{eq:dudley-ball-trunc}
		\Prx_{\bx\sim N(0,I_n)}\sbra{\|\bx\|^2 > {r^\ast}^2} \leq \Prx_{\|\bx\|^2\sim \chi^2(n)}\sbra{\|\bx\|^2 \geq n + 2\sqrt{n {\ln}\pbra{\frac{2}{\eps}}} + 2 {\ln}\pbra{\frac{2}{\epsilon}}} \leq \frac{\eps}{2}.
	\end{equation}
	It thus suffices to $(\eps/2)$-approximate $K' := K\cap B(r^\ast)$ under the Gaussian distance metric. 
	
	It is easy to see by definition of the Hausdorff distance that  
	\[\dH(K,L) \leq \delta \quad\text{is equivalent to}\quad \dH(tK, tL) \leq t\delta.\]
	In particular, consider the convex set $({1}/{r^\ast})K' \sse B(1)$, and let $L \sse\R^n$ be the polytope guaranteed to exist by \Cref{thm:dudley} such that {$K \subseteq L$ and}
	\[\dH\pbra{\frac{1}{r^\ast}K', \frac{1}{r^\ast}L} \leq \frac{\eps'}{2r^\ast}, \qquad\text{which implies that}\qquad \dH(K', L) \leq \frac{\epsilon'}{2}.\] 
	Recalling our choice of $\eps'$, note that~\Cref{thm:dudley} guarantees that $L$ is an intersection of 
	\[O\pbra{\frac{18\sqrt{n}r^\ast}{\eps'}}^{(n-1)/2} = O\pbra{\frac{n^{5/4} + {2n^{3/4}}\sqrt{{\ln}(2/\epsilon)}}{\eps}}^{(n-1)/2}~\text{halfspaces.}\]
	We will show that $L$ is in fact an $(\epsilon/2)$-approximator for $K'$ under the Gaussian distance, which would complete the proof as 
	\begin{equation} \label{eq:dudley-potato}
		\dG(K, L) \leq \dG(K,K') + \dG(K',L) \leq \frac{\eps}{2} + \frac{\eps}{2} = \frac{\eps}{2}.
	\end{equation}
 
	Indeed, by definition of the Hausdorff distance we have that
	\begin{equation} \label{eq:dudley-inclusions}
		K' \sse L \sse K'_{\eps'} \qquad\text{where}~K'_{\eps'} := \cbra{x\in\R^n : \inf_{y\in K} d(x,y) \leq \eps'}.
	\end{equation}
	Consequently, the Gaussian distance between the sets $K'$ and $L$ is	
	\begin{align}
		\dG(K', L) &= \Prx_{\bx\sim N(0,I_n)}\sbra{K'(\bx) \neq L(\bx)} \nonumber \\ 
		&= \Prx_{\bx\sim N(0,I_n)}\sbra{\bx \in L \setminus K'} \nonumber \\ 
		&\leq \Prx_{\bx\sim N(0,I_n)}\sbra{\bx \in K'_{\eps'} \setminus K'} \nonumber \\ 
		&= \int_{t=0}^{\eps'} \int_{x \in \partial K_{t}} \varphi(x)\,dx \,dt \nonumber \\
		&= \int_{t=0}^{\eps'} \GSA(K_t) \, dt \nonumber \\ 
		&\leq \eps' \cdot O(n^{1/4}) \label{eq:dudley-proof-nazarov-app}\\
		&= \frac{\eps}{2}. \label{eq:dudley-final}
	\end{align}
	Here, $K_t$ is defined as in~\Cref{eq:dudley-inclusions}; it is easy to see that if $K$ is a convex, then so is $K_t$ for all $t \geq 0$. 
	Note that we used Keith Ball's bound on the Gaussian surface area of convex sets (\Cref{thm:ball}) to obtain \Cref{eq:dudley-proof-nazarov-app}. 
	The theorem now follows from \Cref{eq:dudley-potato,eq:dudley-final}.
\end{proof}

\subsection{A Relative-Error Sharpening of \Cref{thm:gaussian-dudley}}
\label{subsec:good-dudley}

We now establish a sharpening of \Cref{thm:gaussian-dudley} that gives a much better bound in the regime where $K$ has volume very close to 1:

\begin{theorem} \label{thm:rel-err-dud}
	Let $0 < \delta < 1.$
	Suppose $K\sse\R^n$ is a {closed} convex set
	with 
	\[\Vol(K) = 1-\delta.\]
	Then for all {$\eps > 0$}, there exists a convex set $L$ which is the intersection of
\[
\frac{1}{\delta} \cdot \pbra{\frac{n}{\eps}\log\pbra{\frac{1}{\delta}}}^{O(n)}
\]
halfspaces
	such that
	$\dG(K, L) \leq \eps\delta$; i.e.,~$\FC(K,\eps \delta) \leq \frac{1}{\delta} \cdot \pbra{\frac{n}{\eps}\log\pbra{\frac{1}{\delta}}}^{O(n)}$.
\end{theorem}

\noindent We mention that in the most interesting regime, when $\delta$ is smaller than some absolute constant (which can be taken to be 1/10), our proof of \Cref{thm:rel-err-dud} is self-contained and does not rely on the result of Bronstein--Ivanov (\Cref{thm:dudley}). 

To contrast \Cref{thm:rel-err-dud} with \Cref{thm:gaussian-dudley}, note that \Cref{thm:gaussian-dudley} would give an $\eps\delta$-approximation to $K$ using 
\[\pbra{\wt{O}(n^{5/4}/\eps\delta)}^{(n-1)/2}~\text{halfspaces}\]
where the $\wt{O}$ hides a logarithmic factor in $\log(1/\eps\delta)$. {The crucial difference between \Cref{thm:rel-err-dud} and \Cref{thm:gaussian-dudley} is that the dependence of  \Cref{thm:rel-err-dud} on $\delta$ is only ${\frac 1 \delta} \cdot \log(1/\delta)^{O(n)}$ rather than $(1/\delta)^{O(n)}$.
This can make a major difference if $\delta$ is very small; for example, if $\delta=1/2^{\sqrt{n}}$ and $\eps = 0.01$ (observe that with these parameters, the desired approximator of $K$ must correctly ``lop off'' at least 99\% of the $1/2^{\sqrt{n}}$ amount of mass that lies outside of $K$), then \Cref{thm:gaussian-dudley} would only give a bound of $2^{\tilde{O}(n^{3/2})}$ halfspaces whereas \Cref{thm:rel-err-dud} gives a bound of $2^{\tilde{O}(n)}$ halfspaces.}  Indeed, 
our $2^{\tilde{O}(n^{3/4})}$-halfspaces approximation of $\ell_p$ balls for $p\in(1,2)$, given in~\Cref{sec:ell-p}, will crucially rely on a savings of this sort that comes from the sharper parameters provided by \Cref{thm:rel-err-dud}.

\subsubsection{Proof Overview}
\label{subsec:rel-err-dud-prelims}

We set up some useful notation and give a high-level proof overview of \Cref{thm:rel-err-dud}. Throughout this section as well as the next, let $K\sse\R^n$ be a fixed convex set with volume $1-\delta$ as in the statement of \Cref{thm:rel-err-dud}.  We may assume without loss of generality that $K$ is closed, since $\Vol(K)$ is the same as the volume of the closure of $K$ for any convex set $K$.

We first remark that if $\delta \geq 1/10$ then the claimed bound is an immediate consequence of \Cref{thm:gaussian-dudley}, so we henceforth suppose that $\delta < 1/10.$

As in the proof of \Cref{thm:gaussian-dudley}, we may assume that $K$ is contained in a ball of sufficiently large radius, namely
\[K \sse B\pbra{R} \qquad\text{where}~R := \sqrt{n} + 2\sqrt{\ln\pbra{\frac{2}{\eps\delta}}},
\qquad\text{so}~\Vol(R) \geq 1 - {\frac {\eps \delta} 2},\]
and that the goal is to obtain an $(\eps\delta/2)$-approximation to $K$.\footnote{In particular, we may intersect $K$ with $B(R)$ to obtain a convex set $K'$ and aim for an $(\eps\delta/2)$-approximation to $K'$ as in \Cref{subsec:dudley}; for notational simplicity, however, we do not introduce the new set $K'$ in this section.} 

\begin{definition}[``Length'' function] \label{def:rel-dud-ell}
	Given a vector $v\in\S^{n-1}$, we define the function $\ell : \S^{n-1} \to \R$ as
	\[\ell(v) := \sup_{x\in K} \abra{x, v}.\]
\end{definition}

In other words, $\ell(v)$ is the {``length''} of the set $K$ along the direction $v$ {from the origin. We remark that since $\Vol(K) > 0.9 > 1/2$, the set $K$ must contain the origin.}  We further remark that the function $\ell: \S^{n-1}\to\R$ is sometimes called the \emph{Minkowski} or \emph{support functional} of the set $K$~\cite{TkoczNotes}. 

\begin{notation} \label{def:rel-dud-m}
	We write $m(t)$ for 
	\[m(t) := \Prx_{\by\sim\chi^2(n)}\sbra{\by \geq t^2}.\]
\end{notation}

Note that $m(t)$ is simply the tail mass of the $\chi^2(n)$ distribution beyond the point $t^2$. 
{We can also view it as the probability that a Gaussian draw $\bx \sim N(0,I_n)$ has $\|\bx\|^2 \geq t$, even conditioned on $\bx$ lying on any particular ray $\{rv: r \in \R_{\geq 0}\}$ for any fixed $v \in \S^{n-1}.$}
See \Cref{fig:RED-setup} for a schematic of the setup thus far. 

\begin{figure}
\centering
\usetikzlibrary{patterns}
\begin{tikzpicture}
	\filldraw[rotate=45,pattern=crosshatch,pattern color=black, opacity=0.25] (0,0) ellipse (3.25 and 1.75);
	\filldraw[rotate=45,fill=white] (0,0) ellipse (3 and 3);
	\draw[rotate=45] (0,0) ellipse (3.25 and 1.75);
	\draw (0,0) ellipse (0.5 and 0.5);
	\draw (0, 0) -- (-1.24, 1.24);
	\draw[-latex,dashed] (-1.275, 1.275) -- (-3,3);
	\node[fill=black,inner sep=0.75pt,circle] (0)  at (0,0) {};
	\node (bigB) at (2.75, -2.35) {\small $B(R)$};
	\node (smallB) at (0.75, -0.5) {\small $B(1)$};
	\node (vee) at (-3.1,3.1) {\small $v$};
	\node (ellV) at (-0.375, 0.8) {\small $\ell(v)$};
	\node (K) at (2.5, 2.5) {\small $K$};
	\end{tikzpicture}
\caption{Setup for the proof of \Cref{thm:rel-err-dud}. The length of the solid line is $\ell(v)$, and the {conditional} Gaussian mass on the dashed ray {(conditioned on the Gaussian lying on the ray from the origin in direction $v$)} corresponds to $m(\ell(v))$. The cross-hatched region corresponds to the (at most) $\eps\delta/2$ approximation error incurred by intersecting $K$ with $B(R)$.}
\label{fig:RED-setup}
\end{figure}
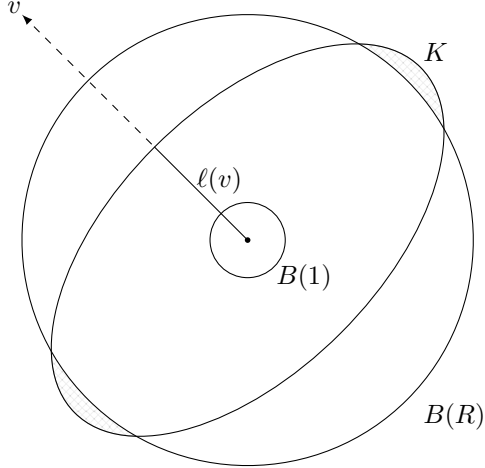

From \Cref{def:rel-dud-ell}, \Cref{def:rel-dud-m}, and the definition of the set $K$ in \Cref{thm:rel-err-dud}, we have that 
\[\Ex_{\bv\sim\S^{n-1}}\sbra{m\pbra{\ell(\bv)}} = \delta\]
since $\Vol(K) = 1-\delta$. Note that this uses the convexity of the set $K$ {and the fact that $K$ contains the origin.}

Finally, since $\Vol(K) = 1-\delta \geq 0.9$, we have that $B(1) \sse K$. To see this, note that if there exists $x\in B(1)$ such that $x\notin K$, then the separating hyperplane through $x$ for $K$ ``lops off'' at least $1 - \Phi(1) > 0.1$ of the Gaussian mass, and hence $\Vol(K) < 0.9$, which is a contradiction. We record this fact for convenience:

\begin{fact} \label{fact:K-contains-b1}
	We have that $B(1) \sse K$.
\end{fact}

We now turn to a high-level description of the approximator $L$. As in the proof of \Cref{thm:gaussian-dudley}, $L$ will be an outer approximator to the set $K$ (i.e. $K\sse L$). We will construct $L$ by placing tangent hyperplanes to $\partial K$ that ``lop off'' {at least a} $(1-\eps)$-fraction of the total $\delta$ amount of Gaussian mass that lies outside the set $K$. In more detail, call a direction $v\in \S^{n-1}$ ``good'' if $m(\ell(v)) \leq (\eps\delta/4)$, and ``bad'' otherwise. 
\begin{itemize}
	\item For all ``good'' directions, we will not worry about lopping off mass. Indeed, even if we failed to lop any probability mass off in all such directions, the total contribution to our overall error would be at most $(\eps\delta/4)$. 
	\item In order to handle ``bad'' directions, we first bucket them according to the value of $m(\ell(v))$ (which, as stated earlier, is {proportional to} the amount of Gaussian mass along $v$ outside of $K$). We will obtain an appropriate covering of {the points in} each bucket, and {we will} place tangent hyperplanes to $K$ at the points in our covering. A {geometric argument} (which will make use of the fact that $B(1)\sse K$), together with {suitable} tail bounds, will then give that the {total} error incurred from {all} ``bad'' directions is at most $(\eps\delta/4)$.
\end{itemize}
Putting everything together, we {will} obtain an $(\eps\delta/2)$-approximation to the set $K$; since we assumed $K \sse B(R)$, we already incurred $(\eps\delta/2)$-error, which overall gives  
\[\dG(K,L) \leq \eps\delta.\]

\subsubsection{Proof of \Cref{thm:rel-err-dud}}
\label{subsec:rel-err-dud-proof}

{As mentioned earlier, we may suppose that $\delta \leq 1/10.$}
We start by describing the bucketing alluded to in the proof overview.

\begin{definition}[Bucketing $\S^{n-1}$] \label{def:buckets}
	Let $\tau>0$ be a parameter that we will fix later. For $k\in\N$, we define 	\[\bucket(k) := \cbra{v\in\S^{n-1} : m(\ell(v)) \in \left((1+\tau)^{-k}, (1+\tau)^{-k+1}\right]}.\]
\end{definition}

We will fix $\tau$ (which will be a small quantity between $0$ and $1$) later in the course of the proof. Informally, $\bucket(k)$ is the set of directions for which the probability that a Gaussian vector in that direction lies outside $K$ is roughly $(1+\tau)^{-k}$. 

For large enough $k$, the amount of mass outside $K$ along the directions in $\bucket(K)$ will be small, and so intuitively we should be able to ignore such directions. (We called such directions ``good'' in the proof overview.) To make this precise, set a parameter $k^\ast$ to be
\[k^\ast := {\frac{1}{\tau}\ln\pbra{\frac{4}{\epsilon\delta}}}.\]
We then have for all $v\in \bucket(k)$ with $k > k^\ast$ that 
\[m(\ell(v)) \leq \pbra{1+\tau}^{-k+1} \leq \exp(-\tau k^\ast) \leq \frac{\eps\delta}{4}.\]
As discussed previously, we will only have to worry about the ``bad'' directions corresponding to 
\begin{equation} \label{eq:bad-buckets}
	\bigsqcup_{k=1}^{k^\ast} \bucket(k)
\end{equation}
and ignore the ``good'' directions; this will incur at most $\eps\delta/4$ error. We will use the following lemma to {separately handle each bucket $\bucket(k)$ for $k \leq k^\ast$}.

\begin{lemma} \label{lemma:bucket-covering}
	Suppose $A\sse\S^{n-1}$ and $\theta \in (0,\pi/2)$. There exists a set $S_A(\theta) \sse A$ consisting of at most ${(12/\theta)^n}$ points such that for all $x \in A$, there exists a point $y \in S_A(\theta)$ such that 
	\[\angle(x,y) = \arccos\abra{x,y} \leq \theta.\]
\end{lemma}

\begin{proof}
	Note that $\angle(x,y) \leq \theta$ for $x,y\in\S^{n-1}$ if and only if $\|x-y\| = 2\sin(\theta/2)$. It thus suffices to construct a $2\sin(\theta/2)$-net (cf.~Section 4.2 of~\cite{vershynin2018high}) of $A$. From Exercise~4.2.10 of~\cite{vershynin2018high},  the size of such a net is upper bounded by the size of a $\sin(\theta/2)$-net for $\S^{n-1}$, which in turn is at most 
	\[\pbra{\frac{3}{\sin\pbra{\theta/2}}}^n.\]
	This bound is standard; see, for example, Corollary 4.2.13 of~\cite{vershynin2018high}. Finally, since $\sin(\cdot)$ is concave on the interval $[0,\pi/2]$, we have that 
	\[\sin\pbra{\frac{\theta}{2}} \geq \frac{\theta}{4},\]
	and so we have that the desired net consists of at most
	\[\pbra{\frac{12}{\theta}}^n~\text{points},\]
	which completes the proof.
\end{proof}

We will construct the approximator $L$ by placing tangent hyperplanes to $K$ at the points in $S$ where we define
\[
	S := \bigsqcup_{k=1}^{k^\ast} S_{\bucket(k)}(\theta^\ast)
\]
where $S_{\bucket(k)}(\theta^\ast)$ is as in \Cref{lemma:bucket-covering} for an appropriate choice of $\theta^\ast$. It follows that the total number of halfspaces in our approximator $L$ is at most 
\begin{equation} \label{eq:num-halfspaces}
	|S| \leq k^\ast \cdot {\pbra{\frac{12}{\theta^\ast}}^n} = \frac{1}{\tau}\log\pbra{\frac{4}{\eps\delta}}\cdot{\pbra{\frac{12}{\theta^\ast}}^n},
\end{equation}
where $\tau$ and $\theta^\ast$ are parameters that we will set below.

\paragraph{Setting Parameters.} Setting the parameters $\tau$ and $\theta^\ast$ fixes our approximator $L$. We will take 
\begin{equation} \label{eq:setting-tau-theta}
	\tau := \frac{\eps\delta}{8}
	\qquad\text{and}\qquad 
	\theta^\ast := \frac{\eps^2}{64}\pbra{R\pbra{2 + \frac{16}{\epsilon}}\pbra{2R^2\pbra{2 + \frac{16}{\eps}} -\frac{n}{4}}}^{-1}.
\end{equation}
Our choices above are dictated by the error analysis below; we note that it follows from \Cref{eq:setting-tau-theta,eq:num-halfspaces} that $L$ has the desired number of halfspaces. Before proceeding, we introduce the following notation:
%
\begin{notation}
	For $k\in[k^\ast]$, we define 	
	\[
	\ell(k)_{\max} := \sup_{v\in\bucket(k)} \ell(v) 
	\qquad\text{and}\qquad  
	\ell(k)_{\min} := \inf_{v\in\bucket(k)} \ell(v).
	\]
\end{notation}

%
%
%

\paragraph{Error Analysis.} We will analyze the error in approximating $K$ by $L$ on a direction-by-direction basis; recall that we only need to worry about the ``bad'' directions (\Cref{eq:bad-buckets}). Fix a ``bad'' direction $v\in\S^{n-1}$, and suppose $v \in \bucket(k)$ (where $k \leq k^\ast$). Consider the two-dimensional setup {(corresponding to the two-dimensional plane $\spann\{O,U,W\}$)} as in \Cref{fig:RED-error-analysis}: 
\begin{itemize}
	\item Let $O$ denote the origin $0^n$;
	\item Let $V \in \partial K$ be the point of intersection of the ray $\{tv : t \geq 0\}$ and $\partial K$; 
	\item Let $U \in \partial K$, {$U \in S_{\bucket(k)}(\theta^\ast)$} be {a} point in $\bucket(k)$ such that $\theta := \angle(\vec{OL},\vec{ON}) \leq \theta^\ast$; and 
	\item By construction, we put down a halfspace tangent to $K$ at $U$; let $W$ be the point of intersection of this halfspace with the ray $\{tv : t \geq 0\}$.
\end{itemize}

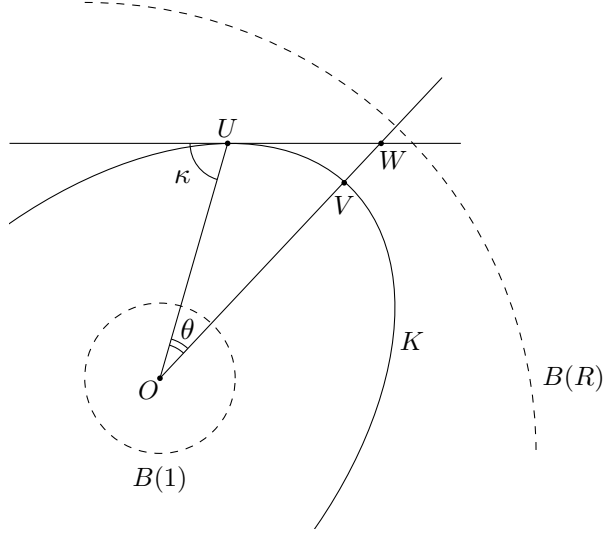
\begin{figure}
\centering
\begin{tikzpicture}
	\node[fill=black,inner sep=0.75pt,circle] (ori)  at (1, 1) {};
	\begin{scope}
	\clip(-1,-1) rectangle (6, 6);
	\draw[rotate=45] (0,0) ellipse (5 and 3);
	\draw[dashed] (ori) ellipse (1 and 1);
	\end{scope}
	\draw[dashed] (0, 6) to [out = 0, in = 90] (6, 0);
	\node (tline) at (4.75,5) {};
	\node (origin) at (1,1) {};
	\node (suline-right) at (5,4.125) {};
	\node (suline-left) at (-1,4.125) {};
	\node (osseg) at (1.9, 4.125) {};
	\pic [draw, -, "$\theta$", double, double distance = 2pt, angle eccentricity=1.5] {angle = tline--origin--osseg};
	\pic [draw, -, "$\kappa$", angle eccentricity=1.5] {angle = suline-left--osseg--origin};
	\draw (1, 1) -- (1.9, 4.125);
	\draw (-1, 4.125) -- (5, 4.125);
	\draw (1,1) -- (4.75, 5);
	\node (O) at (0.85,0.85) {\small $O$};
	\node (U) at (1.9, 4.325) {\small $U$};
	\node (W) at (4.1,3.9) {\small $W$};
	\node (V) at (3.45,3.3) {\small $V$};
	\node (smallBall) at (1, -0.35) {\small $B(1)$};
	\node (bigBall) at (6.5, 1) {\small $B(R)$};
	\node (kay) at (4.35, 1.5) {\small $K$};
	\node[fill=black,inner sep=0.75pt,circle] (Udot)  at (1.9, 4.125) {};
	\node[fill=black,inner sep=0.75pt,circle] (Udot)  at (3.94, 4.125) {};
	\node[fill=black,inner sep=0.75pt,circle] (Udot)  at (3.45, 3.6) {};
\end{tikzpicture}
\caption{A two-dimensional setup to analyze the error of our approximator. Here $V \in \bucket(k)$ and $U \in S_{\bucket(k)}(\theta^\ast)$ and so $\angle(\vec{OV},\vec{OU}){=\theta} \leq \theta^\ast$}
\label{fig:RED-error-analysis}
\end{figure}

\newcommand{\error}{\mathrm{error}}

With this setup in hand, note that the error incurred along direction $v$ is given by 
\[\error(v) := m\pbra{\|V\|} - m\pbra{\|W\|}.\]
We will next establish that this quantity is at most $\eps\delta/4$, which will imply that the total error incurred from all the ``bad'' directions is at most $\eps\delta/4$. Note that 
\begin{align}
	\error(v) &= m\pbra{\|V\|} - m\pbra{\|W\|} \nonumber\\  
	&\leq m\pbra{\ell(k)_{\min}} - m\pbra{\|W\|} \nonumber\\
	&= \pbra{m\pbra{\ell(k)_{\min}} - m\pbra{\ell(k)_{\max}}} + \pbra{m\pbra{\ell(k)_{\max}} - m\pbra{\|W\|}}. \label{eq:okra}
\end{align}
It suffices to show that \Cref{eq:okra} is upper bounded by $\eps\delta/4$. Indeed, recall that we incurred (a) $\eps\delta/2$ error from ``capping'' the set $K$ by intersecting it with $B(R)$; (b) $\eps\delta/4$ error from the ``bad'' directions from $\bucket(k)$ for $k > k^\ast$; and finally (c) $\eps\delta/4$ error from all of the ``good'' directions from \Cref{eq:okra} (as we will establish in the remainder of the argument); this will complete the proof of~\Cref{thm:rel-err-dud}.

We will bound each term in \Cref{eq:okra} separately. By definition of the buckets (\Cref{def:buckets}), we have that 
\begin{equation} \label{eq:bucket-width}
	m\pbra{\ell(k)_{\min}} - m\pbra{\ell(k)_{\max}} \leq \tau(1+\tau)^{-k} \leq \tau = \frac{\eps\delta}{8}.
\end{equation}

In order to bound the second term in \Cref{eq:okra}, recall that by Brahmagupta's formula (cf.~\Cref{fig:RED-error-analysis}), we have 
\[\frac{\sin(\pi-\kappa)}{\|W\|}  = \frac{\sin(\kappa - \theta)}{\|U\|}.\]
This can be rewritten as 
\begin{equation} \label{eq:sin-ratios}
	\|W\| = \frac{\sin\kappa}{\sin(\kappa -\theta)} \cdot \|U\| \leq \frac{\sin\kappa}{\sin(\kappa -\theta)} \cdot \ell(k)_{\max}.
\end{equation}
Recalling~\Cref{fact:K-contains-b1}, since $B(1)\sse K \sse B(R)$, we have that
\begin{gather*}
		\tan\kappa = \frac{\dist\pbra{O, \overleftrightarrow{UW}}}{\dist\pbra{U, O\perp \overleftrightarrow{UW}}} \geq \frac{1}{R} 
\end{gather*}
where we write $\dist(O, \overleftrightarrow{UW})$ to mean the distance between the origin and the line $\overleftrightarrow{UW}$, which is guaranteed to be at least $1$ thanks to the convexity of $K$ and the fact that $B(1)\sse K$; and $\dist(U, O\perp \overleftrightarrow{UW})$ is the distance between $U$ and the point on the line $\overleftrightarrow{UW}$ closest to $O$, which must be at most $\dist(O,U)$ (since the hypotenuse is the longest side in a right triangle), which in turn is at most $R$ by out setup. (See~\Cref{fig:RED-error-analysis}.)
Returning to \Cref{eq:sin-ratios}, the above lower bound on $\tan\kappa$ implies that 
\[\frac{\sin\kappa}{\sin(\kappa - \theta)} = \frac{\sin\kappa}{\sin\kappa\cos\theta-\cos\kappa\sin\theta} = \frac{1}{\cos\theta - \frac{\sin\theta}{\tan\kappa}} \leq \frac{1}{\cos\theta - R\sin\theta}.\]
For $\theta \geq 0$, recalling the standard trigonometric inequalities 
\[\sin\theta \leq \theta \qquad\text{and}\qquad \cos\theta \geq 1 - \frac{\theta^2}{2},\]
we thus get that 
\[\frac{\sin\kappa}{\sin(\kappa - \theta)} \leq \frac{1}{\cos\theta - R\sin\theta} \leq \frac{1}{1 - R\theta - \frac{\theta^2}{2}} \leq 1 + 2R\theta \leq 1 + 2R\theta^\ast\]
where the penultimate inequality can be readily verified {from the definitions of $R$ and $\theta$}, and the final inequality relies on the fact that $\theta \leq \theta^\ast$ since $U \in S_{\bucket(k)}(\theta^\ast)$ and $V\in\bucket(k)$.

Plugging the above bound back into~\Cref{eq:sin-ratios}, we get that 
\[\|W\| \leq \ell(k)_{\max}\pbra{1 + 2R\theta^\ast}.\]
As $m$ (cf. \Cref{def:rel-dud-m}) is a decreasing function, we have that 
\begin{align}
	m\pbra{\ell(k)_{\max}} - m\pbra{\|W\|} &\leq m\pbra{\ell(k)_{\max}} - m\pbra{\ell(k)_{\max}\pbra{1 + 2R\theta^\ast}} \nonumber \\
	&= m\pbra{\ell(k)_{\max}}\pbra{1- \frac{m\pbra{\ell(k)_{\max}\pbra{1 + 2R\theta^\ast}}}{m\pbra{\ell(k)_{\max}}}} \nonumber \\
	&\leq \delta\cdot \pbra{1- \frac{m\pbra{\ell(k)_{\max}\pbra{1 + 2R\theta^\ast}}}{m\pbra{\ell(k)_{\max}}}} \label{eq:from-sin-ratios-to-error}
\end{align}
where \Cref{eq:from-sin-ratios-to-error} relies on the observation that $m(\ell(v)) \leq \delta$ for any direction $v \in \S^{n-1}$. To see this, note that if $tv \in \partial K$ is such that $m(t)> \delta$, then the supporting hyperplane tangent to $K$ at $tv$ will ``lop off'' strictly greater than $\delta$ mass. Consequently, $\vol(K) < 1-\delta$, which is a contradiction. 

Now, note that it suffices to show that 
\begin{equation} \label{eq:RED-final-goal}
	\frac{m\pbra{\ell(k)_{\max}\pbra{1 + 2R\theta^\ast}}}{m\pbra{\ell(k)_{\max}}} \geq 1 -{\frac{\eps}{8}}
\end{equation}
in order to complete the proof. 
The remainder of the proof establishes \Cref{eq:RED-final-goal}. We split into two cases depending on $\ell(k)_{\max}$:

\paragraph{Case 1: $\ell(k)_{\max} \geq {\sqrt{n-1}}$.}.
Let $\Delta >0$ be a parameter that we will set later. We have
\begin{equation} \label{eq:bitter-melon}
	\frac{m\pbra{\ell(k)_{\max}(1 + 2R\theta^*)}}{m(\ell(k)_{\max})} \geq \frac{m\pbra{\ell(k)_{\max}(1 + 2R\theta^*)} - m\pbra{\ell(k)_{\max}(1 + (2+\Delta)R\theta^*)}}{m(\ell(k)_{\max}) - m\pbra{\ell(k)_{\max}(1 + (2+\Delta)R\theta^*)}}.	
\end{equation}
Note that the right hand side of the above equation is 
\begin{align}
	\text{R.H.S. of~\eqref{eq:bitter-melon}}&=\pbra{\int_{\ell(k)_{\max}(1+2R\theta^\ast)}^{\ell(k)_{\max}(1+(2+\Delta)R\theta^\ast)} \mathrm{pdf}_{\chi(n)}(x)\,dx}\pbra{\int_{\ell(k)_{\max}}^{\ell(k)_{\max}(1+(2+\Delta)R\theta^\ast)} \mathrm{pdf}_{\chi(n)}(x)\,dx}^{-1} \nonumber \\
	&\geq \frac{\ell(k)_{\max}\Delta R\theta^\ast \cdot \mathrm{pdf}_{\chi(n)}\pbra{\ell(k)_{\max}(1+(2+\Delta)R\theta^\ast)}}{\ell(k)_{\max}(2+\Delta)R\theta^\ast \cdot \mathrm{pdf}_{\chi(n)}(\ell(k)_{\max})} \label{eq:budweiser}\\
	&=\pbra{\frac{\Delta}{2+\Delta}}\frac{\mathrm{pdf}_{\chi(n)}\pbra{\ell(k)_{\max}(1+(2+\Delta)R\theta^\ast)}}{\mathrm{pdf}_{\chi(n)}(\ell(k)_{\max})} \nonumber \\
	&\geq \pbra{\frac{\Delta}{2+\Delta}}\underbrace{\pbra{(1+(2+\Delta)R\theta^\ast)^{n-1} \exp\pbra{\frac{-R^2}{2}\pbra{2(2+\Delta)R\theta^\ast + (2+\Delta)^2R^2{\theta^\ast}^2}}}}_{=:\Upsilon(n,R,\theta^\ast)},  \label{eq:miller}
\end{align}
where (\ref{eq:budweiser}) used the fact that {the $\mathrm{pdf}_{\chi(n)}(x)$ is decreasing in $x$ for $x\geq \sqrt{n-1}$}
and (\ref{eq:miller}) used the formula for the density function of the $\chi(n)$-distribution, which is given by 
\[
\mathrm{pdf}_{\chi(n)}(x) = 
\begin{cases}
	\frac{1}{2^{n/2-1}\Gamma(n/2)} x^{n-1}\exp\pbra{-\frac{x^2}{2}} & \mathrm{if}~x \geq 0\\
	0 & \mathrm{if}~x < 0
\end{cases},
\]
 and the fact that $\ell(k)_{\max} \leq R$. (This is because we truncated $K$ by intersecting it with $B(R)$.)

We will take 
\[\Delta := \frac{16}{\epsilon},\]
as a consequence of which, \Cref{eq:RED-final-goal} {follows from} showing that 
\[\Upsilon(n,R,\theta^\ast) \geq 1-\frac{\eps^2}{64}.\]
In particular, using the inequality $1+x\leq e^x$, it suffices to show that 
\[\Upsilon(n,R,\theta^\ast) \geq \exp\pbra{-\frac{\eps^2}{64}}.\]
We have that 
\begin{align*}
	\Upsilon(n,R,\theta^\ast) & \geq \pbra{(1+(2+\Delta)R\theta^\ast)^{n-1} \exp\pbra{-2R^3\theta^\ast(2+\Delta)^2}} \\
	& \geq \exp\pbra{\theta^\ast\pbra{\pbra{\frac{2+\Delta}{4}}nR - 2R^3(2+\Delta)^2}} \\ 
	& = \exp\pbra{-\frac{\eps^2}{64}},
\end{align*}
where the first inequality uses $R\theta^\ast \ll1$, the second inequality uses $1+x \geq e^{x/2}$ for $x\in[0,1]$, and the third inequality uses the choice of $\theta^\ast$, which we recall was
\[\theta^\ast := \frac{\eps^2}{64}\pbra{R\pbra{2 + \frac{16}{\epsilon}}\pbra{2R^2\pbra{2 + \frac{16}{\eps}} -\frac{n}{4}}}^{-1}.\]

\paragraph{Case 2: $\ell(k)_{\max} < {\sqrt{n-1}}$.} First, note that $\mathrm{pdf}_{\chi(n)}(x)$ is at most $1$ for all $x$. It follows that 
\begin{equation} \label{eq:RED-case-2-interval-mass}
	\int_{\ell(k)_{\max}}^{\ell(k)_{\max}(1+2R\theta^\ast)} \mathrm{pdf}_{\chi(n)}(x)\,dx \leq 2\ell(k)_{\max}R\theta^\ast \leq 2R\theta^\ast\sqrt{n}.
\end{equation}
Furthermore, since $\ell(k)_{\max} \leq {\sqrt{n-1}}$ {which is the mode of the $\chi(n)$-distribution} it follows that 
\[m(\ell(k)_{\max}) \geq \frac{1}{2}.\]
This, together with~\Cref{eq:RED-case-2-interval-mass}, lets us write
\begin{align*}
	\frac{m\pbra{\ell(k)_{\max}\pbra{1 + 2R\theta^\ast}}}{m\pbra{\ell(k)_{\max}}} &\geq \frac{m\pbra{\ell(k)_{\max}} - 2R\theta^\ast\sqrt{n}}{m\pbra{\ell(k)_{\max}}}\\
	&\geq 1 - 4R\theta^\ast\sqrt{n}\\
	&\geq 1 - \frac{\eps}{8},
\end{align*}
establishing \Cref{eq:RED-final-goal} (the final inequality above is straightforward to verify from our choice of $\theta^\ast$).

Putting both cases together completes the proof of the theorem.
\qed 

%
%


\section{Improved Approximation for the $\ell_2$ Ball}
\label{sec:nazarov-ub}

For the Hausdorff and Lebesgue distance metrics (cf.~\Cref{subsec:distance-metrics}), the $\ell_2$ ball is often an extremal example for known upper bounds {on the vertex or facet complexity of polyhedral approximators} ~\cite{Dudley1974,Bronshteyn1976,ball1997elementary,LSW06}. In this section, we show that---perhaps surprisingly---the $\sqrt{n}$-radius $\ell_2$ ball  $B(\sqrt{n})$, which has Gaussian volume $1/2 + o_n(1)$, can be approximated to within any constant error by an intersection of only $2^{O(\sqrt{n})}$ many halfspaces. This is a substantial improvement on the  exponential-in-$n$ approximation upper bounds obtained in \Cref{sec:generic-ubs} for general convex bodies.

Throughout this section, we will write ${B_2} :=B(\sqrt{n})$ for ease of notation {(the ``2'' subscript is because we are dealing with the $\ell_2$ ball)}.

\begin{theorem} \label{thm:de-nazarov}
Let $0<\eps<c$ {for some sufficiently small absolute constant $c$.} Then there exists a polytope $K$ which is the intersection of 
	\[ s = \exp\pbra{\Theta\pbra{\sqrt{n}\cdot\frac{1}{\epsilon}\log\pbra{\frac{1}{\eps}}}}~\text{halfspaces}\]
	such that 
	$\dG(B_2, K) \leq \eps.$
\end{theorem}

\Cref{thm:de-nazarov} is inspired by Theorem~2.1 of O'Donnell and Wimmer~\cite{o2007approximation}, which shows that the $n$-bit majority function $\MAJ_n : \zo^n\to\zo$, defined as 
\[\MAJ_n(x) := \mathbf{1}\cbra{\sumi x_i \geq n/2},\]
can be approximated by a monotone CNF formula of size $2^{O(\sqrt{n})}$. O'Donnell and Wimmer's construction is probabilistic and {bears a close resemblance to} Talagrand's random CNF~\cite{Talagrand:96}. Our approach for approximating $B$  
employs a modification of a probabilistic construction of a convex body due to Nazarov~\cite{Nazarov:03}. 
Looking ahead, in \Cref{sec:conv-inf} we will show that~\Cref{thm:de-nazarov} is tight for constant $\eps$; more precisely, we will show that any $\eps$-approximation to $B_2$
must have $2^{\Omega(\sqrt{n})}$ facets for a suitable small constant $\eps>0$. 

To prove \Cref{thm:de-nazarov}, we begin by defining a suitable distribution over intersections of randomly chosen halfspaces:
\begin{definition} \label{def:de-nazarov}
	For $w, s > 0$, we write $\Naz(w,s)$ to be the distribution over $s$-facet polytopes in $\R^n$
where draw from $\Naz(w,s)$ is obtained as follows:
	\begin{enumerate}
		\item For $i\in[s]$, draw i.i.d. $\g{i}\sim N(0,I_n)$ and let $\bH_i$ denote the halfspace
		\[\bH_i := \cbra{x\in\R^n : \langle x, \g{i}\rangle\leq w }.\]
		\item Output the convex set $\bK := \bigcap_{i=1}^s \bH_i$.
	\end{enumerate}
\end{definition}

\usetikzlibrary{shapes.geometric}
\begin{figure}
	\centering
	\begin{tikzpicture}

		\node (v1) at (3,0) {};
		\node (v2) at (1.5,2.598) {};
		\node (v3) at (-1.5,2.598) {};
		\node (v4) at (-3,0) {};
		\node (v5) at (-1.5,-2.598) {};
		\node (v6) at (1.5,-2.598) {};
	
		\def\A{(2.3,0) to 
		(1.5,1) to (0.7,1.9) to 
		(-1.15,1.8) to (-2.2,-0.05) to (-1.3,-1.8) to (1.3,-1.8) to (2.3,0);}
		\def\B{(v1) [out=180,in=240] to (v2) [out=240,in=-60] to (v3) [out=-60,in=0] to (v4) [out=0,in=60] to (v5) [out=60,in=120] to (v6) [out=120,in=180] to (v1);}
		
		\filldraw[pattern=crosshatch,opacity=0.125,rotate=90] \A; \filldraw[pattern=crosshatch,opacity=0.125] \B;	
		
		\begin{scope}
			\clip\B;
			\fill[white,even odd rule, rotate=90] \A;
		\end{scope}
		
		\draw[rotate=90] \A; \draw \B;
		
		\node[fill, inner sep=1pt,circle] (ori) at (0,0) {};
		\draw (ori) -- (v2);
		\draw[rotate=90] (ori) -- (0,-1.8);
		\node (rad) at (0.2, 1) {\small $\sqrt{n}$};
		\node (naz) at (1, -0.25) {\small $\approx n^{1/4}$};
	\end{tikzpicture}
	\caption{A cartoon of how a polytope drawn from $\Naz(w,s)$, for suitable $s = 2^{\Theta(\sqrt{n})}$, $w \approx n^{3/4}$, approximates the radius-$\sqrt{n}$ ball in $\R^n$. Our depiction of ${B_2} = B(\sqrt{n})$ is inspired by Milman's ``hyperbolic'' drawings of high-dimensional convex sets~\cite{milman1998surprising}.
	}
	\label{fig:de-nazarov}
\end{figure}
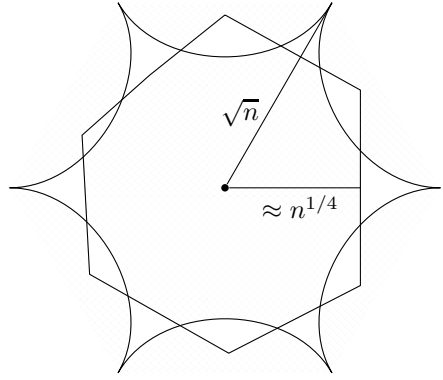

With \Cref{def:de-nazarov} in hand, we turn to the proof of \Cref{thm:de-nazarov}:

\begin{proof}[Proof of \Cref{thm:de-nazarov}]
 We will show that there exists an outcome $K$ in the support of $\Naz(w,s)$ for an appropriate choice of parameters $w$ and $s$ that has the desired properties.
	
	To show this, it suffices to show that 
	\[\Ex_{\bK\sim\Naz(w,s)}\sbra{\Prx_{\bx\sim N(0,I_n)}\sbra{\bK(\bx) \neq B_2(\bx)}} \leq \eps,\]
	which, by commuting the order of integration, is equivalent to showing that 
	\begin{equation} \label{eq:naz-goal}
		\Ex_{\bx\sim N(0,I_n)}\sbra{\Prx_{\bK\sim\Naz(w,s)}\sbra{\bK(\bx)\neq B_2(\bx)}} \leq \epsilon.
	\end{equation}
	\Cref{eq:naz-goal} allows us to control the error on an ``$x$-by-$x$ basis.'' We set parameters
	\begin{equation} \label{eq:dindout}
	\din := \sqrt{n} - \frac{\eps}{4} \qquad\text{and}\qquad \dout := \sqrt{n} + \frac{\eps}{4}.
	\end{equation}
	We will show that when $\|x\| \leq \din$ or $\|x\| \geq \dout$, 
	\begin{equation} \label{eq:new-naz-goal}
		\Prx_{\bK\sim\Naz(w,s)}\sbra{\bK(x) \neq B_2(x)} \leq \frac{\eps}{2}.
	\end{equation}
	Since the Gaussian volume of the annulus $\{y \in \R^n : \din\leq \|y\| \leq \dout\}$ is at most $\frac{\eps}{2}$ {(this is an easy consequence of the standard fact that for $n >1$, the pdf of the $\chi^2(n)$-distribution is everywhere at most 1)}, this establishes \Cref{eq:naz-goal} which in turn completes the proof.
	
	Note that by construction (\Cref{def:de-nazarov}), for any fixed $x\in\R^n$ we have that 
	\[\Prx_{\bK\sim\Naz(w,s)}\sbra{x \in \bK} = \Phi\pbra{\frac{w}{\|x\|}}^s\]
	where $\Phi:\R \to [0,1]$ denotes the cumulative density function of the univariate Gaussian distribution $N(0,1)$. To see this, note that 
	\[\langle x, \g{j} \rangle = \sumi x_i\g{j}_i \sim N(0, \|x\|^2)\quad\text{and so}\quad \Prx_{\g{1},\ldots,\g{s}}\sbra{\langle x, \g{j} \rangle \leq w~\text{for all}~j\in[s]} = \Phi\pbra{\frac{w}{\|x\|}} ^s\]
	due to independence.
	
	This observation informs our setting of the parameters $w$ and $s$. We take $w$ to satisfy the equation
	\begin{equation} \label{eq:naz-w-choice}
		\frac{1-\Phi(w/\dout)}{1-\Phi(w/\din)} = \frac{4}{\eps}\,{\ln}\pbra{\frac{4}{\eps}}
	\end{equation}
	and take $s$ to be 
	\begin{equation} \label{eq:naz-s-choice}
		s:= \frac{\eps}{4}\pbra{1-\Phi\pbra{\frac{w}{\din}}}^{-1}.
	\end{equation}
(We will argue that there is a valid solution to \Cref{eq:naz-w-choice} later on in the proof.) Given \Cref{eq:naz-w-choice,eq:naz-s-choice},
	
	\begin{itemize}
		\item For $x\in\R^n$ with $\|x\|\leq\din$, we have
	\[\Prx_{\bK\sim\Naz(w,s)}\sbra{x\in\bK} = \Phi\pbra{\frac{w}{\|x\|}}^s \geq \Phi\pbra{\frac{w}{\din}}^s = \pbra{1-\frac{\epsilon}{4s}}^s \geq 1-\frac{\epsilon}{4}\]
	where the first inequality relies on the fact that $\Phi(\cdot)$ is an increasing function, the following equality relies on our choice of $s$ in \Cref{eq:naz-s-choice}, and the final inequality makes use of the fact that $(1-a)^b \geq 1-ab$. As $\|x\|\leq \din$ implies that $x\in B_2$, we pick up at most $\eps/4$ error on such points. 
	
		\item For $x\in\R^n$ with $\|x\|\geq\dout$, we have 
	\begin{align*}
		\Prx_{\bK\sim\Naz(w,s)}\sbra{x\in\bK} = \Phi\pbra{\frac{w}{\|x\|}}^s &\leq \Phi\pbra{\frac{w}{\dout}}^s = \pbra{1 - \frac{4}{\epsilon}\pbra{1 - \Phi\pbra{\frac{w}{\din}}} {\ln}\pbra{\frac{4}{\eps}}}^s \\ 
		&= \pbra{1 - \frac{1}{s}\,{\ln}\pbra{\frac{4}{\eps}}}^s\\
		&\leq \frac{\epsilon}{4}
	\end{align*}
	where we once again used the fact that $\Phi(\cdot)$ is increasing in the first inequality. The following equalities follow from rearranging \Cref{eq:naz-w-choice,eq:naz-s-choice}, and the final inequality is due to the fact that $(1+x)\leq \exp(x)$. As $x\geq \dout$ implies that $x\notin B_2$, we pick up at most $\eps/4$ error on such points as well.
	\end{itemize}
	In particular, the above establishes \Cref{eq:new-naz-goal}, which in turn establishes \Cref{eq:naz-goal} and so there exists a set $K$ in the support of $\Naz(w,s)$ that $\eps$-approximates $B_2$. 
	
	It remains to argue that there exists a valid solution to \Cref{eq:naz-w-choice}, and to bound the number of facets of $K$ (i.e. the parameter $s$); we start with the former. Using standard Gaussian tail bounds (\Cref{prop:gaussian-tails}), 
{provided that $w/{\din},w/{\dout}=\Omega(1)$ (which holds with room to spare; below we will see that these quantities are $\Omega(n^{1/4})$),} 	
	we can write \Cref{eq:naz-w-choice} as
	\[\exp\pbra{\frac{w^2}{2}\pbra{\frac{1}{\din^2} - \frac{1}{\dout^2}}}\cdot \Theta(1) =  \frac{4}{\eps}\,{\ln}\pbra{\frac{4}{\eps}}.\]
	Taking logarithms on both sides, {for $\eps$ at most some sufficiently small constant we get that}
%
	\[w^2 = \Theta\pbra{{\ln}\pbra{\frac{4}{\eps}}+ {\ln\ln}\pbra{\frac{4}{\epsilon}}}\cdot \frac{(\din\dout)^2}{\dout^2-\din^2}.\]
	Recalling our choices of $\dout$ and $\din$, we have 
	\[\frac{(\din\dout)^2}{\dout^2-\din^2} = \frac{\pbra{n - \eps^2/16}^2}{\eps\sqrt{n}} = \Theta\pbra{\frac{n^{3/2}}{\epsilon
	}}.\]
	Plugging this back into the previous expression and taking square roots on both sides gives 
	\begin{equation} \label{eq:wvalue}w := \Theta\pbra{{n^{3/4}}\sqrt{\frac{1}{\epsilon}\pbra{ {\ln}\pbra{\frac{4}{\eps}}+  {\ln\ln}\pbra{\frac{4}{\epsilon}}}}}.\end{equation}
	
	This lets us bound $s$, which is the number of facets of $K$. From \Cref{eq:naz-s-choice} and once again using standard tail bounds (\Cref{prop:gaussian-tails}), we have
	\begin{align*}
		s = \frac{\eps}{4}\pbra{1-\Phi\pbra{\frac{w}{\din}}}^{-1} 
		& {\leq} \frac{\eps}{4}\exp\pbra{\frac{w^2}{2\din^2}}\pbra{\frac{\din}{w} - \frac{\din^3}{w^3}}\\
		&\leq \exp\pbra{\Theta\pbra{\frac{w^2}{\din^2}}}.
	\end{align*}
	From \Cref{eq:dindout,eq:wvalue} we get that
	\[\frac{w}{\din} = \Theta\pbra{n^{1/4}\sqrt{\frac{1}{\epsilon}\pbra{\log\pbra{\frac{4}{\eps}}+ \log\log\pbra{\frac{4}{\epsilon}}}}},\]
	from which the claimed bound on $s$ is immediate.	
\end{proof}


\section{Improved Approximations for $\ell_p$ Balls for {$p\in[1,2)$}}
\label{sec:ell-p}

The main result of this section is a proof of the existence of approximators with a sub-exponential number of facets for the $\ell_p$ balls $B_p$, where $1 \leq p < 2$:

\begin{theorem} \label{thm:ell-p-approx}
For $1 \leq p < 2$, let $B_p$ denote the origin-centered $\ell_p$ ball in $\R^n$  
{defined in \Cref{eq:Bp} below}, which has $\Vol(B_p)={\frac 1 2} \pm o_n(1).$  
Let $\paramset < \tau$ where $\tau\in(0,1)$ is some absolute constant.
 Then for $1 \leq p < 2$, there exists a polytope $K_p$ which is the intersection of 
	\begin{equation} \label{eq:ell-p-bound} {s = 
	\exp\pbra{\Theta\pbra{\pbra{{\frac {\log^{9/4}(1/\eps)}{\eps^4}}}\cdot \log \pbra{{\frac n \eps}} \cdot n^{3/4}}}
	~\text{halfspaces}}
	\end{equation}
	such that 
	$\dG(B_p, K_p) \leq \eps.$
In the special case of $p=1$, we have a slightly improved bound of 
\begin{equation} \label{eq:ell-one-bound}
s=\exp\pbra{\Theta\pbra{\pbra{\frac{\log(1/\epsilon)}{\epsilon}}^{3/2} n^{3/4}}}~\text{halfspaces}
\end{equation}
(see~\Cref{remark:ell-1-done}).
\end{theorem}

{We remark that the approximation results of \Cref{thm:ell-p-approx} can be extended to the $\ell_p$ balls of volume exactly $1/2$; we work with $B_p$ instead because it is more convenient.}
In \Cref{subsec:cramer-bounds} we give some basic setup and recall some \Cramer-type tail bounds on sums of  random variables which we will use; the key property of these bounds which we will exploit is that they give tight \emph{multiplicative} control on the relevant tails.
We will then prove~\Cref{thm:ell-p-approx} in~\Cref{subsec:ell-p-balls}.

\subsection{\Cramer-Type Large Deviations Bounds and General Setup}
\label{subsec:cramer-bounds}

Roughly speaking, ``\Cramer-type'' tail bounds are results which state that under suitable conditions, a sum $\bS$ of random variables for which $\E[\bS]=0$ and $\Var[\bS]=1$ will have tail probabilities $\Pr[\bS \geq x]$ which are \emph{multiplicatively} close to the corresponding Gaussian tail bound $\Pr_{\bg \sim N(0,1)}[\bg \geq x].$
The most standard results of this type are for sums of \emph{independent} random variables, as in the following theorem:

\begin{theorem}[Theorem~1 from Chapter VIII of \cite{petrov2012sums}] \label{thm:petrov}
Let $\bX_1, \ldots, \bX_n$ be i.i.d. random variables that satisfy \Cramer's condition, 
i.e., there exists constant $H>0$ such that $\mathbf{E} [e^{H\bX_1}] <\infty$.

Let us assume that the variables $\bX_i$ are centered --- i.e. $\mathbf{E}[\bX_1]=0$ and $\Var[\bX_1] = \sigma^2$. Let $\bS_n := \sum_{i=1}^n \bX_i$ and define $F_n(x) := \Pr[\bS_n \le x \sigma \sqrt{n}]$. If $x >0$ and $x = o(\sqrt{n})$, then 
\[
\frac{1-F_n(x)}{1-\Phi(x)} =\exp \bigg( \frac{x^3}{\sqrt{n}} \cdot \lambda \bigg( \frac{x}{\sqrt{n}} \bigg)\bigg) \cdot \bigg( 1 + O\bigg(\frac{x+1}{\sqrt{n}} \bigg) \bigg). 
\]
Here $\lambda(t) : =\sum_{k=0}^\infty a_k t^k$ is a power series whose coefficients $\{a_k\}$ depend only on the cumulants of the random variable $\bX_1$. Furthermore, the disc of convergence of this series has a positive radius.
\end{theorem}

We will also require a \Cramer-type bound, due to \cite{hu2007cramer}, for when the sum is a sum of randomly chosen values that are drawn \emph{without} replacement from a finite population. In more detail, let $\bX_1,\ldots,\bX_\newN$ be a random sample drawn without replacement from a finite population $\{a_1, \ldots, a_n\}$ with 
\[\sumi a_i = 0 \qquad\text{and}\qquad \sumi a_i^2 = n,\]
 where $\newN < n$.  We 
 will write 
\[\bS_\newN := \sum_{i=1}^\newN \bX_i, \quad p := \frac{\newN}{n}, \quad q := 1-p, \quad\text{and}\quad \omega_n^2 := npq.\]

\begin{theorem}[Theorem~1.2 of~\cite{hu2007cramer}] \label{thm:hrw}
	There exists an absolute constant $A > 0$ such that 
	\[
	{\Prx_{}\sbra{\bS_\newN  \geq t\omega_n} = \pbra{1 - \Phi(t)}\pbra{1 + O(1)(1+t)^3\frac{\beta_{3n}}{\omega_n}}
	}
	\]
	for 
	 {$0\leq t \leq {{\frac 1 A} \cdot} \min\cbra{\frac{\omega_n}{\max_k |a_k|}, \pbra{\frac{\omega_n}{\beta_{3n}}}^{1/3}}$}, where
{$\beta_{3n} := \Ex\sbra{|\bX_1|^3}$}.
\end{theorem}

{\paragraph{The $\ell_p$ Balls $B_p$.}
For $1 \leq p \leq 2$, let $\calA_p$ denote the $p$-absolute moment of a standard univariate Gaussian random variable, 
\[
\calA_p := \Ex_{\bg \sim N(0,1)}\sbra{|\bg|^p},
\]
so e.g.~$\calA_1= \sqrt{\frac 2 \pi}, \calA_2 = 1$, and in general we have 
\begin{equation} \label{eq:moment-formula}
	\calA_p = \sqrt{\frac{2^p}{\pi}}\Gamma\pbra{\frac{p+1}{2}}.
\end{equation}
(See, for example, Equation~18 of~\cite{winkelbauer2012moments}.)
We define the ball $B_p \subset \R^n$ to be 
\begin{equation} \label{eq:Bp}
B_p := \cbra{
x \in \R^n \ : \ \sum_{i=1}^n |x_i|^p = n\calA_p
}, \quad \text{i.e.} \quad
B_p := \cbra{
x \in \R^n \ : \ \|x\|_p = n^{1/p}\calA_p^{1/p}
}.
\end{equation}
(Note that this agrees with our definition of $B_2 = B(\sqrt{n})$ from \Cref{sec:nazarov-ub}.)
It follows from the Berry--Esseen theorem~\cite{berry,esseen} that for each value of $1 \leq p \leq 2$ we have $\Vol(B_p) = {1/2} \pm o_n(1)$. 
}

\subsection{Proof of \Cref{thm:ell-p-approx}}
\label{subsec:ell-p-balls}


The proof will proceed in three steps:
\begin{itemize}
	\item First, we show (\Cref{thm:approx-Bp-by-intersection-of-juntas}) that $B_p$ can be $\eps$-approximated by an intersection of $\newN$-dimensional ``juntas'' that each have very large volume. As in our earlier construction of an approximator for the $\ell_2$ ball, the construction will be probabilistic, but it is a very different probabilistic construction from the $\ell_2$ ball construction of \Cref{thm:de-nazarov} (in particular, it makes essential use of the~\Cramer-type bound (\Cref{thm:hrw}) for sums of independent random variables drawn without replacement from a finite population).
	\item Next (\Cref{thm:Estimate-cj}) we get a high-accuracy estimate of the Gaussian volume $1-\delta$ of each of these juntas, by determining the value of $\delta$ up to a small multiplicative factor. This argument
	uses Petrov's~\Cramer-type bound~(\Cref{thm:petrov}) for sums of i.i.d~random variables.
	\item Finally, we combine our high-accuracy estimate of $\delta$ with our ``relative-error'' universal approximation bound (\Cref{thm:rel-err-dud}) for sets of volume $1-\delta$ to show that each of the juntas can be approximated by an intersection of not-too-many halfspaces.
\end{itemize}
The theorem then follows immediately. 
We note that for the special case when $p = 1$, the second and third steps are not required to obtain a $2^{O_\epsilon(n^{3/4})}$-facet approximator for the cross-polytope; this is because the juntas used in the approximator themselves are polytopes, and hence do not need to be approximated by a polytope using the universal ``relative-error'' bound. 
(See~\Cref{remark:ell-1-done} for more on this.) 


\subsubsection{Approximating $B_p$ by an Intersection of Juntas}
\label{subsubsec:ell-p-juntas}


Let us say that a set $S \subseteq \R^n$ is an \emph{$m$-junta} if $S$ depends on at most $m$ directions in $\R^n$.  More precisely, ``$S$ is an $m$-junta'' means that there are $m$ orthogonal directions $v_1,\dots,v_m \in \R^m$ such that for every $x,y \in \R^n$, if $v_i \cdot x = v_i \cdot y$ for all $i \in [m]$ then $S(x)=S(y).$

The main result of this section is a proof that the set $B_p$ is close to an intersection of ``not too many'' $m$-juntas, where $m$ is ``not too large'' (and moreover, all of the directions for all of the $m$-juntas belong to $\{e_1,\dots,e_m\}$):

\begin{theorem} [Approximating $B_p$ by an intersection of $m$-juntas] \label{thm:approx-Bp-by-intersection-of-juntas}
Let $1 \leq p < 2$ and let $\eps$ be as in \Cref{thm:ell-p-approx}.
There is an intersection $L = \bigcap_{j=1}^M C_j$ of $M$ many $m$-juntas in the support of the distribution $\Dp(M,\newN,\theta)$ defined in \Cref{def:junta-distribution} below which has $\dG(L,B_p) \leq O(\eps)$, where $M$ is as given in \Cref{eq:choiceofM} and $m$ is as given in \Cref{eq:choiceofm}.
\end{theorem}

In the rest of this section we prove \Cref{thm:approx-Bp-by-intersection-of-juntas}. Fix a value $1 \leq p < 2.$
We start by introducing the following distribution over convex sets in $\R^n$:

\begin{definition} \label{def:junta-distribution}
	Let $M, \newN \geq 0$ be integers and let $\theta \geq 0$.  We define $\Dp(M,\newN,\theta)$ to be the distribution over convex sets in $\R^n$ where a draw $\bL\sim\Dp(M,\newN,\theta)$ is obtained as follows:
	\begin{enumerate}
		\item Draw $M$ $\newN$-tuples from $[n]$ without replacement, i.e. draw 
		\[
		\pbra{\i{j}_1, \ldots, \i{j}_\newN} \in {[n] \choose \newN} \quad\text{for}~1\leq j\leq M.\]
		\item For $j\in[M]$, define the convex set 
		\[\bC_j := \cbra{x\in\R^n : \sum_{k=1}^\newN |x_{\i{j}_{k}}|^p \leq \theta}.\]
		\item Output the convex set
		\[\bL := \bigcap_{j=1}^M \bC_j.\]
	\end{enumerate}	
\end{definition}

As in the proof of~\Cref{thm:de-nazarov}, the existence of a set $L$ in the support of the distribution ${\Dp(M,\newN,\theta)}$ that $O(\eps)$-approximates $B_p$ follows from
\begin{equation}~\label{eq:ell-p-goal-0}\Ex_{\bL\sim\Dp(M,\newN,\theta)}\sbra{\Prx_{\bx\sim N(0,I_n)}\sbra{\bL(\bx) \neq B_p(\bx)}} \leq O(\eps),\end{equation}
which, by commuting the order of integration, is equivalent to showing that 
\begin{equation} \label{eq:ell-p-goal-1}
	\Ex_{\bx\sim N(0,I_n)}\sbra{\Prx_{\bL\sim\Dp(M,\newN,\theta)}\sbra{\bL(\bx) \neq B_p(\bx)}} \leq O(\eps).
\end{equation}

As in the earlier argument, \Cref{eq:ell-p-goal-1} allows us to analyze the error incurred by our approximator on an ``$x$-by-$x$ basis.'' For convenience, we define 
\begin{equation} \label{eq:mean-var-ell-p}
	\mu_p := \Ex_{\bx\sim N(0,I_n)}\sbra{\|\bx\|_p^p} = n\calA_p \qquad\text{and}\qquad \sigma_p^2 := \Varx_{\bx\sim N(0,I_n)}\sbra{\|\bx\|_p^p} = n\pbra{\calA_{2p} - \calA_{p}^2}.
\end{equation}
By the Berry-Esseen central limit theorem, we have that for $\bx\sim N(0,I_n)$, 
\[d_{\mathrm{CDF}}\pbra{\|\bx\|_p^p, N\pbra{\mu_p, \sigma_p^2}} = O\pbra{\frac{1}{\sqrt{n}}}\]
where $d_{\mathrm{CDF}}$ denotes Kolmogorov or CDF distance. This observation in turn motivates the following parameter definitions:
\begin{equation}~\label{eq:def-din}
\din := \mu_p - \eps\sigma_p
\qquad\text{and}\qquad 
\dout := \mu_p + \eps\sigma_p.
\end{equation}
By the Berry--Esseen theorem, for $\eps$ as in the statement of~\Cref{thm:ell-p-approx}, we have that 
\begin{equation} \label{eq:boundary-pts}
	\Prx_{\bx\sim N(0,I_n)}\sbra{\din \leq \|\bx\|_p^p \leq \dout} = O(\eps).
\end{equation}
Intuitively, we will ``give up'' on points $x$ with $\din\leq \|x\|_p^p \leq \dout$, and in the rest of the argument our focus will be on achieving high accuracy on most other points. Towards this end, the following definition will be useful:

\begin{definition} \label{def:ell-p-good}
	Let $x\in \R^n$ be a point with $\|x\|^p_p \notin (\din,\dout)$. If $\|x\|_p^p\leq \din$, let $z = sx$ be the point with $\|z\|_p^p = \din$ (so $s\geq 1$ in this case), and if $\|x\|_p\geq\dout$, let $z = sx$ be the point with $\|z\|_p^p = \dout$ (so $s \leq 1$ in this case). We say that $x$ is \emph{good} if $z$ satisfies the following:
	\begin{enumerate}
		\item $\|z\|_{2p}^{2p} - \frac{\|z\|_p^{2p}}{n} \in \sbra{n\pbra{\calA_{2p} - \calA_p^2}\pbra{1 - \frac{O(1)}{n^{5/12}}}, n\pbra{\calA_{2p} - \calA_p^2}\pbra{1 + \frac{O(1)}{n^{5/12}}}}$;
		
		\item $\max_{i\in[n]} |z_i|^p = \pbra{\Theta(\log n)}^{p/2}$; and 
		
		\item $n^{-1}\sumi |z_i|^{3p} = \Theta(1)$. 
	\end{enumerate}
	If $x$ is not good (this includes the possibility that $\|x\|_p^p \in (\din,\dout)$), we say that it is \emph{bad}.
\end{definition}

We claim that with probability at least $1-O(\eps)$, a draw from $N(0, I_n)$ is good: 

\begin{claim} \label{claim:good-prob}
	For $\bx\sim N(0,I_n)$, 
	$\Prx\sbra{\bx~\text{is bad}} = O(\eps).$
\end{claim}

\begin{proof}
%
%
%
%
{
		Let $\bx \sim N(0,I_n)$ and let $\bs'>0$ be the rescaling factor such that $\bz := \bs'\bx$ has $\|\bz\|^p_p={\mu_p}$.
The Berry-Esseen theorem, together with standard tail bounds on the $N(\mu_p,\sigma_p^2)$ Gaussian distribution (recall~\Cref{eq:mean-var-ell-p}),
 implies that with probability $1-O(\eps)$,
\[\bs' = \pbra{1 \pm {\frac {O(1)}{n^{5/12}}}}.\]
Since both $\din$ and $\dout$ are within a multiplicative $(1 \pm O({\frac 1 {\sqrt{n}}}))$-factor of $\mu_p$, it 
suffices to show that each of Items 1, 2, and 3 from~\Cref{def:ell-p-good} hold with $\bx$ rather than $\bz=\bs\bx$ in the relevant expression on the left-hand side, with probability $1-O(\eps)$; we do this below.
}

	For Item~1, note that by Chebyshev's inequality, with probability at least $1-\eps$, 
	\[\|\bx\|_{p}^{p} \in \sbra{\mu_p - \frac{\sigma_p}{\sqrt{\eps}}, \mu_p + \frac{\sigma_p}{\sqrt{\eps}}}.\]
	In particular, it follows that with probability at least $1-2\eps$, we have 
	\begin{align*}
		\|{\bx}\|_{2p}^{2p} - \frac{\|{\bx}\|_p^{2p}}{n}
		&\in \sbra{\mu_{2p} - \frac{\mu_p^2}{n} - \frac{\sigma_p^2}{n\epsilon} \pm \pbra{\frac{n\sigma_{2p} + 2\mu_p\sigma_p}{n\sqrt{\eps}}}}.\\
		\intertext{Recalling the expressions for $\mu_p$ and $\sigma^2_{2p}$ in terms of $\calA_p$ and $\calA_{2p}$ (cf.~\Cref{eq:mean-var-ell-p}), we have}
		\|{\bx}\|_{2p}^{2p} - \frac{\|{\bx}\|_p^{2p}}{n}
		&\in \sbra{n\pbra{\calA_{2p} - \calA_p^2} - \frac{(\calA_{2p} - \calA_p^2)}{\epsilon} \pm \frac{\Theta_p(\sigma_{2p})}{\sqrt{\eps}}} \\ 
		&= \sbra{n\pbra{\calA_{2p} - \calA_p^2}\pbra{1 \pm \frac{\Theta_p(1)}{\sqrt{n}{\eps}}}} \\ 
		&\sse \sbra{n\pbra{\calA_{2p} - \calA_p^2}\pbra{1 \pm \frac{\Theta_p(1)}{n^{5/12}}}}
	\end{align*}
	where the final containment is due to our lower bound on $\eps$ in the statement of~\Cref{thm:ell-p-approx}.
	
	For Item~2, by a standard Gaussian tail bound (\Cref{prop:gaussian-tails}) we have that 
	\[\Pr\sbra{|\bx_i|^p > \pbra{2\sqrt{(\ln n)}}^p} = \Pr\sbra{|\bx_i| > 2\sqrt{\ln n}} < \frac{1}{n^2}\] 
	for each $i \in [n],$ so a union bound over all $i \in [n]$ gives that 
	\[\Pr\sbra{\max_{i \in [n]} |\bx_i|^p \leq \pbra{2\sqrt{(\ln n)}}^p} > 1 - {\frac 1 n},\] which is at least $1 - \eps$.
	
		Turning to Item~3, note that
	\[\Ex_{\bx\sim N(0, I_n)}\sbra{\sumi|\bx_i|^{3p}} = n\calA_{3p} \qquad\text{and}\qquad \Varx_{\bx\sim N(0, I_n)}\sbra{\sumi |\bx_i|^3} \leq n\calA_{6p}.\]
	Consequently, by Chebyshev's inequality we have that with probability at least $1-\eps$ over the draw of $\bx\sim N(0,I_n)$, the following holds:
	\[
	\calA_{3p} - \frac{{\sqrt{\calA_{6p}}}}{\sqrt{\eps n}} \leq \frac{1}{n}\sumi |\bx_i|^3 \leq \calA_{3p} + \frac{{\sqrt{\calA_{6p}}}}{\sqrt{\eps n}}.
	\]
	Since $p\in [1,2)$, we have that both ${\cal M}_{3p}$ and ${\cal M}_{6p}$ are $\Theta(1)$ (cf.~\Cref{eq:moment-formula}). Recalling the lower bound on $\eps$ then immediately implies Item 3.
\end{proof}

 We will show that for an appropriate choice of parameters $M, \newN,$ and $\theta$, we have that whenever $\|x\|_1 \notin (\din,\dout)$ and $x$ is good,
\begin{equation} \label{eq:ell-p-goal-2}
	\Prx_{\bL\sim\Dp(M, \newN, \theta)}\sbra{\bL(x)\neq B_1(x)} \leq O(\eps). 
\end{equation}
\Cref{eq:ell-p-goal-2,eq:boundary-pts}, together with~\Cref{claim:good-prob}, immediately imply~\Cref{eq:ell-p-goal-1}, which completes the proof of~\Cref{thm:ell-p-approx}. 


\paragraph{Case 1: $\|x\|_p^p\leq\din$ {and $x$ is good}.}

In this case, note that $B_p(x) = 1$, and so it suffices to show that for any {fixed $j \in [M]$}, we have 
\begin{equation} \label{eq:ell-p-yes-goal}
	\Prx_{\Dp(M,\newN,\theta)}\sbra{x \in \bC_j} \geq 1-\frac{O(1)\cdot\eps}{M},
\end{equation}
since this immediately implies that 
\[\Prx_{\Dp(M,\newN,\theta)}\sbra{x \in \bigcap_{j=1}^M \bC_j} \geq \pbra{1-\frac{O(1)\cdot\eps}{M}}^M \geq 1-O(\eps),\]
as desired. 

We will require the following monotonicity claim, which allows us to consider points with maximal $\ell_p$ norm. (This claim is why the definition of a ``good'' point $x$, \Cref{def:ell-p-good} above, used the rescaled point $z = sx$.)

\begin{claim} \label{claim:monotonicity}
	Let $x\in\R^n$ have $\|x\|_p^p \leq \din$. Let $z = sx$ for $s\geq 1$ be such that $\|z\|_p^p = \din$. Then 
	\[\Prx_{\Dp(M, \newN, \theta)}\sbra{\sum_{k=1}^n |z_{\i{j}_k}|^p \geq \theta} \geq \Prx_{\Dp(M, \newN, \theta)}\sbra{\sum_{k=1}^n |x_{\i{j}_k}|^p \geq \theta}.\]	
\end{claim}

\begin{proof}
	The claim is immediate from the fact that the map 
	\[|x_{\i{j}_k}|^p \mapsto s^p|x_{\i{j}_k}|^p = |z_{\i{j}_k}|^p\]
	gives a monotone coupling between the two random variables, and hence $|z_{\i{j}_k}|^p$ stochastically dominates $|x_{\i{j}_k}|^p$ (cf.~Theorem~4.2.3 of~\cite{roch-ch4}).
\end{proof}

As in~\Cref{claim:monotonicity}, let $z = sx$ for $s\geq 1$ be such that $\|z\|_p^p = \din$. Let $z^{p}$ denote the vector
\[z^{(p)} := \pbra{|z_1|^p,\ldots,|z_n|^p}.\]
In particular, note that $\|z^{{(p)}}\|_1 = \|z\|_p^p$. We observe that since $x$ is good, the point $z$ satisfies Items~1 through 3 of~\Cref{def:ell-p-good}. 

We will make use of the \Cramer{}-type bound given by~\Cref{thm:hrw} to prove~\Cref{eq:ell-p-yes-goal}. In order to apply~\Cref{thm:hrw}, let
\[v := z^{(p)} - \frac{\|z\|_p^p(1,\ldots, 1)}{n} \qquad\text{and so}\qquad \sumi v_i = 0.\]
Note that 
\begin{align*}
	\sumi v_i^2 &= \|z\|_{2p}^{2p} - \sumi \pbra{\frac{ 2|z_i|^p\|z\|_p^p}{n} - \frac{\|z\|_p^{2p}}{n^2}} \\
	&= \|z\|_{2p}^{2p} - \frac{\|z\|_p^p}{n} \pbra{ \sumi 2|z_i|^p - \frac{\|z\|_p^p}{n}} \\ 
	&= \|z\|_{2p}^{2p} - \frac{\|z\|_p^{2p}}{n} \\ 
	&\in \sbra{n\pbra{\calA_{2p} - \calA_p^2}\pbra{1 \pm \frac{\Theta_p(1)}{n^{5/12}}}}
\end{align*}
by Item~1 of~\Cref{def:ell-p-good}. In particular, let $c_x$ be such that 
\[
\sumi v_i^2 = c_x n, 
\quad\text{i.e.}\quad 
c_x = \pbra{\calA_{2p} - \calA_p^2}\pbra{1\pm\frac{\Theta_p(1)}{n^{5/12}}}.
\]
For notational convenience, we define 
\begin{equation}
\label{eq:ccx}
c := \pbra{\calA_{2p} - \calA_p^2}
\qquad\text{and so}\qquad
c_x = c\pbra{1\pm\frac{\Theta_p(1)}{n^{5/12}}}.
\end{equation}

We will apply~\Cref{thm:hrw} to the finite population $\{a_i\}_{i\in[n]}$, where
\[a_i := \frac{v_{i}}{\sqrt{c_x}}.\]
It is immediate that $\sumi a_i = 0$ and that\ $\sumi a_i^2 =n$, as required by~\Cref{thm:hrw}. Furthermore, we have
\begin{equation}~\label{eq:def-omega}\omega_n^2 = \newN\pbra{1-\frac{\newN}{n}}
.\end{equation}
Since $x$ is good, by Item~2 of~\Cref{def:ell-p-good} we have that $\max_k |a_k| = \pbra{\Theta(\log n)}^{p/2}$, and by Item~3 of ~\Cref{def:ell-p-good} we have that $\beta_{3n} = \Theta(1)$. It follows from~\Cref{thm:hrw} that for any fixed $j \in [M],$ with $\bC_j$ as in~\Cref{def:junta-distribution}, we have 
\begin{equation}
\Prx_{\Dp(M,\newN,\theta)}\sbra{a_{\i{j}_1} + \ldots + a_{\i{j}_\newN}\geq {t\omega_n}} = \pbra{1-\Phi(t)}\pbra{1 + \frac{O(1)(1+t)^3}{\sqrt{\newN}}}
\label{eq:byhrw}
\end{equation}
for $t \leq O\pbra{\min\cbra{\sqrt{\newN/(\log n)^p}, \newN^{1/6}}} = O(\newN^{1/6})$,
where we used $\omega_n^2=\Theta(m)$ which holds {given our ultimate choice of $m$, see \Cref{eq:choiceofm}, using the lower bound 
{\[\paramset\]
from the theorem statement}.}

Let us relate the probability in the L.H.S. of \Cref{eq:byhrw} to our goal probability, which is the L.H.S. of \Cref{eq:ell-p-yes-goal}. 
Recalling that 
\[a_i = c_x^{-1/2}\pbra{z^{(p)}_i - \frac{\|z\|_p^p}{n}},\] 
the event $a_{\i{j}_1} + \ldots + a_{\i{j}_\newN} \geq {t \omega_n}$ is the same as the event 
\[\frac{1}{\sqrt{c_x}}\sum_{k=1}^\newN \pbra{z^{(p)}_{\i{j}_k} - \frac{\|z\|_p^p}{n}} \geq {t \omega_n}, \qquad\text{or equivalently}\qquad  \sum_{k=1}^\newN z_{\i{j}_k}^{(p)} \geq \|{z}\|_p^p\frac{\newN}{n} + \sqrt{c_x}{t\omega_n}.\]
This motivates our choice of $\theta$, which corresponds to the RHS of the second inequality above (but with $c$ in place of $c_x$, since the definition of $\theta$ cannot depend on the specific point $x$), where we fix 
\begin{equation} \label{eq:choiceoft}
t:=c_1 m^{1/6}, \quad \quad \text{where~}c_1 := c_1(\eps) = {c_2 \cdot {\frac {\eps^2} {\ln(1/\eps)}}}
\end{equation} for an absolute constant $c_2>0$.  In more detail, recalling that $\|z\|_p^p=\din$, we take
\begin{equation}
\label{eq:choice-of-theta}
\theta:= \din\frac{\newN}{n} + {\sqrt{c} c_1 m^{1/6} \omega_n},
\end{equation}
and we take $M$ to be
\begin{equation}
\label{eq:def-of-M}
M := {\frac \eps {1-\Phi(t)}}.
\end{equation}
Note that by a standard Gaussian tail bound and our choice of $t$, we have 
\begin{equation}
\label{eq:choiceofM}
M = \eps \cdot \exp(\Theta(c_1^2m^{1/3})) 
= \eps \cdot \exp\pbra{\Theta\pbra{{\frac {\eps^4}{\ln^2(1/\eps)}}m^{1/3}}}.
\end{equation}

With these choices, we have
\begin{align}
	\Prx_{\Dp(M,\newN,\theta)}\sbra{x \notin \bC_j} &\leq \Prx_{\Dp(M,\newN,\theta)}\sbra{z\notin \bC_j}
	\tag{\Cref{claim:monotonicity}} \nonumber\\
	 &= \Prx_{\Dp(M,\newN,\theta)}\sbra{\sum_{k=1}^\newN z^{(p)}_{\i{j}_k} \geq \theta}
	 \tag{Definition of $\bC_j$} \nonumber\\
	 &= \Prx_{\Dp(M,\newN,\theta)}\sbra{\sum_{k=1}^\newN z^{(p)}_{\i{j}_k} \geq \|z\|_p^p{\frac m n} + 
	 \sqrt{c_x} c_1 m^{1/6} \omega_n + \pbra{\sqrt{c} - \sqrt{c_x}}c_1 m^{1/6} \omega_n} \tag{\Cref{eq:choice-of-theta}, using that here $\|z\|_1=\din$}\nonumber\\
	 &= \Prx_{\Dp(M,\newN,\theta)}\sbra{{\frac 1 {\sqrt{c_x}}}\sum_{k=1}^\newN \pbra{z^{(p)}_{\i{j}_k} - {\frac {\|z\|_p^p} n}} \geq t \omega_n + \pbra{{\frac {\sqrt{c}}{\sqrt{c_x}}} - 1}t \omega_n} \tag{Definition of $t$}\\
	&=\Prx_{\Dp(M,\newN,\theta)}\sbra{\sum_{k=1}^\newN a_{\i{j}_k} \geq 
	t'  \omega_n}, \label{eq:happy}
\end{align}
where $t' = {\pbra{1  \pm {\frac {O(1)}{n^{5/12}}}}} \cdot t$ by \Cref{eq:ccx}.

To bound \Cref{eq:happy}, we observe that applying \Cref{eq:byhrw} with $t'$ in place of $t$ throughout, and recalling our choice of $t$, we have that
\[
(\ref{eq:happy})=\Theta(1) \cdot \pbra{1-\Phi(t')}.
\]
Since $t \leq m^{1/6}$ and we will ensure that $m < n$, we have $t < n^{1/6}$, so $O(1)/n^{5/12} = o(1/t)$.  Hence since $t' = {\pbra{1  \pm {\frac {O(1)}{n^{5/12}}}}} \cdot t$, we have that  $1-\Phi(t') = \Theta(1 - \Phi(t))$.  So 
\[(\ref{eq:happy}) = 
\Theta(1) \cdot (1 - \Phi(t)) = \Theta\pbra{{\frac \eps M}},\] 
recalling the definition of $M$ from \Cref{eq:def-of-M}; in other words, we have that 
\[\Prx_{\Dp(M,\newN,\theta)}\sbra{x\notin \bC_j} \leq {\frac {O(\eps)}{M}}\]
as desired (cf.~\Cref{eq:ell-p-yes-goal}).


\paragraph{Case 2: $\|x\|_p^p\geq\dout$ {and $x$ is good}.}

In this case, we will show that for  
\ignore{a good $x$ and}any fixed $j\in[M]$, we have 
\begin{equation} \label{eq:ell-p-no-goal}
	\Prx_{\Dp(M,\newN,\theta)}\sbra{x \in \bC_j } \leq 1 - \frac{\ln(1/\eps)}{M}.
\end{equation}
Indeed, if~\Cref{eq:ell-p-no-goal} holds, then we have that 
\[\Prx_{\Dp(M,\newN,\theta)}\sbra{x \in \bigcap_{j=1}^M \bC_j} \leq \pbra{1 - \frac{\ln(1/\eps)}{M}}^M \leq \eps\]
as desired, completing the proof. 

We will use the following analogue of~\Cref{claim:monotonicity}:

\begin{claim} \label{claim:ell-p-dout-monotone}
	Let $x \in \R^n$ have $\|x\|_p^p\geq \dout$. Let $z = sx$ for $s \leq 1$ be such that $\|z\|_p^p = \dout$. Then
	\[\Prx_{\Dp(M,\newN,\theta)}\sbra{\sum_{k=1}^n |x_{\i{j}_k}|^p \geq \theta } \geq \Prx_{\Dp(M,\newN,\theta)}\sbra{\sum_{k=1}^n |z_{\i{j}_k}|^p \geq \theta }.\]
\end{claim}

{As in \Cref{claim:ell-p-dout-monotone}, let $z = sx$ for $s \leq 1$ be such that $\|z\|_p^p = \dout$.} As before, let 
\[z^{(p)} = (|z_1|^p, \ldots, |z_n|^p).\]
{We observe that since $x$ is good, $z$ satisfies items (1)-(3) of \Cref{def:ell-p-good}.
{Similar to Case~1, we let
\[v := {z}^{(p)} - \frac{\|{z}\|_1(1, \ldots, 1)}{n}.
\]
Consequently, we have 
\[
\sumi v_i = 0, \qquad 
\text{and}
\qquad 
\sumi \pbra{v_i}^2 = \|z\|_{2p}^{2p} - \frac{\|z\|_p^{2p}}{n} \in \sbra{n\pbra{\calA_{2p} - \calA_p^2}\pbra{1 \pm \frac{\Theta_p(1)}{n^{5/12}}}}.
\]
As before, let $c_x$ be such that 
\[\sumi v_i^2 = c_xn, 
\]
and so $c_x = c\cdot \pbra{1  \pm {\frac {\Theta_p(1)}{n^{5/12}}}}$ where $c = \calA_{2p} - \calA_p^2$ as in Case 1. 
As before, we will apply~\Cref{thm:hrw} to the finite population $\{a_i\}_{i\in[n]}$, where
\[a_i := \frac{v_{{i}}}{\sqrt{c_x}}.\]
As before we have $\sumi a_i = 0$ and $\sumi a_i^2 =n$, as required by~\Cref{thm:hrw}, and as before we have $\omega_n^2 = \newN\pbra{1-{\frac \newN n}}$.
}
This lets us write 
\begin{align}
	\Prx_{\Dp(M,\newN,\theta)}\sbra{x \notin \bC_j}
	&\geq \Prx_{\Dp(M,\newN,\theta)}\sbra{z\notin \bC_j}
	\tag{\Cref{claim:ell-p-dout-monotone}} \nonumber\\
	 &= \Prx_{\Dp(M,\newN,\theta)}\sbra{\sum_{k=1}^\newN z^{(p)}_{\i{j}_k} \geq \theta}
	 \tag{Definition of $\bC_j$} \nonumber\\
	 &= \Prx_{\Dp(M,\newN,\theta)}\sbra{\sum_{k=1}^\newN z^{(p)}_{\i{j}_k} \geq \|z\|_p^p{\frac m n} + 
	 \sqrt{c_x} c_1 m^{1/6} \omega_n + \pbra{\sqrt{c} - \sqrt{c_x}}c_1 m^{1/6} \omega_n - {\frac {2 m \eps}{\sqrt{n}}} } \tag{\Cref{eq:choice-of-theta} and using that now $\|z\|_1=\dout$, where $\dout = \din + {\frac {2m\eps}{\sqrt{n}}}$}\nonumber\\
	 &= \Prx_{\Dp(M,\newN,\theta)}\sbra{{\frac 1 {\sqrt{c_x}}}\sum_{k=1}^\newN \pbra{z^{(p)}_{\i{j}_k} - {\frac {\|z\|_p^p} n}} \geq t \omega_n + \pbra{{\frac {\sqrt{c}}{\sqrt{c_x}}} - 1}t \omega_n - {\frac {2m \eps}{\sqrt{c_x}{\sqrt{n}}}}} \tag{Definition of $t$}\\
	&=\Prx_{\Dp(M,\newN,\theta)}\sbra{\sum_{k=1}^\newN a_{\i{j}_k} \geq 
	t'  \omega_n}, \label{eq:happyout}
\end{align}
where
\begin{equation} 
t' = t\pbra{1  \pm {\frac {O(1)}{n^{5/12}} - {\frac {2\newN\eps}{\sqrt{c_x}{\sqrt{n}}t\omega_n}}}},
\label{eq:newtprime}
\end{equation}
using $\pbra{{\frac {\sqrt{c}}{\sqrt{c_x}}} - 1} = \pm {\frac {O(1)}{n^{5/12}}}.$
Let us rewrite this as 
\begin{equation} \label{eq:myDelta}
t' = t\pbra{1 - \Delta},
\quad
\text{where}
\quad
\Delta := {\frac {2\newN\eps}{\sqrt{c_x}{\sqrt{n}}t\omega_n}} \pm {\frac {O(1)}{n^{5/12}}}
= {\frac {\Theta(\eps \sqrt{\newN})}{\sqrt{n} t }} \pm {\frac {O(1)}{n^{5/12}}},
\end{equation}
where the last equality is because $\omega_n = \Theta(\sqrt{\newN})$.
We observe that because we will ultimately take 
\[\newN = \Theta\pbra{{\frac{(\log(1/\eps))^{9/4}}{{\eps^4}}} \cdot n^{3/4}}\]
and we have $t=c_1 \newN^{1/6}$ and {$\paramset$}, we have 
${\frac {\Theta(\eps \sqrt{\newN})}{\sqrt{n} t }} = \omega(1)/n^{5/12}$, so $\Delta$ is simply equal to 
${\frac {\Theta(\eps \sqrt{\newN})}{\sqrt{n} t }}$.
Now we apply \Cref{thm:hrw} (equivalently, \Cref{eq:byhrw} with $t'$ in place of $t$ throughout).
We get that 
\[
(\ref{eq:happyout}) = \Theta(1) \cdot (1-\Phi(t')).
\]
By \Cref{eq:myDelta}, 
\begin{align}
(\ref{eq:happyout}) = \Theta(1) \cdot
\pbra{1-\Phi(t')} &= \Theta(1)\cdot(1-\Phi(t)) \cdot \exp(t^2\Delta - t^2\Delta^2/2) \nonumber\\
&=\Theta(1) \cdot {\frac \eps M} \exp\pbra{\Theta\pbra{{\frac {\eps \sqrt{\newN}t}{\sqrt{n}}}} - \Theta\pbra{{\frac {\eps^2 \newN}{n}}}},
 \label{eq:happy2}
\end{align}
where the last equality uses $(1-\Phi(t))={\frac \eps M}$ and $\Delta = {\frac {\Theta(\eps \sqrt{\newN})}{\sqrt{n} t }}$.
Now we choose
\begin{equation} \label{eq:choiceofm}
\newN := c_3 \cdot {\frac{(\log(1/\eps))^{9/4}}{{\eps^4}}} \cdot n^{3/4} \quad \text{for a suitable positive constant $c_3$}
\end{equation}
(note that by the bound on $\eps$ in the statement of the theorem we have $m \leq n/2$), 
and {using $t=c_1 \newN^{1/6}$} we get that 
$(\ref{eq:happy2}) \geq {\frac {\ln(1/\eps)}{M}}$, as required to establish \Cref{eq:ell-p-no-goal}.

In particular, it follows from~\Cref{eq:ell-p-goal-2,eq:ell-p-yes-goal,eq:ell-p-no-goal} that there exists a set $L$ in the support of the distribution ${\Dp(M,\newN,\theta)}$ that $O(\eps)$-approximates $B_p$.
This concludes the proof of \Cref{thm:approx-Bp-by-intersection-of-juntas}.

\begin{remark} \label{remark:ell-1-done}
{Note that the above argument already establishes the special case of~\Cref{thm:ell-p-approx} when $p = 1$ (i.e. the existence of a {polytopal} $O(\eps)$-approximator to the $\ell_1$-ball $B_1$) with the weaker bound of 
	\[M\cdot 2^\newN = {\exp\pbra{\Theta\pbra{\frac{(\log(1/\epsilon))^{9/4}}{{\epsilon^4}} n^{3/4}}}}~\text{halfspaces}.\] 
	 This is simply because when $p=1$, each $\newN$-junta $\bC_j$ is itself a $2^{\newN}$-facet polytope.

	 The stronger bound claimed in \Cref{eq:ell-one-bound} can be achieved by a slight adjustment of the parameters used in this section.  For the case $1<p<2$ we need the arguments in \Cref{subsubsec:hello-petrov}, and those arguments demand that we set $m$ and $t$ as in \Cref{eq:choiceofm} and \Cref{eq:choiceoft}; but for the $p=1$ case, where we only require \Cref{thm:approx-Bp-by-intersection-of-juntas}, it can be verified that it suffices to take {$m = \Theta\pbra{\pbra{\frac{\log(1/\epsilon)}{\epsilon}}^{3/2} n^{3/4}}$ and $t$ to be $\Theta(m^{1/6})$.}
	 }	 
\end{remark}


\subsubsection{Bounds on $\Vol(\bC_j)$}
\label{subsubsec:hello-petrov}

We observe that for every outcome of the random set $\bC_j$ from \Cref{def:junta-distribution}, the volume is the same.  So let us fix $C_j$ to be any outcome of $\bC_j$. 
The goal of this section is to prove the following rather tight estimate on $\Vol(C_j)$: 

\begin{theorem} [Upper and lower bounds on $\Vol(C_j)$] \label{thm:Estimate-cj}
For $C_j$ as defined above, 
\[1-\frac{\ln M}{M} \le\Vol(C_j) \le 1- \frac{1}{3M}.\]
\end{theorem}
One direction of this bound is easy to prove. In particular, we have the following simple claim. 
\begin{claim}~\label{thm:Estimate-cj-ub}
$\Vol(C_j) \le  1 - 1/(3M)$.
\end{claim}
\begin{proof}
Towards contradiction, if $\Vol(C_j) > 1-1/(3M)$, then for $\Vol(L) = \Vol(\cap_{j=1}^M C_j) \ge 2/3$ (by a union bound). This contradicts \eqref{eq:ell-p-goal-0}, thus proving the claim. 
\end{proof}

The remainder of this section is devoted to proving the following claim (which will complete the proof of \Cref{thm:Estimate-cj}). 
\begin{claim}~\label{claim:Estimate-cj-lb}
$\Vol(C_j) \ge 1 - (\ln M)/M$.
\end{claim}
The main workhorse in proving this claim will be \Cref{thm:petrov}. To apply this theorem, we begin with a simple and easy to verify fact. 
\begin{fact}~\label{fact:finitemoment}
Let $p \le 2$. Then, for $\bg \sim N(0,1)$ and $\by := |\bg|^p$, there exists $H>0$ such that 
$\mathbf{E}[e^{H\by}]<\infty$. In fact, for $p<2$, this is true for any $H>0$. 
\end{fact}

Note that the parameters $\calA_p$ and $\calA_{2p}$ (following \eqref{eq:ell-p-goal-1}) are set so that 
\[
\Ex_{\bx\sim N(0,1)}\sbra{|\bx|^p} = \calA_p \qquad\text{and}\qquad \Varx_{\bx\sim N(0,1)}\sbra{|\bx|^p} = \calA_{2p} - \calA_{p}^2. 
\]
By definition,  for any $j$, 
\begin{equation}~\label{eq:volume-Cj}
\Vol(C_j) = \Prx_{\bg \sim N(0,I_m)} \sbra{\sum_{j=1}^m |\bg_j|^p \leq  \theta }. 
\end{equation}

To analyze \eqref{eq:volume-Cj}, our principal tool will be  \Cref{thm:petrov}. In particular, for each $1 \le j \le m$, define
$$
\bX_j = |\bg_j|^p  - \calA_{p}.
$$
From \Cref{fact:finitemoment}, it follows that each random variable $\bX_j$ satisfies \Cramer's condition. Further, note that $\Ex[\bX_j]=0$ and $\sigma^2:=\frac{\sigma_p^2}{n}=\Var(\bX_j) = \calA_{2p} - \calA_{p}^2 =c$. We now note that 
\begin{equation} \nonumber
\Vol(C_j) =\Pr\big[\sum_{j=1}^m |\bg_j|^p \le \theta\big]  =\Pr\big[\sum_{j=1}^m \bX_j \le \theta -m \cdot \calA_p\big] 
\end{equation}
Defining $\mathbf{S}_m: =\sum_{j=1}^m \bX_j$, we get that 
\[
\Vol(C_j) = \Pr[\mathbf{S}_m \le z \cdot \sigma \cdot \sqrt{m}], 
\]
where $z$ is defined as 
\begin{eqnarray}
z = \frac{\theta - m \cdot \calA_p}{\sigma \cdot \sqrt{m}}
= \frac{\din\frac{\newN}{n} + {\sqrt{c} c_1 m^{1/6} \omega_n} - m \cdot \calA_p}{\sigma \cdot \sqrt{m}} = \frac{(n \calA_p - \epsilon \cdot \sigma_p)\frac{\newN}{n} + {\sqrt{c} c_1 m^{1/6} \omega_n} - m \cdot \calA_p}{\sigma \cdot \sqrt{m}}.
 \nonumber 
\end{eqnarray}
Here the first expression uses 
\eqref{eq:choice-of-theta} and 
second expression uses the definition of $\din$ \eqref{eq:def-din}. 
The last expression in the above equation simplifies to 
\[
\frac{ {\sqrt{c} t \omega_n}- \epsilon \cdot \sigma_p\frac{\newN}{n}  }{\sigma \cdot \sqrt{m}} =  \frac{ {\sqrt{c} t \omega_n}- \frac{\epsilon \cdot \sigma \cdot \newN}{\sqrt{n}} }{\sigma \cdot \sqrt{m}}
\] 
where we plug in the expression for $t$ \eqref{eq:choiceoft}. As $c =\sigma^2$, plugging in the value of $\omega_n$ from \eqref{eq:def-omega}, we get that 
\[
z= t \cdot \sqrt{1-\frac{m}{n}} - \epsilon \cdot \sqrt{\frac{m}{n}}. 
\]
As $t =\Theta(\epsilon \cdot m^{1/6}/\sqrt{\log(1/\epsilon)})$ \eqref{eq:choiceoft} and $\epsilon \le 1$, it means that $z =O(m^{1/6})$. Thus, we can apply \Cref{thm:petrov} to get that 
\begin{equation}~\label{eq:applpetrov}
\frac{1-\Vol(C_j)}{1-\Phi(z)} = \frac{1-\Pr[\mathbf{S}_m \le z \cdot \sigma \cdot \sqrt{m}]}{1-\Phi(z)} = \exp \bigg( \frac{z^3}{\sqrt{m}} \cdot \lambda \bigg(\frac{z}{\sqrt{m}} \bigg)\bigg) \cdot \bigg(1 + O \bigg(\frac{z+1}{\sqrt{m}} \bigg) \bigg).
\end{equation}
As $z=
O( m^{1/6})$, it follows that 
the right hand side of \eqref{eq:applpetrov} is $\Theta(1)$. Note that we are using here that $\lambda(\cdot)$ is a power series with a positive radius of convergence. Thus, we have that 
\begin{equation}
\frac{1-\Vol(C_j)}{1-\Phi(z)} = \Theta(1).\label{eq:applpetrov2} 
\end{equation}
Next, we record the following simple fact.

\bigskip

\ignore{
\blue{
Let me try to line up what is going on in the file with what I think we want.

\medskip

The thing we have been calling $\delta$ is $1-\Vol(\bC_j).$ We already know by the union bound that $\delta \geq 0.49/M$ and what we want to show is that $\delta < 1000/M$. Anindya has established for us that 
\[\delta = 1-\Vol(\bC_j) = \Theta(1-\Phi(z)),\]
and so we have 
\[\delta = \Theta(1-\Phi(t))\frac{1-\Phi(z)}{1-\Phi(t)} = \frac{\eps}{M} \cdot\Theta\pbra{\frac{1-\Phi(z)}{1-\Phi(t)}}.
\]
So we would be done if we could show that 
\[\frac{1-\Phi(z)}{1-\Phi(t)} \leq {\frac{1000}{\eps}}.\]
So don't we want a variant of~\Cref{fact:gaussian-decay} where you lower bound 
\[\frac{1-\Phi(b)}{1-\Phi(a)}\]
}
\bigskip

}

\begin{fact}~\label{fact:gaussian-decay}
Suppose $a$, $b$ and $\eta$ are parameters such that $a>b> \ln(1/\eta) >2$ 
and  $a-b \le \frac{\ln (1/\eta)}{a}$. Then, 
\[
\frac{1-\Phi(b)}{1-\Phi(a)} = O(1/\eta). 
\]
\end{fact}
\begin{proof}
Since  $1-\Phi(x)$ is a decreasing function, it suffices to prove this for $b = a- \frac{\ln (1/\eta)}{a}$. 
Consider the function $g(x) = \ln(1-\Phi(x))$. Then, $g'(x) =\frac{-1}{1-\Phi(x)} \cdot \phi(x)$, where $\phi(x)$ is the density of $N(0,1)$ at $x$.  Hence, by Rolle's theorem, it follows that for some $c \in [b,a]$,
\[
\frac{\phi(c)}{1-\Phi(c)} \cdot (a-b) = \ln \bigg( \frac{1-\Phi(b)}{1-\Phi(a)} \bigg). 
\]
The factor $\phi(c)/(1-\Phi(c))$ is known as the hazard rate of the normal distribution and Lemma~2, Chapter VII in \cite{Feller} shows that 
\[
\frac{\phi(c)}{1-\Phi(c)} \le \frac{c^3}{c^2-1}. 
\] 
Plugging this back, we see that 
\[
\ln \bigg( \frac{1-\Phi(b)}{1-\Phi(a)} \bigg) \le (a-b) \cdot \frac{b^3}{b^2-1} \le (a-b) \cdot \bigg( b + \frac{2}{b} \bigg). 
\]
Using the setting of our parameters, the right hand side is at most $\ln (1/\eta) +2$. Exponentiating both sides, we get the claimed bound. 
\end{proof}

\ignore{
\begin{fact}~\label{fact:gaussian-decay}
Suppose $a$, $b$ and $\epsilon$ are parameters such that $a>b> \ln(1/\epsilon) >4$ 
and  $a-b \le \frac{\ln (1/\epsilon)}{a}$. \blue{(we think the inequality in the assumption will be
 $a-b \le \frac{\ln (1/\epsilon)}{a}$)
} Then, \blue{(we think the conclusion will be LHS $\leq$ RHS)}
\[
\frac{1-\Phi(b)}{1-\Phi(a)} = \Omega(1/\epsilon). 
\]
\end{fact}
\begin{proof}
Since  $1-\Phi(x)$ is a decreasing function, it suffices to prove this for $b = a- \frac{\ln (1/\epsilon)}{a}$. 
Consider the function $g(x) = \ln(1-\Phi(x))$. Then, $g'(x) =\frac{-1}{1-\Phi(x)} \cdot \phi(x)$, where $\phi(x)$ is the density of $N(0,1)$ at $x$.  Hence, by Rolle's theorem, it follows that for some $c \in [b,a]$,
\[
\frac{\phi(c)}{1-\Phi(c)} \cdot (a-b) = \ln \bigg( \frac{1-\Phi(b)}{1-\Phi(a)} \bigg). 
\]
Furthermore, for $c \ge 4$, \blue{(we'll use Feller reference in bottom of  \href{https://math.stackexchange.com/questions/1326879/limit-of-normal-hazard-rate}{https://math.stackexchange.com/questions/1326879/limit-of-normal-hazard-rate} to flip the direction)}
\[
\frac{\phi(c)}{1-\Phi(c)} \ge c. 
\]
This implies that 
\[
\ln \bigg( \frac{1-\Phi(b)}{1-\Phi(a)} \bigg) \ge b (a-b)  = \bigg(a-\frac{\ln(1/\epsilon)}{a} \bigg) \cdot \frac{\ln(1/\epsilon)}{a} \ge \ln(1/\epsilon)-1. 
\]
Exponentiating both sides, this finishes the proof. 

\end{proof}
}

We now observe that 
\begin{eqnarray}
t(t-z) = t \bigg( \epsilon \cdot \sqrt{\frac{m}{n}} + t \bigg(1-\sqrt{1-\frac{m}{n}} \bigg)\bigg)  \le t \bigg( \epsilon \cdot \sqrt{\frac{m}{n}} + \frac{tm}{2n}\bigg) = \epsilon \cdot A + \frac{A^2}{2}, \nonumber
\end{eqnarray} 
where $A:= t\sqrt{m/n}$. By our setting of our parameters (once $\epsilon$ is less than a sufficiently small constant), $A>1$ and thus, the right hand side of the above equation is upper bounded by $A^2$. Next, plugging in the values of $m$ and $t$ in terms of $n$ and $\epsilon$, 
$$
A^2 = \frac{t^2 m}{n} = \Theta \pbra{\frac{\epsilon^4 m^{1/3} \cdot m}{\ln^2(1/\epsilon) \cdot n}} =\Theta \pbra{ \frac{ \epsilon^4 \ln^3(1/\epsilon)}{\ln^2(1/\epsilon) \cdot \epsilon^{16/3}}} = \Theta\pbra{\frac{1}{\epsilon^{4/3}} \cdot \ln(1/\epsilon)}. 
$$
This implies that 
\[
t(t-z) = \Theta\pbra{\frac{1}{\epsilon^{4/3}} \cdot \ln(1/\epsilon)}, 
\]
and thus applying \Cref{fact:gaussian-decay}, we have that 
\begin{equation}
\frac{1-\Phi(z)}{1-\Phi(t)} = 2^{\Theta(\epsilon^{-4/3} \cdot \ln(1/\epsilon))}. \nonumber
\end{equation}
Combining the above with \eqref{eq:applpetrov2}, it follows that 
\[
1-\Vol(C_j) = 2^{\epsilon^{-4/3} \cdot \ln(1/\epsilon)} \cdot (1-\Phi(t)) = \frac{2^{\Theta(\epsilon^{-4/3} \cdot \ln(1/\epsilon))}}{M},
\]


\ignore{
{\color{red} Anindya: we have to make sure the parameters satisfy this.} 
This implies that 
\[
t(t-z) \le \ln(1/\epsilon) + \epsilon \cdot \sqrt{2 \ln(1/\epsilon)} \le 1 + \ln(1/\epsilon), 
\]
for a sufficiently small $\epsilon$. Applying \Cref{fact:gaussian-decay}, we have that 
\begin{equation}
\frac{1-\Phi(z)}{1-\Phi(t))} = O(1/\epsilon). \nonumber
\end{equation}
Combining the above with \eqref{eq:applpetrov2}, it follows that 
\[
1-\Vol(C_j) = O((1/\epsilon)) \cdot (1-\Phi(t)) = O(1/M),
\]
}
Using the bound on $\epsilon$
{(i.e. $\epsilon \ge (\log n)^{-1/2}$)} from the theorem statement, it follows that the right hand side is bounded by $(\ln M)/M$.
 This proves \Cref{claim:Estimate-cj-lb} and thus, \Cref{thm:Estimate-cj}.


\subsubsection{Putting Everything Together via \Cref{thm:rel-err-dud}}
\label{subsubsec:bye-bye-ell-p}

\Cref{thm:Estimate-cj} gives us that for each $j \in [M]$, we have $\Vol(C_j)=1-\delta$ where $\Omega(1/M) = \delta= O((\log M)/M)$
 Since each $C_j$ is not an intersection of any finite number of halfspaces, it remains to approximate each $C_j$ to high \emph{relative-error} accuracy, which we do using \Cref{thm:rel-err-dud}.  \Cref{thm:rel-err-dud} (where we write ``$\tau$'' for the parameter which is ``$\epsilon$'' in its theorem statement) gives us that for each $C_j$, which is a closed convex set in $\R^\newN$, there is an intersection of halfspaces in $\R^\newN$, which we denote $L_j$, which is an intersection of
\begin{equation} \label{eq:almost-done}
\frac{1}{\delta} \cdot \pbra{\frac{\newN}{\tau}\log\pbra{\frac{1}{\delta}}}^{O(\newN)}
=O(M) \cdot \pbra{\frac{\newN}{\tau}\log\pbra{\Theta(M)}}^{O(\newN)}
\end{equation}
halfspaces and which satisfies 
$\dG(C_j,L_j) \leq \tau \delta = O(\tau \cdot (\log M))/M$.
We choose $\tau = \epsilon/(\log M)$,
which gives us that $\tau \delta = O(\eps/M)$, and gives us that
\[
(\ref{eq:almost-done}) = 
O(M) \cdot \pbra{\frac{\newN}{\eps}\log\pbra{\Theta(M)}}^{O(\newN)}.
\]
Since $L := \bigcap_{j \in [M]} C_j$ satisfies $\dG(B_p,L)\leq O(\eps),$
the intersection $L^\ast := \bigcap_{j \in [M]} L_j$ of these $M$ approximators
satisfies $\dG(L,L^\ast) \leq M \tau \delta \leq O(\eps)$, and we get
\[
\dG(B_p,L) \leq \dG(B_p,L) + \dG(L,L^\ast) \leq O(\eps) + M \tau \delta \leq O(\eps).
\]
Recalling the value of $\newN$ from \Cref{eq:choiceofm} and the value of $M$ from \Cref{eq:choiceofM}, we see that the number of halfspaces in $L^\ast$ is at most
\[
M \cdot (\ref{eq:almost-done}) = 
\Theta(M^2) \cdot \pbra{\frac{\newN}{\eps}\log\pbra{\Theta(M)}}^{O(\newN)}
=
\exp\pbra{\Theta\pbra{\pbra{{\frac {\log^{9/4}(1/\eps)}{\eps^4}}}\cdot \log \pbra{{\frac n \eps}} \cdot n^{3/4}}}
\]
as claimed, and the proof of \Cref{{thm:ell-p-approx}} is complete. \qed

\part{Lower Bounds}

In \Cref{sec:nonexplicit} we leverage recent information-theoretic lower bounds from \cite{DS21colt} on the sample complexity of weak learning convex sets under $N(0,I_n)$ to give a non-constructive average-case lower bound: we show that for any $\eps = \omega (n^{-1/2} \cdot \log n)$, there exists some convex set $K \subset \mathbb{R}^n$ that requires $2^{\Omega(\eps\cdot \sqrt{n})}$ halfspaces to approximate to within error $1/2 - \eps$.

In \Cref{sec:noise-sensitivity} we give a sufficient condition for a convex set to be hard to approximate, by showing that if a convex set $K$ has $\GNS_{1/\poly(n)}(K) \geq 1/2 - 1/\poly(n)$, then any intersection of halfspaces that achieves some particular $1/2 - 1/\poly(n)$ error must use at least $2^{n^{\Omega(1)}}$ halfspaces.  We also establish the existence of convex sets $K$ with $\GNS_{1/\poly(n)}(K) \geq 1/2 - 1/\poly(n)$, again using the information-theoretic weak learning lower bounds of \cite{DS21colt}.

In \Cref{sec:conv-inf} we show that if $L$ is an intersection of $s$ halfspaces, then the total convex influence of $L$ is at most $O(\log s)$.  We use this structural result to prove that if a symmetric convex set $K$ has total convex influence $\TInf[K]$ at least $\Omega(\sqrt{n})$, then any intersection of halfspaces that approximates $K$ to some sufficiently small constant accuracy must use $2^{\Omega(\sqrt{n})}$ many halfspaces.  This lets us conclude that several explicit convex sets, including the volume-$1/2$ $\ell_2$ ball and the volume-$1/2$ $\ell_1$ ball, require $2^{\Omega(\sqrt{n})}$ halfspaces for (sufficiently small) constant error approximation.


\section{Non-Explicit Average-Case Lower Bounds} \label{sec:nonexplicit}

The goal of this section is to prove the following theorem:

\begin{theorem} [Non-explicit average-case lower bound.] \label{lb:non-explicit-convex}
There is an absolute constant $C>0$ such that for any sufficiently large $n$ and any $\epsilon = \omega (n^{-1/2} \cdot \log n)$, there exists some convex set $K \subset \mathbb{R}^n$ such that for any $L$ which is an intersection of at most $M:=2^{C \cdot \epsilon \cdot \sqrt{n}}$ halfspaces, 
\[
\Vol(K \, \triangle \, L) > \frac12 -\epsilon. 
\]
\end{theorem}

We note that \Cref{lb:non-explicit-convex} gives an \emph{average-case} lower bound, i.e.~it establishes ``strong inapproximability'' by intersections of not-too-many halfspaces. For example, taking $\eps = n^{-1/4}$, it shows that there is an $n$-dimensional convex set $K$ such that no intersection of $2^{cn^{1/4}}$ many halfspaces can approximate $K$ to accuracy even as large as $1/2 + n^{-1/4}$.  

The proof of this theorem will go via the notion of \VC-dimension, a fundamental measure of complexity in statistical learning theory. Recall that given a set ${\cal F}$ of Boolean functions over some domain $X$, the \emph{\VC-dimension of ${\cal F}$} is the largest size of any $S \subseteq X$ which is {\em shattered} by ${\cal F}$, meaning that every possible $0/1$ labeling of the points in $S$ is achieved by some function in ${\cal F}.$  (See the book \cite{KearnsVazirani:94} for more details.)

To use \VC-dimension in our context, we will also need the notion of agnostic learning.  We recall the following standard definition:

\begin{definition}
A class ${\cal F}$ of Boolean functions over $\mathbb{R}^n$ is said to be {\em agnostically PAC learnable} with sample complexity $m(\epsilon,\delta)$ if for every $\epsilon, \delta>0$, the following holds. There is an algorithm ${\cal A}$ which for any $f: \mathbb{R}^n \rightarrow \{0,1\}$ and any distribution ${\cal D}$ over $\mathbb{R}^n$, given $m(\eps,\delta)$ many i.i.d.~labeled samples of the form $(\bx, f(\bx))$ (where each $\bx \sim {\cal D}$), outputs a hypothesis $h: \R^n \to \{0,1\}$ such that with probability $1-\delta$,
\[
\Prx_{\bx \sim {\cal D}} \sbra{h(\bx) \not = f(\bx)} \le \epsilon + \min_{h^\ast \in {\cal F}} \Prx_{\bx \sim {\cal D}} \sbra{h^\ast(\bx) \not = f(\bx)}. 
\]
\end{definition}

A central result of statistical learning theory is that any concept class ${\cal F}$ of Boolean functions is agnostically PAC learnable where the sample complexity is proportional to the 
\VC-dimension of ${\cal F}$.  We state a sharp form of the bound below, which is due to Talagrand~\cite{Tal94}:
\begin{theorem}~\label{VCdim-learnability}
Any concept class ${\cal F}$ of Boolean functions over $\mathbb{R}^n$ is agnostically PAC learnable with sample complexity 
\[
m_{{\cal F}}(\epsilon, \delta): = \Theta \bigg( \frac{\VC\text{-}\mathrm{dim}({\cal F})}{\epsilon^2} + \frac{\log (1/\delta)}{\epsilon^2} \bigg), 
\]
where $\VC\text{-}\mathrm{dim}({\cal F})$ is the \VC-dimension of ${\cal F}$. 
\end{theorem}
Recall that $\Facet(n,M)$ denotes the class of convex sets in $\R^n$ which are intersection of at most $M$ halfspaces.  The following is shown in \cite{BEH+:89} (and is now a standard fact, see e.g. Exercise~3.4 of \cite{KearnsVazirani:94}):
\begin{fact}~\label{fact:vcdim-inter}
$\VC\text{-}\mathrm{dim}(\Facet(n,M)) = O(nM \log M)$. 
\end{fact}

The final ingredient we will require is a recent result of De and Servedio (Theorem~2 in ~\cite{DS21colt}) which establishes a lower bound on the query complexity of any algorithm which ``weakly learns'' (meaning that the output hypothesis has an error rate only slightly less than $1/2$) convex sets over the Gaussian space. (This lower bound of course also holds for the sample complexity of any algorithm which only receives independent labeled samples $(\bx,f(\bx))$ where each $\bx \sim N(0,I_n)$.)

\begin{theorem}[Theorem~2 of \cite{DS21colt}] \label{thm:our-BBL-lb}
For sufficiently large $n$, for any $s \geq n$, there is a distribution ${\cal D}$ over centrally symmetric convex sets $\bK \subset \R^n$ with the following property:  for a target convex set $\bK \sim {\cal D},$ for any black box query algorithm $A$ making at most $s$ many queries to $\bK$, the expected error of $A$ (the probability over $\bK \sim {\cal D}$, over any internal randomness of $A$, and over a random Gaussian $\bx \sim N(0,1^n)$, that the output hypothesis $h$ of $A$ predicts incorrectly on $\bx$) is at least $1/2 - {\frac {O(\log s )}{n^{1/2}}}$.
\end{theorem}

We now combine these ingredients to establish \Cref{lb:non-explicit-convex}:

\begin{proofof}{\Cref{lb:non-explicit-convex}}
First, let us assume that for every convex set $K$, there is a convex set $L$ which is an intersection of at most $M$ halfspaces such that 
\[
\Vol(K \, \triangle \, L) \le \frac12 -\epsilon. 
\]
Now, consider the task of ``weak learning'' an unknown convex set $K \subseteq \mathbb{R}^n$  given i.i.d.~labeled samples of the form $(\bx, K(\bx))$ where $\bx \sim N(0,I_n)$ and $K(\cdot)$ is identified with the indicator function of the convex set $K$. To do this, we run the agnostic PAC learning algorithm from \Cref{VCdim-learnability} for the class $\Facet(n,M)$ with parameters $\epsilon/2$ and $\delta$ on the samples. 

Now, given that there is a convex set $L \in \Facet(n,M)$ such that $\Vol(K \, \triangle \, L) \le \frac12 -\epsilon$, it follows that 
with probability $1-\delta$, the algorithm from \Cref{VCdim-learnability} outputs a hypothesis $h$ such that 
\[
\Prx_{\bx \sim N(0,I_n)} [h(\bx) \not = K(\bx)] \le \frac{1}{2} - \epsilon + \epsilon/2 = \frac{1}{2} - \epsilon/2. 
\]
Further, the sample complexity of this algorithm is given by $S (\epsilon, \delta)$ defined as 
\[
S(\epsilon,\delta) = \Theta \bigg(\frac{\VC\text{-}\mathrm{dim}(\Facet(n,M))}{\epsilon^2} + \frac{\log(1/\delta)}{\epsilon^2} \bigg)
\]
Applying~\Cref{fact:vcdim-inter}, we get that 
\[
S(\epsilon,\delta) = O \bigg( \frac{n M \log M}{\epsilon^2} + \frac{\log(1/\delta)}{\epsilon^2}\bigg) 
\]
In particular, if we set $\delta = \epsilon/4$, then we get that there is a PAC learning  algorithm for convex sets (where the data $\bx \sim N(0,I_n)$) with sample complexity 
\[
S(\epsilon, \epsilon/4) = O \bigg( \frac{n M \log M}{\epsilon^2} + \frac{\log(1/\epsilon)}{\epsilon^2}\bigg) 
\] 
which has expected error at most $1/2 - 3 \epsilon/4$ (where the expectation is over the internal randomness as well as the randomness of the data points $\bx \sim N^n(0,1)$). On the other hand, by~\Cref{thm:our-BBL-lb}, we get that the expected error must be at least $1/2 - O(n^{-1/2} \cdot \log S(\epsilon, \epsilon/4))$. Combining these bounds, we get that 
\[
\frac{1}{2} - \frac{3\epsilon}{4} \ge \frac{1}{2} - O(n^{-1/2} \cdot \log S(\epsilon, \epsilon/4)). 
\]
This implies that $\log S(\epsilon, \epsilon/4)) = \Omega(\epsilon \cdot \sqrt{n})$. This implies that 
\[
\frac{n M \log M}{\epsilon^2} + \frac{\log(1/\epsilon)}{\epsilon^2}  = 2^{\Omega(\epsilon \cdot \sqrt{n})}. 
\]
Since $\epsilon = \omega( n^{-1/2} \cdot \log n )$ and $M \ge 1$, the first term on the left hand side is the dominant one, and thus we have
\[
M \log M = \frac{\epsilon^2}{n} \cdot 2^{\Omega(\epsilon \cdot \sqrt{n})}. 
\]
Again using $\epsilon = \omega( n^{-1/2} \cdot \log n )$, we have that 
$$
M \log M = 2^{\Omega(\epsilon \cdot \sqrt{n})} \quad \quad \textrm{and thus} \quad \quad M = 2^{\Omega(\epsilon \cdot \sqrt{n})}. 
$$
This finishes the proof of \Cref{lb:non-explicit-convex}.
\end{proofof}




	


\section{Average-Case Lower Bounds via Gaussian Noise Sensitivity} \label{sec:noise-sensitivity}

One drawback of \Cref{lb:non-explicit-convex} is that it is non-constructive: not only does it not exhibit any particular convex set which is hard to approximate, it does not give any criterion that suffices to make a convex set hard to approximate.

In this section we describe a general structural property --- namely, having high \emph{Gaussian noise sensitivity} (see \Cref{def:GNS}) at low noise rates --- and show that if a convex set has this property, then any intersection-of-halfspaces approximator for it must use many halfspaces.  
More precisely, we prove the following:

\begin{theorem} \label{thm:high-GNS-hard}
Fix any constant $c_1>0$ and let $0 \leq \kappa = \kappa(n)<1/2$ (note that $\kappa$ may depend on $n$).
Let $K \subset \R^n$ be a convex set such that $\GNS_{n^{-c_1}}(K) \geq 1/2 - \kappa.$
Then there is a constant $c_2$ (depending on $c_1$) such that if $L \subset \R^n$ is any intersection of $2^{n^{c_2}}$ many LTFs, we have
\[
\Vol(K \ \triangle \ L) \geq {\frac 1 2} - {2} \kappa^{1/4} - \frac{3}{2}n^{-c_2}.
\]
\end{theorem}

\Cref{thm:high-GNS-hard} raises the question of whether there exist convex sets that have very high Gaussian noise sensitivity at very low noise rates. This is a natural question, but it does not seem to have been explicitly considered in prior work~\cite{OD23pc}. In \cite{Nazarov:03} Nazarov gave a randomized construction of an $n$-dimensional convex set that has Gaussian surface area $\Omega(n^{1/4})$, thus showing that Ball's $O(n^{1/4})$ upper bound \cite{Ball:93} on the Gaussian surface area of any $n$-dimensional convex sets is best possible up to the hidden constant.  A straightforward geometric argument applied to Nazarov's construction gives that there are $n$-dimensional convex sets $K$ for which $\GNS_{n^{-1/4}}(K)$ is at least some absolute constant, but these arguments do not seem to give a stronger lower bound of the form $1/2 - o_n(1)$.  

In \Cref{sec:existence} we show how known results on the information-theoretic hardness of \emph{weak learning} convex sets under the Gaussian distribution (the main lower bound result of \cite{DS21colt}, which we stated as \Cref{thm:our-BBL-lb} in \Cref{sec:nonexplicit}) can be leveraged to establish the existence of convex sets with very high Gaussian noise sensitivity (inverse-polynomially close to $1/2$) at very low noise rates (inverse polynomial). This result, stated below, may be of independent interest:

\begin{theorem} \label{thm:highGNSlownoise}
Fix any constant $\tau < 1/4$.  For sufficiently large $n$, there are convex bodies $K \subset \R^n$ 
that have $\GNS_{n^{-\tau}}(K) \geq 1/2 - n^{-\tau}.$
\end{theorem}

Combining \Cref{thm:highGNSlownoise} with \Cref{thm:high-GNS-hard} we immediately get an average-case lower bound for intersections of $2^{n^{\Omega(1)}}$ halfspaces, albeit one which is quantitatively not quite as strong as \Cref{lb:non-explicit-convex}:

\begin{corollary} \label{cor:average-case}
There is an absolute constant $c>0$ such that for sufficiently large $n$, there are convex sets $K \subset \R^n$ such that if $L \subset \R^n$ is any intersection of $2^{n^{c}}$ many LTFs, then $\Vol(K \ \triangle \ L) \geq 1/2 - n^{-c}$.
\end{corollary}

\subsection{Random Zooms, Hypervariance, and Local Hyperconcentration}
\label{subsec:zoom-prelims}

\cite{OSTK21} introduced the notion of a ``random zoom'' of a function over Gaussian space and observed that it is a natural analogue of the well-studied notion of a random restriction of a function defined over the Boolean cube $\{-1,1\}^n$.  Random zooms are defined as follows:

\begin{definition}
    For a polynomial $p: \R^n \to \R$, $0 \leq \lambda \leq 1$, and $x \in \R^n$, we define the function $\zoom{p}{\lambda}{x}$ by
    \[
        \zoom{p}{\lambda}{x}(y) = p\pbra{\sqrt{1-\lambda} x + \sqrt{\lambda} y},
    \]
and we refer to the function $\zoom{p}{\lambda}{x}(\cdot)$ as the \emph{$\lambda$-zoom of $p$ at $x$}.
\end{definition}

For intuition, if $x \in \R^n$ is fixed we view $\sqrt{1-\lambda} x + \sqrt{\lambda} \by$ as a ``$\lambda$-noisy'' version of~$x$, and we view changing a function~$p$'s input from $\bx \sim N(0,I_n)$ to $\sqrt{1-\lambda} x + \sqrt{\lambda} \by$ as ``zooming into $p$ at~$x$ with scale~$\lambda$.''
We note that if $p$ is a degree-$d$ polynomial then so is $\zoom{p}{\lambda}{x}$, for any $\lambda<1$.

The following observation will be useful:
\begin{observation} \label{obs:completion-to-Gaussian}
If $\bx, \by \sim N(0,I_n)$ are independent and $0 \leq \lambda \leq 1$, then $\sqrt{1-\lambda} \bx + \sqrt{\lambda} \by$ is distributed identically to $\bz \sim N(0,I_n)$.  Thus the distribution of $\zoom{f}{\lambda}{\bx}(\by)$ is identical to $f(\bz),$ for any function $f$ and any $0 \leq \lambda \leq 1.$
\end{observation}

To adopt the terminology of \cite{RST:15}, \Cref{obs:completion-to-Gaussian} says for a standard Gaussian $\bx$, a random zoom of $f$ at $\bx$ ``completes to the Gaussian distribution''.
\Cref{obs:completion-to-Gaussian} is similar to how, in the Boolean function context of $f: \bits^n \to \{0,1\}$, evaluating a random restriction on $f$ at a uniform random input is the same as evaluating $f$ at a uniform random input; we mention that this simple observation plays an important role in a number of classical as well as more recent correlation bounds for small-depth circuits over $\bits^n$, see e.g. \cite{Hastad86,Cai86,RST:15}.  See \cite{OSTK21} for a more detailed discussion of how random zooms for functions over Gaussian space are analogous to random restrictions of functions over $\bn$.

The ``Gaussian noise'' (or ``Ornstein--Uhlenbeck'') operator (see, e.g., \cite[Def.~11.12]{odonnell-book}) corresponds to the expectation of the zoom of a function:
\begin{definition} \label{def:gaussian-noise-operator}
    Given $0 < \rho \leq 1$, the operator $\U_\rho$ acts on polynomials $p : \R^n \to \R$ via
    \[
        (\U_\rho p)(x) = \Ex_{\by \sim N(0,I_n)}\sbra{p\pbra{\rho x + \sqrt{1-\rho^2} \by)}} = \Ex_{\by \sim N(0,I_n)}\sbra{ \zoom{p}{(1-\rho^2)}{x}(\by)}.
    \]
The operator $\U_\rho$ acts diagonally in the Hermite polynomial basis $(h_\alpha)_{\alpha \in \N^n}$ (see \Cref{appendix:hermite} for a brief overview of the basics of Hermite analysis):
    \begin{equation}    \label{eqn:U-formula}
        \U_\rho p = \sum_{\alpha \in \N^N} \rho^{|\alpha|} \wh{p}(\alpha) h_\alpha.
    \end{equation}
    In particular, if $p$ is a degree-$d$ polynomial, so too is $\U_\rho p$.

    Following \cite{OSTK21}, we will need to extend the definition of $\U_\rho$ to $\rho > 1$, which we can do via the formula \Cref{eqn:U-formula}; equivalently, by stipulating that $\U_{\rho^{-1}} = \U_{\rho}^{-1}$. Now we can define the \emph{hypervariance} of a function:
\end{definition}

\begin{definition} \label{def:hypervar}
    Let $g: \R^n \to \R$ be a polynomial.  For $R > 1$, we define the \emph{$R$-hypervariance} of~$g$ to be
    \[
        \HV_R[g] := \Var[\U_R g] = \sum_{\alpha \neq 0} R^{2|\alpha|} \wh{g}(\alpha)^2.
    \]
    (For $R = 1$, this reduces to the usual variance of~$g$.)
    If $\HV_R[g] \leq \eps \|g\|^2_2$, then we say that $g$ is \emph{$(R,\eps)$-attenuated}.
\end{definition}

The key structural result of \cite{OSTK21} states, roughly speaking, that a random zoom of a low-degree polynomial is very likely to be ``hyperconcentrated.'' In more detail, Theorem~85 of \cite{OSTK21}, specialized to the case that its ``nice'' distribution $(g_{\bup})_{\bup \sim \Upsilon}$ is a single degree-$d$ polynomial $p$, gives us the following:

\begin{theorem} [Local Hyperconcentration Theorem for a single polynomial, Theorem~85 of \cite{OSTK21}] \label{thm:OSTK-main-theorem}
There exists a constant $c>0$ so that for any $1>\eps,\beta>0$ and $R\geq 1$, if $p(x_1,\dots,x_n)$ is a degree-$d$ polynomial, then taking
$$
\lambda \leq \frac{c\eps \beta}{R d^{9/2}},
$$
with probability at least $1-\beta$ over $\bx \sim N(0,I_n)$ we have that
$$
\hypvar_R(p_{\lambda|\bx}) \leq \eps^2  \| p_{\lambda|\bx}\|_2^2  
\quad
\text{(i.e.~$p_{\lambda|\bx}$ is $(R,\eps)$-attenuated).}
$$
\end{theorem}
We also recall the following result from \cite{OSTK21}, which roughly speaking says that ``attenuated polynomials are very likely to take values that are multiplicatively close to their means'':
\begin{proposition} [Proposition~32 of \cite{OSTK21}] \label{prop:atten-hyperconcy}
    Let $g$ be a polynomial that is $(R,\eps)$-attenuated, with $R \geq \sqrt{2}$ and $\eps \leq 1$. Write $\mu = \E[g]$ and assume $\mu \neq 0$.
Then for $0 < \gamma \leq 1$, for $\bz \sim N(0,I_n)$ we have $g(\bz) \approx_{\gamma} \mu$ (and hence $\sign(g(\bz))=\sign(\mu)$) except with probability at most $(2\sqrt{\eps}/\gamma)^{\frac12 R^2 + 1}$.
\end{proposition}

\subsection{PTFs and Intersections of Halfspaces Collapse Under Random Zooms}
\label{sec:collapse}

\Cref{thm:OSTK-main-theorem} and \Cref{prop:atten-hyperconcy} together give us that a random zoom of any degree-$d$ PTF will with high probability become a close-to-constant function:

\begin{corollary} [Random zooms of PTFs are close-to-constant whp]  \label{cor:zoom-close-to-constant}
Let $f: \R^n \to \{-1,1\}$, $f=\sign(p)$ be a degree-$d$ PTF, and let $0<\lambda<c'/d^{9/2}$ for a sufficiently small constant $c'$.  Then for any $0<\tau \leq 1$, we have
\begin{equation} \label{eq:bound}
\Prx_{\bx \sim N(0,I_n)}\sbra{\Var(\zoom{f}{\lambda}{\bx})\geq \tau} \leq 
{\frac {C d^{9/2} \lambda}{\tau}},
\end{equation}
where $C>0$ is a universal constant.
\end{corollary}

\begin{proof}
We assume that $p$ is a non-constant polynomial since otherwise the claim is trivially true.
Let $R=\sqrt{2}$, let $\eps=\tau/64$, and let $\lambda = c\eps \beta/(Rd^{9/2})$ where $c$ is the constant from \Cref{thm:OSTK-main-theorem}; equivalently, let $\beta$ be $16 \sqrt{2} d^{9/2} \lambda/(c\tau)$ (this is the RHS of \Cref{eq:bound}, taking $C = 16\sqrt{2}/c$). By taking $c'$ suitably small we have $\beta < 1$, and clearly $\eps<1$, so by \Cref{thm:OSTK-main-theorem} with probability at least $1-\beta$ over $\bx \sim N(0,I_n)$ we have that $\zoom{p}{\lambda}{\bx}$ is $(R,\eps)$-attenuated. 

Fix any outcome $x$ of $\bx$ such that $\zoom{p}{\lambda}{\bx}$ is $(R,\eps)$-attenuated and $\mu_{\bx} := \E[\zoom{p}{\lambda}{x}] \neq 0$ (note that the probability that $\mu_{\bx}=0$ is zero for $\bx \sim N(0,I_n)$, so the total probability of such $x$-outcomes is at least $1-\beta$). Taking $\gamma=1/2$ in \Cref{prop:atten-hyperconcy}, we get that $\zoom{f}{\lambda}{x}(\bz)$ equals $\sign(\E[\zoom{p}{\lambda}{x})$ except with probability at most $(2\sqrt{\eps}/\gamma)^{\frac12 R^2 + 1} = 16\eps=\tau/4$, and hence for such an $x$ we have $\Var(\zoom{f}{\lambda}{x}) \leq 4 \cdot {\frac \tau 4} \cdot (1-{\frac \tau 4}) \leq \tau.$
\end{proof}

To connect \Cref{cor:zoom-close-to-constant} with intersections of halfspaces, we recall the fact (from  \cite{KOS:08}) that any intersection of not-too-many halfspaces is close to a low-degree PTF:

\begin{theorem} [Low-degree PTF approximators for intersections of halfspaces, Theorems~15 and~20 of \cite{KOS:08}] \label{thm:KOS}
Let $K \subseteq \R^n$ be any convex set.  Then for any $\eps > 0$ there is a polynomial $p: \R^n \to \R$ of degree $d = O(\log(k)/\eps^2)$ such that the PTF $f=\sign(p)$ satisfies
\[
\Vol(K \ \triangle \ f^{-1}(1)) \leq \eps.
\]
\end{theorem}

Combining \Cref{thm:KOS} with \Cref{cor:zoom-close-to-constant}, we get that an intersection of few halfspaces simplifies to a close-to-constant function with high probability under a random zoom:

\begin{corollary} 
[Intersections of few halfspaces collapse to close-to-constant under random zooms]
\label{cor:polytope-collapse}
Let $L \subseteq \R^n$ be an intersection of at most $k$ halfspaces, and let $0<\delta<1.$
For a suitable choice of $\lambda = \Theta(\delta^{20}/(\log k)^{9/2}),$ we have
\[
\Prx_{\bx \sim N(0,I_n)} \sbra{\Var\pbra{\zoom{L}{\lambda}{\bx}} \geq \delta} \leq \delta.
\]
\end{corollary}

\begin{proof}
We argue by contradiction; so suppose that 
\begin{equation}
\label{eq:to-contradict}
\Prx_{\bx \sim N(0,I_n)} \sbra{\Var\pbra{\zoom{L}{\lambda}{\bx}} \geq \delta} > \delta.
\end{equation}
Let $\eps := \delta^2/64$.
By \Cref{thm:KOS}, there is a PTF $f=\sign(p)$ of degree $d := C_1 \log(k)/\eps^2$ which is an $\eps$-approximator for $L$, i.e. $\Vol(L \ \triangle \ f^{-1}(1)) \leq \eps.$
For a suitable choice of the hidden constant in $\lambda=\Theta(\delta^{20}/(\log k)^{9/2})$ (note that this choice of $\lambda$ satisfies the conditions of \Cref{cor:zoom-close-to-constant}) and $\tau=\sqrt{\eps}$, by \Cref{cor:zoom-close-to-constant} we have that 
\begin{equation} \label{eq:by-zoom-cor}
\Prx_{\bx \sim N(0,I_n)}\sbra{\Var(\zoom{f}{\lambda}{\bx})\geq \tau} \leq 
{\frac {C d^{9/2} \lambda}{\tau}} = \delta/2.
\end{equation}
Combining \Cref{eq:to-contradict} and \Cref{eq:by-zoom-cor}, we get that
\begin{equation} \label{eq:labor}
\Prx_{\bx \sim N(0,I_n)}
\sbra{
\Var\pbra{\zoom{L}{\lambda}{\bx}} \geq \delta
\text{~and~}
\Var(\zoom{f}{\lambda}{\bx})\leq \tau
} \geq \delta/2.
\end{equation}
Now, if $\Var_{\by}[a(\by)] \geq \delta$ and $\Var_{\by}[b(\by)] \leq \tau$ for $\{-1,1\}$-valued functions $a,b$ and $\delta>\tau$ then it must be the case that $\Pr_{\by}[a(\by) \neq b(\by)] \geq (\delta - \tau)/4.$
Recalling that $\delta=8\sqrt{\eps}$ and $\tau=\sqrt{\eps}$, \Cref{eq:labor} gives
\[
\Prx_{\bx \sim N(0,I_n)}\sbra{
\Prx_{\by \sim N(0,I_n)}\sbra{
\zoom{L}{\lambda}{\bx}(\by) \neq \zoom{f}{\lambda}{\bx}(\by)
} \geq 7\sqrt{\eps}/4
}\geq 4\sqrt{\eps},
\]
from which \Cref{obs:completion-to-Gaussian} gives 
\[
\Prx_{\bz \sim N(0,I_n)}\sbra{f(\bz) \neq L(\bz)} \geq
{\frac {7 \sqrt{\eps}} 4} \cdot {4 \sqrt{\eps}} = 7 \eps.
\]
This contradicts  $\Vol(L \ \triangle \ f^{-1}(1)) \leq \eps$, and the proof is complete.
\end{proof}

\subsection{Gaussian Noise Stability and Expected Variance After Random Zooms}

We recall the definitions of \emph{Gaussian noise stability} and \emph{Gaussian noise sensitivity}:

\begin{definition} [cf.~Definition~11.18 of \cite{odonnell-book}] \label{def:GNS}
For $f \in L^2(\R^n,\gamma)$ and $\rho \in [-1,1]$, the \emph{Gaussian noise stability of $f$ at $\rho$} is defined to be
\[
\Stab_\rho[f] := \Ex_{\bz,\bz'}\sbra{f(\bz)f(\bz')} = \abra{f,\U_{\rho}f},
\]
where $\bz,\bz'$ are $\rho$-correlated $N(0,I_n)$ Gaussians (equivalently, $\bz,\bg$ are independent $N(0,I_n)$ Gaussians and $\bz'=\rho \bz + \sqrt{1-\rho^2} \bg$).
For $\rho \in [0,1]$, the \emph{Gaussian noise sensitivity of $f$ at $\rho$} is
\[
\GNS_\rho[f] := {\frac 1 2} - {\frac 1 2} \Stab_{1-2\rho}[f].
\]
\end{definition}
We observe that for $f: \R^n \to \{-1,1\}$, we have
\[
\GNS_\rho[f] = \Prx_{\bz,\bz'}\sbra{f(\bz) \neq f(\bz')} = 2 \Prx_{\bz,\bz'}\sbra{f(\bz)=1, f(\bz') = -1},
\]
where now $\bz,\bz'$ are $(1-2\rho)$-correlated $N(0,I_n)$ Gaussians.

Gaussian noise sensitivity is important for us because the expected variance of $\zoom{K}{\lambda}{\bx}$, over $\bx \sim N(0,I_n)$, is exactly captured by Gaussian noise sensitivity:

\begin{lemma} \label{lem:expected-variance}
Let $K \subseteq \R^n$ be a convex set (viewed as a $\pm 1$-valued function on $\R^n$). For $0 \leq \lambda \leq 1$, we have that
\[
\Ex_{\bx \sim N(0,I_n)}
\sbra{
\Var_{\by \sim N(0,I_n)}[\zoom{K}{\lambda}{\bx}(\by)]
}
=
2\GNS_{\lambda/2}[K].
\]
\end{lemma}
\begin{proof}
Since $\Var[f(\by)]={\frac 1 2} \E_{\by,\by'}[(f(\by)-f(\by'))^2]$ where $\by,\by'$ are i.i.d., the LHS above is equal to
\begin{align*}
&
{\frac 1 2} \Ex_{\bx \sim N(0,I_n)}
\sbra{
\Ex_{\by,\by'\sim N(0,I_n)}\sbra{(\zoom{K}{\lambda}{\bx}(\by)-\zoom{K}{\lambda}{\bx}(\by'))^2}
}\\
&=
{\frac 1 2} \Ex_{\bx,\by,\by' \sim N(0,I_n)}
\sbra{
\pbra{K\pbra{\sqrt{1-\lambda} \bx + \sqrt{\lambda} \by} -
K\pbra{\sqrt{1-\lambda} \bx + \sqrt{\lambda} \by'}}^2
}\\
&= {\frac 1 2}
\Ex_{\bx,\by,\by' \sim N(0,I_n)}
\sbra{2 - 2
K\pbra{\sqrt{1-\lambda} \bx + \sqrt{\lambda} \by} 
K\pbra{\sqrt{1-\lambda} \bx + \sqrt{\lambda} \by'}
}\\
&= 1 - \Stab_{1-\lambda}[K] = 2 \GNS_{\lambda/2}[K]. \qedhere
\end{align*}
\end{proof}

\subsection{Proof of~\Cref{thm:high-GNS-hard}}

All the pieces are now in place for us to prove \Cref{thm:high-GNS-hard}. We first sketch the idea of the argument and then give quantitative details: Let $K$ be a convex set whose Gaussian noise stability at some low noise rate $\lambda/2$ is high, i.e.~
$\GNS_{\lambda/2}(K) \geq {\frac 1 2} - \kappa$ where $\kappa$ is small, and let $L$ be any intersection of ``not too many'' halfspaces. Our goal is to show that $L$ cannot be a good approximator for $K$; to do this, let us consider what happens if we apply a random zoom at noise rate $\lambda$ to both $K$ and $L$:

\begin{itemize}

\item  By the results of the previous subsection, the expected variance of $K$ after the random zoom is $2\GNS_\lambda[K]$, which is very close to 1. This means that almost every outcome $\bx$ of the random zoom must be such that the convex body $\zoom{K}{\lambda}{\bx}$ has Gaussian volume very close to $1/2$.

\item But by the results of \Cref{sec:collapse}, almost every outcome $\bx$ of the random zoom causes $L$ to collapse to a function with very \emph{low} variance, and hence almost every outcome of the convex body $\zoom{L}{\lambda}{\bx}$ has Gaussian volume very close to either 0 or 1.

\end{itemize}

This means that ``zoom outcome by zoom outcome'' $\zoom{L}{\lambda}{\bx}$ is typically a  poor approximator for $\zoom{K}{\lambda}{\bx}$; it follows (formally, by \Cref{obs:completion-to-Gaussian}) that $L$ must overall be a poor approximator for $K$.

We now proceed with the formal proof. 

\medskip

\noindent \emph{Proof of \Cref{thm:high-GNS-hard}.} 
Let $K \subset \R^n$ be a convex set with $\GNS_{n^{-c_1}} \geq {\frac 1 2} - \kappa.$
By \Cref{lem:expected-variance}, we have
\[
\Ex_{\bx}\sbra{
\Var_{\by}\sbra{
\zoom{K}{2n^{-c_1}}{\bx}(\by)
}
}
\geq 1 - 2\kappa.
\]
Since $\Var_{\by}\sbra{\zoom{K}{2n^{-c_1}}{\bx}(\by)} \in [0,1]$ for all outcomes of $\bx$, a standard reverse Markov argument gives that
\[
\Prx_{\bx}\sbra{
\Var_{\by}\sbra{
\zoom{K}{2n^{-c_1}}{\bx}(\by)
}
\geq 1 - \sqrt{2 \kappa}}
\geq 1 - \sqrt{2 \kappa}.
\]
Consequently, for at least a $1 - \sqrt{2 \kappa}$ fraction of outcomes of $\bx$, we have that
\begin{equation} \label{eq:firstie}
\Prx_{\by}\sbra{
\zoom{K}{2n^{-c_1}}{\bx}(\by) = 1
} \in \sbra{{\frac 1 2} - (2 \kappa)^{1/4},{\frac 1 2} + (2 \kappa)^{1/4}}.
\end{equation}

On the other hand let $L$ be any intersection of at most $2^{\Theta(n^{c_2})}$ many halfspaces over $\R^n$, where $c_2 = 2c_1/49.$ 
By \Cref{cor:polytope-collapse} (taking its $\delta = 1/\log k = n^{-c_2}$), we get that 
\[
\Prx_{\bx \sim N(0,I_n)} \sbra{\Var\pbra{\zoom{L}{\lambda}{\bx}} \geq n^{-c_2}} \leq n^{-c_2}.
\]
So for a least a $1 - n^{-c_2}$ fraction of outcomes of $\bx$, we have that
\begin{equation} \label{eq:secondie}
\min \cbra{\Prx_{\by}\sbra{\zoom{L}{\lambda}{\bx}(\by)=1},\Prx_{\by}\sbra{\zoom{L}{\lambda}{\bx}(\by)=-1}} \leq n^{-c_2}.
\end{equation}
Combining \Cref{eq:firstie} and \Cref{eq:secondie}, for at least a $1 - \sqrt{2 \kappa} - n^{-c_2}$ fraction of outcomes of $\bx$, we have that
\[
\abs{
\Prx_{\by}\sbra{
\zoom{K}{2n^{-c_1}}{\bx}(\by)=1} 
-
\Prx_{\by}\sbra{\zoom{L}{\lambda}{\bx}(\by)=1
}
} 
\geq
{\frac 1 2} - (2 \kappa)^{1/4} - n^{-c_2}.
\]
Applying \Cref{obs:completion-to-Gaussian}, we get that
\[
\Vol(K \,\triangle\, L) \geq 
\pbra{{\frac 1 2} - (2 \kappa)^{1/4} - n^{-c_2}} \cdot
\pbra{1 - \sqrt{2 \kappa} - n^{-c_2}},
\]
which, after some straightforward algebra, gives the claimed bound.
\qed


\subsection{Existence of Convex Sets with High Gaussian Noise Sensitivity at Low Noise}
\label{sec:existence} 

We now turn to the proof of \Cref{thm:highGNSlownoise}. 
The proof is by contradiction; so fix a constant $0<\tau<1/4$, and suppose that every convex set $K \subset \R^n$ has $\GNS_{n^{-\tau}}(K) < 1/2 - n^{-\tau}$. 
Viewing $K$ as a $\{-1,1\}$-valued indicator function (where $x \in K$ iff $K(x)=1$) and recalling \Cref{def:GNS}, this is equivalent to the assertion that every convex set $K \subseteq \R^n$ satisfies $\Stab_{1-2n^{-\tau}}(K) > 2n^{-\tau}$.

We recall the well-known fact (see \cite[Prop.~11.37]{odonnell-book}) that the noise stability at noise rate $\rho$ has an exact expression in terms of the amount of ``Hermite weight'' (in the Hermite expansion) at each weight level $0,1,\dots .$  More precisely, we have that for  $f \in L^2(\R^n,\gamma)$ and $\rho \in [-1,1]$, 
\[
\Stab_\rho[f] = \sum_{\alpha \in \N^n} \rho^{|\alpha|} \widetilde{f}(\alpha)^2,
\]
where $\widetilde{f}(\alpha)$ is the $\alpha$-th Hermite coefficient in the Hermite expansion of $f$. In our context, this gives us that every convex set $K$ satisfies
\begin{equation} \label{eq:stablb}
\Stab_{1-2n^{-\tau}}[K] = \sum_{\alpha \in \N^n} \pbra{1-2n^{-\tau}}^{|\alpha|} \widetilde{K}(\alpha)^2 > 2n^{-\tau}.
\end{equation}

Fix $\ell := 2\tau\ln(n)n^\tau$.  Since, by Parseval's identity, we have $\sum_{\alpha \in \N^n} \tilde{f}(\alpha)^2 = \E[K^2] = 1,$ we have that all of the Hermite coefficients at levels $\ell$ and beyond can together only make a small contribution to $\Stab_{1-2n^{-\tau}}[K]$, i.e.
\begin{equation} \label{eq:highlowcontrib}
\sum_{|\alpha| \geq \ell}\pbra{1-2n^{-\tau}}^{|\alpha|} \widetilde{K}(\alpha)^2\leq
\sum_{|\alpha| \geq \ell}e^{-\tau\ln(n)} \widetilde{K}(\alpha)^2\leq n^{-\tau}.
\end{equation}
\Cref{eq:stablb,eq:highlowcontrib} together give a lower bound on the Hermite weight of any convex set $K \subseteq \R^n$ at levels $0$ through $\ell$, i.e.
\begin{equation} \label{eq:hermlb}
\sum_{|\alpha| \leq \ell} \widetilde{K}(\alpha)^2
\geq
\sum_{|\alpha| \geq \ell}\pbra{1-2n^{-\tau}}^{|\alpha|} \widetilde{K}(\alpha)^2 > n^{-\tau}.
\end{equation}
Since the Hermite polynomials form an orthonormal basis for $f \in L^2(\R^n,\gamma)$, \Cref{eq:hermlb} is equivalent to the existence of a degree-$\ell$ polynomial $p$ whose $L^2$-distance from $K$ is bounded away from 1, i.e.
\[
\Ex_{\bx \sim N(0,1)^n}\sbra{(p(\bx)-K(\bx))^2} < 1-n^{-\tau}.
\]

Given the existence of such a polynomial $p$, it is well known that the ``low-degree algorithm'', which uses calls to a black-box oracle for $K$ to estimate each Hermite coefficient $\widetilde{K}(\alpha)$ for $|\alpha| \leq \ell$, makes at most $n^{O(\ell)}$ black-box oracle calls to $K$ and constructs a degree-$\ell$ polynomial $g$ such that with probability at least $1-1/n$ (over the internal randomness of the low-degree algorithm) the polynomial $g$ satisfies 
\[
\Ex\sbra{(g-K)^2} < 1 - n^{-\tau}/2.
\]
Given such a polynomial $g$, a standard approach (see Lemma~3 of \cite{bfjkmr94}) transforms $g$ into a randomized hypothesis $h: \R^n \to \zo$ which satisfies $\Pr_{\bx \sim N(0,1)^n}\sbra{h(\bx) \neq K(\bx)} \leq {\frac 1 2} \E[(g-K)^2] < 1/2 - n^{-\tau}/4.$
(The transformation is as follows: for each $x \in \R^n$, have $h(x)=-1$ with probability $p := {\frac {(1-g(x))^2}{2(1+g(x)^2)}}$ and $h(x)=1$ with probability $1-p$.)

Summarizing the above, we have the following: Under the assumption that every convex set $K \subset \R^n$ has $\GNS_{n^{-\tau}}(K) < 1/2 - n^{-\tau}$, there is a ``weak learning algorithm'' with the following performance guarantee: given black-box oracle access to any unknown convex set $K \subseteq \R^n$, the algorithm makes at most $n^{O(\ln(n)n^\tau)}$ oracle calls and returns a randomized hypothesis function $h: \R^n \to \{-1,1\}$ which satisfies
\begin{equation} \label{eq:us-good-error}
\Pr\sbra{h(\bx) \neq K(\bx)} < 1/2 - n^{-\tau}/4 + 1/n
< 1/2 - n^{-\tau}/8,
\end{equation}
where $\bx \sim N(0,1)^n$ and the probability is over the randomness of $\bx$, the internal randomness of the learning algorithm, and the internal randomness of the randomized hypothesis $h$. In other words, the expected error of the weak learning algorithm is at most $1/2 - n^{-\tau}/8.$ 

However, this is in tension with the main lower bound of \cite{DS21colt}, which gives an information-theoretic lower bound on the minimum possible error that can be achieved by any weak learning algorithm for convex sets that does not make too many black-box queries, and was stated earlier as \Cref{thm:our-BBL-lb}.
Taking $\tau' = {\frac 1 2} (\tau + 1/4) < 1/4$, and taking $s=n^{O(n^{\tau'})}$, \Cref{thm:our-BBL-lb} gives that any algorithm making at most $s$ black-box oracle calls must have expected error at least $1/2 - O(\ln(n) n^{\tau/2 - 3/8}).$  Since $\tau$ is strictly less than $1/4$ this is in contradiction with \Cref{eq:us-good-error}, and \Cref{thm:highGNSlownoise} is proved.

\section{Lower Bounds for the $\ell_1$ and $\ell_2$ Balls (and More) via Convex Influences}
\label{sec:conv-inf}

Taking $\eps$ to be a sufficiently small constant in \Cref{lb:non-explicit-convex}, we can infer the existence of some convex set in $\R^n$ such that $2^{\Omega(\sqrt{n})}$ halfspaces are required for any $\eps$-approximator, but that result does not let us conclude that any particular convex set is hard to approximate.

In this section we show that the $\ell_1$ and $\ell_2$ balls $B_1$ and $B_2$ (defined in  \Cref{eq:Bp}) are each hard to approximate:

\begin{theorem} \label{thm:ball-conv-inf-lb}
	Any intersection of halfspaces that approximates $B_2$ to error $\eps$ must have at least $2^{\Omega(\sqrt{n})}$ facets, for some absolute positive constant $\eps > 0.$
	{The same is true for $B_1$.}
\end{theorem}

\Cref{thm:ball-conv-inf-lb} implies that the upper bound obtained in~\Cref{sec:nazarov-ub} is tight in the constant error regime, up to constant factors in the exponent. 

Our proof of \Cref{thm:ball-conv-inf-lb} will crucially make use of the notion of \emph{convex influence} which was introduced by~\cite{DNS21itcs,DNS22}. More generally, we prove a lower bound on the number of facets required to approximate any symmetric\footnote{Recall that a set $K\sse\R^n$ is \emph{symmetric} if $x\in K$ implies $-x\in K$.}  convex set whose convex influence is asymptotically maximal up to constant multiplicative factors (cf.~\Cref{thm:high-influence-lb}); see~\Cref{subsec:convex-influence-lower-bound} for more on this.

\begin{remark}
We note that the results of \Cref{sec:noise-sensitivity} do not apply to $B_2$, since $B_2$ is well known to be highly noise stable. In more detail, \cite{Ball:93} showed that every origin-centered ball $B(r) \subseteq \R^n$ of any radius $r$ has Gaussian surface area $\GSA(B(r)) \leq 1/\sqrt{\pi} + o_n(1)$, and it is known (see e.g.~Corollary~14 of \cite{KOS:08}) that for any noise rate $\delta$ and any convex set $K$, we have $\GNS_\delta(K) \leq \sqrt{\pi} \cdot \sqrt{\delta} \cdot \GSA(K).$
It can similarly be shown that $B_1$ is also too noise-stable for the results of \Cref{sec:noise-sensitivity} to apply.
\end{remark}

\subsection{Convex Influences}
\label{subsec:convex-influence}

The following notion was introduced in \cite{DNS21itcs,DNS22} as an analogue of the well-studied notion of \emph{influence of a variable on a Boolean function} (cf. Chapter~2 of \cite{odonnell-book}). 

\begin{definition} \label{def:convex-influence}
	Given a convex set $K\sse\R^n$ with $0^n\in K$, the \emph{convex influence of a direction $v\in \S^{n-1}$ on $K$} is defined as
	\[\Inf_v[K] := \Ex_{\bx\sim N(0, I_n)}\sbra{K(\bx)\pbra{{1 - \abra{\bx, v}^2}}},\]
{where $K(\cdot)$ is the 0/1-valued indicator function of the convex set $K$.}
	We further define the \emph{total convex influence of $K$} as 
	\[\TInf[K] := \sum_{i=1}^n \Inf_{e_i}[K] = \Ex_{\bx\sim N(0, I_n)}\sbra{K(\bx)\pbra{{n- \|\bx\|^2}}}.\]
\end{definition}

We note that the definitions of $\Inf_{v}[K]$ and $\TInf[K]$ as defined in \cite{DNS22} include an additional multiplicative factor of ${1}/{\sqrt{2}}$ that we omit here. 
The total convex influence as defined above can be understood as capturing the rate of growth of the Gaussian measure of a convex set under dilations. More formally, we have the following:

\begin{proposition}[Dilation formulation of convex influence] \label{prop:inf-dilation}
	Given $K\sse\R^n$ with $0^n\in K$, we have 
	\[\TInf[K] = \lim_{\delta\to0}\frac{\vol\pbra{(1+\delta)K} - \vol(K)}{\delta}.\]
\end{proposition}

Note that if $K\sse\R^n$ is convex with $0^n\in K$, then $\TInf[K]$ is non-negative. \Cref{prop:inf-dilation} is analogous to the well-known Margulis--Russo lemma \cite{Margulis:74,Russo:78} from the analysis of Boolean functions, and a proof of it can be found in Appendix~A of \cite{De2021}. We note that a similar ``dilation formulation'' holds for the convex influence of a single direction $v\in \S^{n-1}$ on $K$, although we will not require it here. 

We will use the following alternative formulation of the total convex influence of a convex set, which was communicated to us by Joe Neeman~\cite{neeman:comm}:

\begin{lemma}[Influence via a surface integral] \label{lem:neeman}
	Given {a measurable set} $K\sse\R^n$ with $0^n\in K$ and a direction $v\in \S^{n-1}$, we have 
	\begin{equation*}
		\TInf[K] = \int_{\partial K} \abra{x, \nu_x}\cdot\phi(x)\,d\sigma(x)
	\end{equation*}
	where $\nu_x$ denotes the unit normal to $\partial K$ at $x$.
\end{lemma}

\begin{proof}
	The proof is a straightforward computation using integration by parts. Recall that via \Cref{def:convex-influence}, we have  
\begin{align}
	\TInf[K] &= \Ex_{\bx\sim N(0, I_n)}\sbra{K(\bx)\pbra{n- \|\bx\|^2}} \nonumber \\
	&= \int_K (n - \|x\|^2)\cdot\phi(x)\,dx \nonumber\\
	&= \int_K \calL\pbra{\frac{\|x\|^2}{2}}\cdot\phi(x)\,dx  \nonumber\\
	\intertext{where for a function $f:\R^n\to\R$ we define $\calL(f) := \Delta f - \abra{\nabla f, \nabla f}$. Integrating by parts then gives}
	\TInf[K] &= \int_{\partial K} \abra{x, \nu_x} \cdot \phi(x)\,d\sigma(x),\nonumber
\end{align}
completing the proof.
\end{proof}

\begin{remark}
	We note that the surface integral formulation of total convex influence can be viewed as analogous to the fact that the total influence of a Boolean function is equal to its average sensitivity (cf. Chapter 2 of~\cite{odonnell-book}). Although we will not require it for our purposes, we note that a similar formulation holds for the convex influence of a direction $v\in \S^{n-1}$ on $K$:
	\begin{equation*} \label{eq:neeman}
		\Inf_v[K] = \int_{\partial K} \abra{x, v}\cdot\abra{\nu_x, v}\cdot\phi(x)\,d\sigma(x)
	\end{equation*}
	where as before $\nu_x$ denotes the unit normal to $\partial K$ at $x$.
\end{remark}

\subsection{Bounds on the Convex Influence of Polytopes}
\label{sec:influence-structural-analogs}


In this subsection we give an upper bound on the total convex influence of an intersection of halfspaces in terms of the number of halfspaces. An analogous statement for CNF formulas over $\zo^n$, {showing that the influence of any $s$-clause CNF is at most $O(\log s)$,} was first given by Boppana~\cite{Boppana1997} using the technique of random restrictions~\cite{Hastad:86}.\footnote{In fact, Boppana~\cite{Boppana1997} obtains an upper bound of $O(\log^{d-1}(s))$ on the total influence of functions computed by depth-$d$ size-$s$ circuits.} The proof of \Cref{prop:convex-boppana} is inspired by the proof of \Cref{thm:nazarov} due to Nazarov; we give a self-contained proof of \Cref{prop:convex-boppana} below (see~\cite{ball-lecture} for a proof sketch of Nazarov's bound).

\begin{proposition}[Convex influence upper bound for intersections of halfspaces.]
\label{prop:convex-boppana}
	Let $K \sse\R^n$ be an intersection of $s \geq 3$ halfspaces that contains the origin. 
	Then 
	\[\TInf[K] < 7\ln s.\]
\end{proposition}

We remark that the upper bound of \Cref{prop:convex-boppana} is best possible up to the hidden constant; the $\ell_\infty$ ball $K=\{x \in \R^n \ : \ \|x\|_\infty \leq r\}$ is an intersection of $2n$ halfspaces, and it is shown in Example~18 of \cite{DNS22} that for a suitable choice of $r$ we have $\TInf[K] = \Theta(\log n).$

\medskip

\begin{proof}[Proof of~\Cref{prop:convex-boppana}]
We let $K = \bigcap_{i=1}^s H_i$ for $s > 1$ where each $H_i$ is a halfspace of the form
	\[H_i := \cbra{x\in\R^n : \abra{x,v_i} \leq \theta_i}\]
	where $v_i \in \S^{n-1}$ for $i\in[s]$.  Note that each $\theta_i \geq 0$ since $K$ contains the origin.
	Using \Cref{lem:neeman}, we have 
	\begin{align}
		\TInf[K] = \int_{\partial K} \abra{x,\nu_x}\cdot\phi(x)\,d\sigma(x)
		= \sum_{i=1}^s \pbra{\int_{{\partial H_i \cap \partial K}} \abra{x,\nu_x}\cdot\phi(x)\,d\sigma(x)} \label{eq:useful}.
	\end{align}
	Now, we observe that (i) $\GSA(K) = \sum_{i=1}^s \pbra{\int_{{\partial H_i \cap \partial K}} \phi(x)\,d\sigma(x)}$; (ii) for $x \in \partial H_i \cap \partial K$ we have $\abra{x,\nu_x}=\theta_i$; and (iii) for each $i \in [s]$, we have $\int_{\partial H_i \cap \partial K} \phi(x)\,d\sigma(x) \leq \int_{\partial H_i } \phi(x)\,d\sigma(x)=\GSA(H_i)$.  Combining these three observations, we get that
	\begin{align}
(\ref{eq:useful})		&~{\leq \max_{i \in [s]:\theta_i\leq \sqrt{2 \ln s}} \theta_i \cdot \GSA(K) + \sum_{\substack{j \in [s] \\ \theta_j > \sqrt{2\ln s}}} \theta_j\cdot\GSA(H_j).} 
		\label{eq:bop-sum}
	\end{align}
	We will control each of the two quantities in \Cref{eq:bop-sum} separately. For the first, we have that 
	\begin{equation} \label{eq:bop-sum-1}
\max_{i \in [s]:\theta_i\leq \sqrt{2 \ln s}} \theta(i) \cdot \GSA(K)  \leq \sqrt{2\ln s}
\pbra{\sqrt{2 \ln s} + 2}< 5\ln s,
	\end{equation}
	where the first inequality is by Nazarov's bound on GSA (\Cref{thm:nazarov}). For the second sum, we have that
	\begin{align}
		\sum_{\substack{j \in [s] \\ \theta_j > \sqrt{2\ln s}}} \theta_j\cdot\GSA(H_j) 
		&= \sum_{\substack{j \in [s] \\ \theta_j > \sqrt{2\ln s}}} \theta_j \cdot e^{-\theta_j^2/2}\nonumber\\
		& \leq s\cdot\max_{j : \theta_i > \sqrt{2\ln s}} \cbra{\theta_j\cdot e^{-\theta_j^2/2}}. \label{eq:gaga} 
		\end{align}
Since $xe^{-x^2/2}$ is a decreasing function for $x\geq 1$, and since $s >1$, it follows that $\theta_j\cdot e^{-\theta_j^2/2}$ is maximized for $\theta_j = \sqrt{2\ln s}$ which lets us conclude that
\begin{align}
(\ref{eq:gaga}) \leq \sqrt{2 \ln s} < 2 \ln s.\label{eq:bop-sum-2}
	\end{align}
	The result follows from \Cref{eq:bop-sum,eq:bop-sum-1,eq:gaga,eq:bop-sum-2}.
\end{proof}

\subsection{Lower Bounds for Approximating Convex Sets with Maximal Influence}
\label{subsec:convex-influence-lower-bound}

We finally establish the following lower bound on approximating symmetric convex sets with close-to-maximal convex influence. {(The specific constants in the theorem below were chosen mostly for concreteness; other constants could have been used instead.)}

\begin{theorem} \label{thm:high-influence-lb}
	Suppose $K\sse\R^n$ is a symmetric convex set with $\Vol(K) = 1/2 \pm o_n(1)$ and $\TInf[K]  \geq 0.1\sqrt{n}$. Then any convex polytope $L$ that $6.25 \times 10^{-7}$-approximates $K$ must have at least
	$2^{5.1 \times 10^{-6}\sqrt{n}}$ halfspaces.
\end{theorem}

\begin{proof}


	
\Cref{thm:high-influence-lb} can be inferred along the lines of the proof of Theorem~A.1 from~\cite{o2007approximation}, but we give a slightly simpler argument below.
		Our proof will make use of the Hermite basis; we refer the reader to~\Cref{appendix:hermite} for a primer on Hermite analysis over the Gaussian measure. 
	Writing $e_i \in \R^n$ for the standard basis vector along the $i^\text{th}$ coordinate direction, we have that 
	\[\Inf_{e_i}[K] = \sqrt{2}\cdot\wt{K}(2e_i).\]
		(This is an immediate consequence of the fact that the degree-2 univariate Hermite polynomial $h_2(x)$ is ${\pbra{1 - x^2} /{\sqrt{2}}}.$)
	Since $K$ is symmetric, it follows from Proposition~9 of~\cite{DNS22} that $\Inf_{e_{\bi}}[K]\geq 0$.
	We will use the following simple claim that relies on this fact:
	
\begin{claim} \label{claim:anindya-paley-zyg}
	Suppose $K$ is as in the statement of~\Cref{thm:high-influence-lb}. Then at least $0.002$-fraction of directions $\{e_1, \ldots, e_n\}$ must have 
	\begin{equation} \label{eq:many-K-big}
	\Inf_{e_i}[K] \geq \frac{0.05}{\sqrt{n}}.
	\end{equation}
\end{claim}

\begin{proof}
	Let $\bi \sim [n]$ uniformly at random. We have 
	\[\Ex_{\bi \sim [n]}[\Inf_{e_{\bi}}[K]] \geq \frac{0.1}{\sqrt{n}}\qquad\text{and}\qquad \Varx_{\bi \sim [n]}\sbra{\Inf_{e_{\bi}}[K]} \leq \frac{1}{n}\]
	where the upper bound on the variance follows from Parseval's formula. Recall Cantelli's inequality, which says that for a non-negative random variable $\bX$ with mean $\mu$ and variance $\sigma$, 
	\[\Pr\sbra{\bX > \mu-\theta\sigma} \geq \frac{\theta^2}{1+\theta^2} \]
	for $\theta \in [0,1]$. (Cantelli's inequality is a straightforward consequence of the Paley--Zygmund inequality.) Since $\Inf_{e_{\bi}}[K]$ is non-negative due to the symmetry of $K$ (Proposition~9 of~\cite{DNS22}),
	it follows that 
	\[\Prx_{\bi\sim [n]}\sbra{\Inf_{e_{\bi}}[K] > \frac{0.05}{\sqrt{n}}} \geq \frac{0.0025}{1.0025} > 0.002,\]
	which completes the proof.
\end{proof}

Without loss of generality, let $\{e_1, \ldots, e_t\}$ be the coordinate directions for which $\Inf_{e_i}[K] \geq 0.05/\sqrt{n}$ where $t\geq 0.002n$, as guaranteed by~\Cref{claim:anindya-paley-zyg}. 

Now, let $L$ be a convex polytope that $6.25 \times 10^{-7}$-approximates $K$, as in the theorem statement.
We observe that  $L$ must contain the origin, since if it did not, by the symmetry of $K$ and the fact that $\Vol(K) = {\frac 1 2} \pm o_n(1)$ we would have that $\dG(K,L) \geq 0.249.$ Hence we can (and will) analyze the influences of various coordinates on $L$.	

Suppose that 
$\TInf[L] \leq 2.5\times 10^{-5}\sqrt{n}$. It then follows that at most $t/2$ of $\{e_1, \ldots, e_t\}$ have 
\[\Inf_{e_i}[L] \geq \frac{{0.05}}{2\sqrt{n}}.\]
In other words, for at least $t/2 \geq 0.001n$ of the coordinates in $\{e_1, \ldots, e_t\}$, we have 
\begin{equation} \label{eq:many-L-small}
\Inf_{e_i}[L] < \frac{{0.05}}{2\sqrt{n}}.
\end{equation}
Call these coordinates ``bad'' coordinates. Using Parseval's formula and \Cref{eq:many-K-big,eq:many-L-small}, we get that
\begin{align*}
\dist(K,L) &= \Ex_{\bx\sim N(0,I_n)}\sbra{(K(\bx) - L(\bx))^2} \\
&= \sum_{\alpha\in\N^n} (\wt{K}(\alpha) - \wt{L}(\alpha))^2 \\
&\geq \sum_{e_i~\text{is bad}} (\wt{K}(2e_i) - \wt{L}(2e_i))^2 \\
&> 
{0.001n \cdot \pbra{\frac{{0.05}}{2\sqrt{n}}}^2} = 6.25 \times 10^{-7},
\end{align*}
{which contradicts the fact that $\dG(K,L) \leq 6.25 \times 10^{-7}$.
Hence we must have $\TInf[L] > 2.5\times 10^{-5}\sqrt{n}$, and by \Cref{prop:convex-boppana} this means that $L$ must be an intersection of at least 
\[2^{\pbra{{\frac {0.25 \times 10^{-4}}{7 \ln 2}}}\sqrt{n}} > 2^{5.1 \times 10^{-6} \sqrt{n}}\] many halfspaces.}
\end{proof}

As a consequence of \Cref{thm:high-influence-lb}, we immediately obtain lower bounds for approximating the $\ell_2$ and $\ell_1$ balls of Gaussian measure $\approx 1/2$:


\begin{example}[$\ell_2$ ball] \label{eg:ball-conv-inf}
	Example~13 of~\cite{DNS22} establishes that 
\[\TInf[B_2] = \Theta(\sqrt{n}).\]
\Cref{lem:neeman} and Gaussian isoperimetry, however, allow us to obtain a lower bound on the constant hidden by the $\Theta(\cdot)$. Since $B_2$ is an intersection of infinitely many halfspaces all of which are at distance $\sqrt{n}$ from the origin, we have 
\[\TInf[B_2] = \sqrt{n}\cdot\GSA(B_2) \geq \sqrt{\frac{n}{2\pi}}(1- o_n(1)) \geq 0.398\sqrt{n}\]
where the first equality is due to \Cref{lem:neeman} and the second inequality follows from Gaussian isoperimetry.
Together with~\Cref{thm:high-influence-lb}, this immediately implies that any $6.25\times 10^{-7}$-approximation to $B(\sqrt{n})$ requires {$2^{\Omega(\sqrt{n})}$} facets. This in turn implies that \Cref{thm:de-nazarov} is tight up to constant factors when $\eps = 6.25\times 10^{-7}$. 
\end{example}

\begin{examples}[Cross-polytope] \label{eg:cross-polytope-inf}
	Let $B_1$ be the cross-polytope, i.e. the set
	\[B_1 := \cbra{x\in\R^n : \|x\|_1 \leq\sqrt{\frac{2}{\pi}}n}.\]
	Note that $B_1$ is an intersection of $2^n$ halfspaces, and recall  from~\Cref{subsec:cramer-bounds} that $\vol(B_1) = 1/2 \pm o_n(1)$. 
	By the Gaussian isoperimetric inequality, we have that $\GSA(B_1) \geq 0.398$, and so using \Cref{lem:neeman} we have
	\[\TInf[B_1] = \sqrt{\frac{2}{\pi}n}\cdot\GSA(B_1) \geq 0.1\sqrt{n}.\]
	Together with~\Cref{thm:high-influence-lb}, this immediately implies that any $6.25\times 10^{-7}$-approximation to $B_1$ requires {$2^{\Omega(\sqrt{n})}$} facets. (Recall from~\Cref{thm:ell-p-approx} that $B_1$ can be $0.01$-approximated using $2^{\Theta(n^{3/4})}$ halfspaces.)
\end{examples}


\section*{Acknowledgements}

We thank Elisabeth Werner for answering questions about~\cite{LSW06}. 
We thank Ryan O'Donnell and Ashwin Padaki for helpful discussions. 
A.D. is supported by NSF grants CCF-1910534 and CCF-2045128. 
S.N. is supported by NSF grants IIS-1838154, CCF-2106429, CCF-2211238, CCF-1763970, and CCF-2107187. 
R.A.S. is supported by NSF grants IIS-1838154, CCF-2106429, and CCF-2211238. 
This work was partially completed while the authors were visiting the Simons Institute for the Theory of Computing as a part of the program on ``Analysis and TCS: New Frontiers.''

\begin{flushleft}
\bibliographystyle{alpha}
\bibliography{allrefs}
\end{flushleft}

\appendix

\section{Hermite Analysis over $N(0,I_n)$}
\label{appendix:hermite}

Our notation and terminology follow Chapter~11 of ~\cite{odonnell-book}. We say that an $n$-dimensional \emph{multi-index} is a tuple $\alpha \in \N^n$, and we define 
\[
|\alpha| := \sum_{i=1}^n \alpha_i.
\]

For $n \in \N_{>0}$, we write $L^2(\R^n)$ to denote the space of functions $f: \R^n \to \R$ that have finite second moment under the Gaussian distribution, i.e. $f\in L^2(\R^n)$ if 
\[
\|f\|^2 = \Ex_{\bz \sim N(0,I_n)} \left[f(\bz)^2\right]^{1/2} < \infty.
\]
We view $L^2(\R^n)$ as an inner product space with 
\[\la f, g \ra := \Ex_{\bz \sim  N(0,I_n)}[f(\bz)g(\bz)]\]
We recall the Hermite basis for $L^2(\R, \gamma)$:

\begin{definition}[Hermite basis]
	The \emph{Hermite polynomials} $(h_j)_{j\in\N}$ are the univariate polynomials defined as
	$$h_j(x) = \frac{(-1)^j}{\sqrt{j!}} \exp\left(\frac{x^2}{2}\right) \cdot \frac{d^j}{d x^j} \exp\left(-\frac{x^2}{2}\right).$$
\end{definition}

The following fact is standard:

\begin{fact} [Proposition~11.33 of~\cite{odonnell-book}] \label{fact:hermite-orthonormality}
	The Hermite polynomials $(h_j)_{j\in\N}$ form a complete, orthonormal basis for $L^2(\R, \gamma)$. For $n > 1$ the collection of $n$-variate polynomials given by $(h_\alpha)_{\alpha\in\N^n}$ where
	$$h_\alpha(x) := \prod_{i=1}^n h_{\alpha_i}(x)$$
	forms a complete, orthonormal basis for $L^2(\R^n, \gamma)$. 
\end{fact}

Given a function $f \in L^2(\R^n, \gamma)$ and $\alpha \in \N^n$, we define its \emph{Hermite coefficient on} $\alpha$ as $\widetilde{f}(\alpha) = \la f, h_\alpha \ra$. It follows that $f:\R^n\to\R$ can be uniquely expressed as 
\[f = \sum_{\alpha\in\N^n} \widetilde{f}(\alpha)h_\alpha\]
with the equality holding in $L^2(\R^n, \gamma)$; we will refer to this expansion as the \emph{Hermite expansion} of $f$. One can check that Parseval's and Plancharel's identities hold in this setting:
\[\abra{f,f} = \sum_{\alpha\in \N^n}\widetilde{f}(\alpha)^2 \qquad\text{and}\qquad \abra{f,g} = \sum_{\alpha\in \N^n}\widetilde{f}(\alpha)\widetilde{g}(\alpha).\]

\end{document}